\documentclass[a4paper, 11pt]{article}
\usepackage{fullpage}
\usepackage{graphicx}
\usepackage{booktabs}
\usepackage{amsmath,amssymb,amsthm,mathtools,mathrsfs}
\usepackage[ruled,vlined,linesnumbered]{algorithm2e}
\usepackage{xspace}

\usepackage{tikz,pgfplots,pgfplotstable}
\usetikzlibrary{positioning,intersections,calc,arrows}
\makeatletter
\DeclareRobustCommand{\rvdots}{%
  \vbox{
    \baselineskip4\p@\lineskiplimit\z@
    \kern-\p@
    \hbox{.}\hbox{.}\hbox{.}
  }}
\makeatother
\newcommand{\ot}{\leftarrow}
\newcommand{\qopt}{{\rm QOPT}}  
\newcommand{\sfda}{{\rm SFDA}}  
\newcommand{\spe}{{\rm SPE}}    
\newcommand{\SPEM}{\textup{\texttt{SPEM}}\xspace}
\newcommand{\ACCEPT}{\texttt{ACCEPT}\xspace}
\newcommand{\REJECT}{\texttt{REJECT}\xspace}
\newcommand{\OR}{\texttt{OR}\xspace}
\newcommand{\AND}{\texttt{AND}\xspace}
\newcommand{\NOT}{\texttt{NOT}\xspace}
\newcommand{\TRUE}{\textsc{True}\xspace}
\newcommand{\FALSE}{\textsc{False}\xspace}
\newcommand{\BRANCHING}{\texttt{BRANCHING}}

\newcommand{\argmax}{\mathop{\rm arg\,max}}
\newcommand{\QSAT}{\texttt{QUANTIFIED 3SAT}\xspace}

\newtheorem{theorem}{Theorem}[section]
\newtheorem{lemma}[theorem]{Lemma}
\newtheorem{proposition}[theorem]{Proposition}
\newtheorem{property}[theorem]{Property}
\newtheorem{claim}[theorem]{Claim}
\newtheorem{observation}[theorem]{Observation}
\newtheorem{corollary}[theorem]{Corollary}
\newtheorem{conjecture}[theorem]{Conjecture}
\theoremstyle{definition}
\newtheorem{definition}[theorem]{Definition}
\newtheorem{example}[theorem]{Example}
\newtheorem{problem}[theorem]{Problem}
\theoremstyle{remark}
\newtheorem{remark}[theorem]{Remark}

\begin{document}

\title{Subgame Perfect Equilibria of Sequential Matching Gamest\thanks{A preliminary version~\cite{KYY2018} appeared in EC'18.}}
\author{Yasushi Kawase\thanks{Tokyo Institute of Technology and RIKEN AIP Center. Email: \texttt{kawase.y.ab@m.titech.ac.jp}}
\and Yutaro Yamaguchi\thanks{Osaka University and RIKEN AIP Center. Email: \texttt{yutaro\_yamaguchi@ist.osaka-u.ac.jp}}
\and Yu Yokoi\thanks{National Institute of Informatics. Email: \texttt{yokoi@nii.ac.jp}}}
\date{\empty}
\maketitle

\begin{abstract}
We study a decentralized matching market in which firms sequentially 
make offers to potential workers. For each offer, the worker can choose ``accept'' or ``reject,'' but the decision is irrevocable. 
The acceptance of an offer guarantees her job at the firm,
but it may also eliminate chances of better offers from other firms in the future. 
We formulate this market as a perfect-information extensive-form game played by the workers. 
Each instance of this game has a unique subgame perfect equilibrium (SPE), 
which does not necessarily lead to a stable matching and has some perplexing properties.

We show a dichotomy result that characterizes the complexity of computing the SPE\@.
The computation is tractable if each firm makes offers to at most two workers or each worker receives offers from at most two firms.
In contrast, it is PSPACE-hard even if both firms and workers are related to at most three offers.
We also study engineering aspects of this matching market.
It is shown that, for any preference profile, we can design an
offering schedule of firms so that the worker-optimal stable matching 
is realized in the SPE.
\end{abstract}


\section{Introduction}
Imagine a decentralized job market consisting of firms and workers. 
Each firm has one position to fill and has a preference ordering over potential workers. 
Additionally, each worker has a preference on positions. 
To fill a position, each firm first makes an offer to its favorite worker, 
and if the offer is rejected, then the next offer is made to the second best worker, and so on. 
The offers of the firms are not synchronized with each other. 
That is, firms act as if they simulate an asynchronous version of 
the firm-oriented deferred acceptance algorithm~\cite{MW70,MW71}, 
which is known to find the position-optimal \emph{stable matching}
just like the original synchronous version \cite{GS62}. 
In contrast to these algorithms, in which each worker can keep a tentative contract and 
decline it when she gets a better offer, 
the current market does not allow tentative contracts. 
Once a worker receives an offer from some firm, 
she must decide immediately whether to accept or not and cannot change the decision later. 
Thus, the acceptance of an offer guarantees her job at that position,
but it may also eliminate chances of better offers from other firms in the future. 
It is assumed that the offers are made in accordance with a prescribed schedule, i.e.,
there is a linear order on the set of all possible offers. 
All the workers know this order, but whether a worker will get each scheduled offer 
from each firm or not depends on the actions of other workers.
Thus, this market has the sequential structure of decision problems encountered by the strategic workers.

\smallskip

Our model deals with any order of offers that is consistent 
with every firm's preference. 
In particular, we call a market {\em position-based} 
if all the offers by the same firm are successively placed in the order,
i.e., there is a linear order on the set of firms according to which each firm makes all its offers.
This case represents, for instance, an academic job market in which different positions 
have different hiring seasons. Job-seeking researchers know these seasons, and they can guess 
each position's preference ordering over the candidate researchers 
and also each researcher's preference ordering over the positions.
Once a researcher accepts an offer, she cannot decline it because unilaterally rescinding a contract 
damages her reputation and adversely affects her future career.
The position-based case can also be interpreted as a job market
in which institutes have public invitations in different seasons and
researchers strategically decide whether to apply or not for each invitation.
\smallskip

We formulate this market as a perfect-information extensive-form game among the workers,
which we call the {\em sequential matching game} (a formal definition will be given in Section~\ref{sec:model}).
Each round of the game corresponds to an offer from a firm. 
The offers are made in some fixed order. 
In each round, the worker who receives the offer is the player who takes an action, where the possible actions are \ACCEPT and \REJECT. If the worker chooses \ACCEPT, then the firm and the worker are matched and leave the market, and we move to a subgame in which they are removed. If the worker chooses \REJECT, then they stay in the market, and we move to a subgame in which that worker is eliminated from the preference list of the firm. The game ends if there is no firm with a nonempty list. Each strategy profile uniquely defines a matching, or an assignment obtained at the end of the game. Each worker's preference over outcome matchings depends only on her own assignment. 

This paper investigates a subgame perfect equilibrium (SPE) of this sequential matching game, i.e., an action profile that represents the best actions of the workers in all rounds under the assumption that all the other workers will take their best actions in the future. For every offer, different actions result in different outcomes for the worker: by rejecting an offer, it is impossible for her to obtain the same assignment afterwards. Hence, the best action is defined uniquely at each round by backward induction. Therefore, any instance of the sequential matching game has a unique SPE. To clarify the setting, we provide \mbox{a small example here}.

\begin{example}\label{ex:intro}
There are three research institutes $p_1$, $p_2$, and $p_3$ each of which has one position to fill.
Three job-seeking researchers $q_1$, $q_2$, and $q_3$ are awaiting offers. 
The institutes have preferences on acceptable candidates, and 
researchers have \mbox{preferences on possible institutes as follows:}
\begin{align*}
  &p_1:~ q_1                                 && q_1:~ p_2 \succ_{q_1} p_1\\
  &p_2:~ q_2 \succ_{p_2} q_1 \succ_{p_2} q_3 && q_2:~ p_3 \succ_{q_2} p_2\\
  &p_3:~ q_3 \succ_{p_3} q_2                 && q_3:~ p_2 \succ_{q_3} p_3.
\end{align*}
Each researcher prefers being matched (with an acceptable institute) to being unmatched. 
The institute $p_{1}$ starts scouting first, 
and then $p_{2}$ and $p_{3}$ follow in order. Thus, the following order is defined on the offers.
\[\text{The offering order}: (p_{1}, q_{1}),\, (p_{2}, q_{2}),\, (p_{2}, q_{1}),\, (p_{2}, q_{3}),\, (p_{3}, q_{3}),\, (p_{3}, q_{2}),\]
where offers related to previously matched institutes or researchers are skipped. 

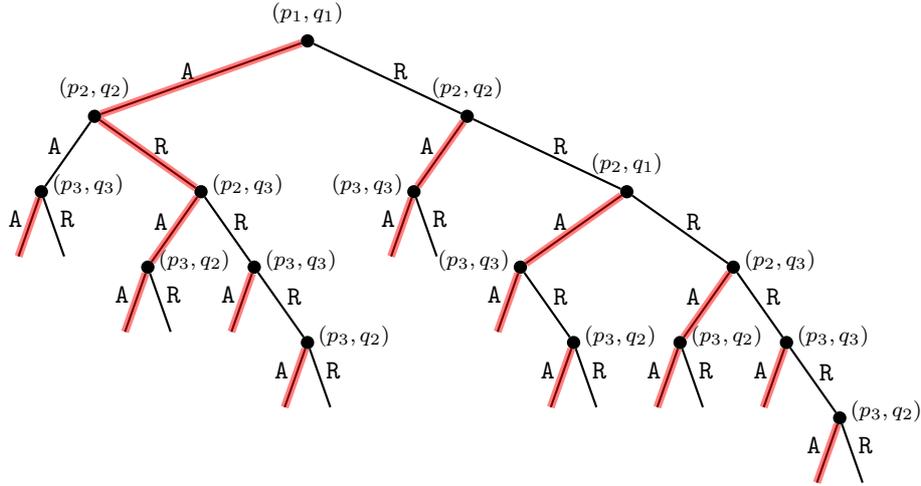
\begin{figure}[htbp]
\centering
\scalebox{1}{
\begin{tikzpicture}[thick,xscale=0.7,yscale=1.0,
    p/.style={circle,fill=black,inner sep=0pt,outer sep=0pt,minimum size=5pt}]
  \node[p,label={[]{\scriptsize$(p_1,q_1)$}}] at (6.5,6)  (1-1) {}; 
  \node[p,label={[]{\scriptsize$(p_2,q_2)$}}] at (2.5,5)  (2-1) {}; 
  \node[p,label={[]{\scriptsize$(p_2,q_2)$}}] at (9.5,5)  (2-2) {}; 
  \node[p,label={[]{\scriptsize$(p_2,q_1)$}}] at (12.5,4)  (3-4) {}; 
  \foreach \x/\i/\j in {4.5/2/4,14.5/5/3}
    \node[p,label={[right]{\scriptsize$(p_2,q_3)$}}] at (\x,\j) (4-\i) {};
  \foreach \x/\i/\j in {1.5/1/4,5.5/3/3,15.5/7/2}
    \node[p,label={[right]{\scriptsize$(p_3,q_3)$}}] at (\x,\j) (5-\i) {};
  \foreach \x/\i/\j in {8.5/4/4,10.5/5/3}
    \node[p,label={[left]{\scriptsize$(p_3,q_3)$}}] at (\x,\j) (5-\i) {};
  \foreach \x/\i/\j in {3.5/3/3,6.5/5/2,11.5/9/2,13.5/10/2,16.5/12/1}
    \node[p,label={[right]{\scriptsize$(p_3,q_2)$}}] at (\x,\j) (6-\i) {};

  \foreach \i/\j in {1/3,2/3,3/2,4/2,5/2,6/1,7/1,8/3,9/3,10/2,11/1,12/1,13/1,14/1,15/1,16/0,17/0} \node[] at (\i,\j) (7-\i) {};

  \foreach \u/\v in {1-1/2-1,2-1/5-1,2-2/5-4,3-4/5-5,4-2/6-3,4-5/6-10,5-1/7-1,5-3/7-5,5-4/7-8,5-5/7-10,5-7/7-15,6-3/7-3,6-5/7-6,6-9/7-11,6-10/7-13,6-12/7-16}
    \draw (\u) -- (\v) node [pos=.5,font=\small,xshift=-5pt,yshift=3pt] {\texttt{A}};
  \foreach \u/\v in {1-1/2-2,2-1/4-2,2-2/3-4,3-4/4-5,4-2/5-3,4-5/5-7,5-1/7-2,5-3/6-5,5-4/7-9,5-5/6-9,5-7/6-12,6-3/7-4,6-5/7-7,6-9/7-12,6-10/7-14,6-12/7-17}
    \draw (\u) -- (\v) node [pos=.5,font=\small,xshift=5pt,yshift=3pt] {\texttt{R}};
  \foreach \u/\v in {1-1/2-1,2-1/4-2,2-2/5-4,3-4/5-5,4-2/6-3,4-5/6-10,5-1/7-1,5-3/7-5,5-4/7-8,5-5/7-10,5-7/7-15,6-3/7-3,6-5/7-6,6-9/7-11,6-10/7-13,6-12/7-16}
    \draw[line width=3pt,red,opacity=.5] (\u) -- (\v);  
\end{tikzpicture}}
\vspace*{-4mm}
\caption{The tree representation of the game. The bold red edges indicate the SPE\@.}\label{fig:tree}
\end{figure}
When the first offer $(p_{1}, q_{1})$ is made, $q_{1}$ has two possible actions. 
If she selects \ACCEPT, then she is assigned to $p_{1}$ in the outcome matching.
If she selects \REJECT, she has to anticipate what happens after the rejection (see Fig.~\ref{fig:tree}).  
After $q_{1}$'s rejection, $p_{1}$ leaves the market because it has no other candidate.
Then $p_{2}$ makes an offer to $q_{2}$, and $q_{2}$ may select \ACCEPT or \REJECT.
Under the assumption that all workers similarly take the best actions by 
anticipating other workers' future actions, it is concluded that 
$q_{1}$ would be unmatched in the outcome if she rejects $(p_{1}, q_{1})$.
Thus, \ACCEPT is her best action in the first round.
In this manner, we can define the best action for each round.
Fig.~\ref{fig:tree} shows all possible rounds, and 
red edges represent the best actions. 
The SPE of this instance then results in an assignment 
that matches pairs $\{(p_{1}, q_{1}),\, (p_{2}, q_{3}),\, (p_{3}, q_{2})\}$.
Note that this does not coincide with a unique stable matching
$\{(p_{1}, q_{1}),\, (p_{2}, q_{2}),\, (p_{3}, q_{3})\}$ under 
this preference profile.
\end{example}

Recently, a sequential game on a matching instance has been studied actively under various settings~\cite{RV1991,APR1998,AR2000,AR2005,Pais2008,SW2008,HW2011}. 
In many of those settings, it has been shown that the SPE leads to a stable matching.
However, such a characterization of the SPE does not hold in our setting. 
Indeed, the outcome matching of the above example is unstable. 
Moreover, we will see in Section~\ref{subsec:bo} that 
the outcome matching in our setting may violate not only the standard stability but also weaker stabilities such as
\emph{vNM-stability}~\cite{Ehlers2007} and \emph{essential stability}~\cite{KT2016}.
Besides, unlike many other models, the outcome matching of the SPE changes drastically depending on the order of offers in our model. 
We also see in Section~\ref{subsec:bo} that not only the matching itself but also the set of matched firms and workers may change, i.e., 
there is no ``rural hospital theorem''-like property.
These distinctive features indicate the difficulty in capturing SPEs in our model.

This paper studies SPEs of the sequential matching game from two different perspectives: equilibrium computation and indirect market regulation.
In the first part \mbox{(Sections~\ref{sec:restrict} and \ref{sec:hardness})}, we reveal the complexity of computing the outcome matching of the SPE for a given preference profile and a given offering order.
That is, we analyze how hard it is to predict the market outcome assuming that all the workers follow their best strategies. 
In the second part (Sections~\ref{sec:imp_qopt} and \ref{sec:imp}), we see our model from an engineering perspective. 
For a given preference profile, we consider designing an offering order 
so that the outcome matching of the SPE admits a certain socially desirable property, e.g., stability. 
Here we describe the details of each part.

\subsection*{Equilibrium Computation}
In the first part, we consider the complexity of computing an SPE in our model. Note that representing an SPE itself obviously requires exponential time since the tree representation has exponential size. 
Therefore, it is more reasonable to consider the following decision problem: given an instance and a history (a possible sequence of actions), decide whether the next player selects \ACCEPT or not in the SPE\@. 
Note that each subgame of a given instance is again an instance of the sequential matching game. 
Then, we would rather consider \mbox{the following equivalent problem.} 

\begin{problem}[\SPEM]\label{prob:spem}
Given an instance of the sequential matching game, decide whether the first offered worker selects \ACCEPT or not in the SPE\@. 
\end{problem}

We call this decision problem \SPEM. Note that the outcome matching of the SPE, which we call {\em the SPE matching}, can be obtained by solving \SPEM repeatedly. Conversely, \SPEM is solved if we can compute the SPE matching. We classify subclasses of \SPEM by the maximum length of the preference lists of firms and of workers. For positive integers $s$ and $t$, the problem $(s,t)$-\SPEM is the restriction of \SPEM in which the preference list of each firm and that of each worker have at most $s$ and $t$ entries, respectively. 
Moreover, $(s,\infty)$-\SPEM and $(\infty,t)$-\SPEM denote the restrictions in which only one side has a limitation on the list length. 
We provide the following dichotomy theorem that completely characterizes
the computational complexity of $(s,t)$-\SPEM for all $s$ and $t$.

\begin{theorem}\label{thm:dichotomy}
The problem $(s,t)$-\SPEM can be solved in polynomial time if
$s\leq 2$ or $t\leq 2$.
Otherwise, it is PSPACE-complete.
\end{theorem}
For the above theorem, we show the tractability of  
$(2,\infty)$-\SPEM and $(\infty,2)$-\SPEM in Section~\ref{sec:restrict} and
the PSPACE-completeness of $(3,3)$-\SPEM in Section~\ref{sec:hardness}.

We will first show that both $(2,\infty)$-\SPEM and $(\infty,2)$-\SPEM are in the complexity class P
by providing an efficient algorithm to compute the SPE matching. 
As we have observed in the above example, the SPE matching of an $(\infty,2)$-\SPEM instance 
is not necessarily stable. However, fortunately, we can compute it by an algorithm based on the deferred acceptance algorithm, which we call the {\em sequentially fixing deferred acceptance algorithm} (SFDA). 
This algorithm repeatedly computes the worker-optimal stable matching and 
fixes the matched pair that appears first in the offering order. 
The correctness of SFDA implies that, 
whenever a worker receives an offer from some firm, 
her best action is \ACCEPT if and only if she is assigned to that firm in the worker-optimal stable matching 
of the ``current'' subgame, rather than of the original instance. 
SFDA also works for $(2,\infty)$-\SPEM, i.e., it outputs the SPE matching. 
Furthermore, we can show that for $(2,\infty)$-\SPEM,
the output of SFDA coincides with the worker-optimal stable matching
independently of the offering order.
Hence, for $(2,\infty)$-\SPEM, the SPE matching is exactly the worker-optimal stable matching.

In contrast to the above tractable cases, the problem $(3,3)$-\SPEM is far from tractable.
Actually, we show that it is PSPACE-hard, which means that there exists no polynomial-time algorithm to solve it unless P${}={}$PSPACE\@.
Our proof is based on a reduction from \QSAT, which is a PSPACE-complete problem.
Even if  each preference list is of length at most three,
sequential matching games have a messy structure,
which enables us to construct gadgets simulating logic gates such as NOT, OR, and AND.

\subsection*{Order Design for Socially Desirable Matchings}
As mentioned above, in general the SPE matching in our model can be far from stable.
Moreover, the SPE matching does not attain other criteria of social welfare such as Pareto-efficiency and first-choice maximality. 
On the other hand, as we will show in Example \ref{ex:simple} (Section \ref{sec:model}), the SPE matching may differ according to offering orders. 
These facts give rise to the following question.
``Can we lead the market into socially desirable status by designing an appropriate offering schedule?'' 
This is the issue investigated in the second part.

Let put ourselves in the position of a central authority who can partially control the market by regulating an offering schedule of firms.
Suppose that we cannot regulate the behaviors of the workers but we can elicit their true preferences.
We first observe preferences of firms and workers and then design an offering schedule, i.e., an offering order, which should be consistent to each firm's preference. 
The market process proceeds on the instance that consists of the observed preference profile and the designed offering order. 
Our purpose is to design an offering order such that a desirable matching is obtained in the SPE, i.e., 
we aim to set up the market so that the workers' strategic behaviors result in larger social welfare.

Specifically, we will focus on the following criteria: stability, Pareto-efficiency, and first-choice maximality. 
Arguably, stability is the most important criterion in two-sided matching markets, which represents a kind of fairness for all firms and workers.
A stable matching is called worker-optimal (resp., firm-optimal) if it is most preferred by workers (resp., firms) among all the stable matchings. 
Worker-side Pareto-efficiency is also a well-discussed criterion when much emphasis is put on workers' welfare \cite{AS2003,APR2009,Kesten2010}.
In some applications, worker-side first-choice maximality attracts a significant interest, 
where this criterion requires that the number of workers assigned to their first choices is maximized \cite{DMS2018}.
Firm-side Pareto-efficiency and firm-side first-choice maximality are symmetrically defined.

Unfortunately, in Section~\ref{sec:imp}, it turns out that most of these criteria cannot be realized in the SPE in general.
For each of them except worker-optimal stability, we can provide a preference profile such that no offering order yields the SPE matching satisfying the desired property.
These impossibility results imply the hardness of controlling the behaviors of workers by indirect regulation, i.e., designing firms' offering schedule. 

Surprisingly, however, we discover that the worker-optimal stable matching can be achieved in the SPE for any preference profile. 
The following theorem is shown in Section~\ref{sec:imp_qopt}.
\begin{theorem}\label{thm:ordering}
Given a preference profile, one can construct an offering order so that the SPE matching is the worker-optimal stable matching.
\end{theorem}

We prove this theorem by giving an efficient algorithm for constructing a desired offering order in the statement,
which is two-phased and intuitively as follows.
We first make the firms offer to all the workers preferable to their partners in the worker-optimal stable matching.
The rest of the offers are position-based, which is defined according to the order of being matched 
in the worker-oriented deferred acceptance algorithm.

It is worth remarking that Theorem~\ref{thm:ordering} has no assumption on the lengths of \mbox{preference lists}. 
This positive result somewhat compensates for the strong negative results in the first part. 
While SPE matchings are hard to compute and have perplexing properties in general, 
a regulation of firms' offering schedule can lead them to be stable for any preference profile.

\subsection*{Related Work}
The complexity of finding a Nash equilibrium in normal form games has been well studied in the context of algorithmic game theory~\cite{NRTV07}.
However, little work has been done on a subgame perfect equilibrium in extensive-form games. 
It is known that computing an SPE is PSPACE-complete for the sequential versions of unrelated machine scheduling, congestion games~\cite{LST2012}, and cost sharing games~\cite{AHK2016}.

In the context of economics, the outcomes generated by decentralized matching markets have been studied under various conditions~\cite{RV1991,APR1998,AR2000,AR2005,Pais2008,SW2008,HW2011,KB2019}.
Most existing studies have analyzed when a decentralized market will yield a stable matching.
Haeringer and Wooders \cite{HW2011} considered a similar setting to ours:
in each stage, each firm offers a position to a worker, and then each worker irrevocably chooses
whether to accept one of the received offers or to reject all of them.
Roughly speaking, their game is based on a synchronous version of the deferred acceptance algorithm, while our game is based on an asynchronous version.
They claimed that, when the workers sequentially decide their actions in each stage, every SPE results in the worker-optimal stable matching \cite[Theorem 2]{HW2011},
but this claim has been proven not true in the present study (see Remark~\ref{rem:HW2011} for details).

\section{Preliminaries}\label{sec:prel}

\subsection{Model}\label{sec:model}
In this section, we give a formal definition of the sequential matching game.
The game is played by a sequential process of acceptance/rejection for offers
in a matching instance, and is formulated as a perfect-information extensive-form game.

A \emph{matching instance} is a tuple $I = (P,Q,E,{\succ})$, where each component is defined as follows. 
There are a finite set of \emph{firms} $P=\{p_1,\ldots,p_n\}$ and a finite set of \emph{workers} $Q=\{q_1,\ldots,q_m\}$.
Each firm has one position, and we often identify each firm with its position.
The set of \emph{acceptable pairs} is denoted by $E\subseteq P\times Q$.
For $e=(p,q)\in E$, let us define $\partial_P(e)\coloneq p$ and $\partial_Q(e)\coloneq q$.
For each firm $p\in P$ and each worker $q \in Q$,
the acceptable partner sets are denoted by $\Gamma_I(p)\coloneq\{q'\in Q\mid (p,q')\in E\}$ and by $\Gamma_I(q)\coloneq\{p'\in P\mid (p',q)\in E\}$, respectively.
The set of acceptable pairs that contain $p\in P$ and $q\in Q$ are denoted by $\delta_I(p)\coloneq\{e \in E \mid \partial_P(e) = p\}$ and $\delta_I(q)\coloneq\{e\in E \mid \partial_Q(e) = q\}$, respectively.
Each $r\in P\cup Q$ has a strict ordinal preference $\succ_r$ over $\Gamma_I(r)$,
and ${\succ}$ denotes the profile $({\succ_r})_{r \in P \cup Q}$.
We sometimes write $(p, q) \succ_p (p, q')$ to denote $q \succ_p q'$ for $p \in P$.

A mapping $\mu\colon P \cup Q \to P \cup Q$ is called a \emph{matching} in $I = (P, Q, E, {\succ})$ if
$\mu(r) \in \Gamma_I(r) \cup \{r\}$ and $\mu(\mu(r)) = r$ for every $r \in P \cup Q$.
A firm or worker $r \in P \cup Q$ is said to be \emph{matched $($with $\mu(r)$$)$} if $\mu(r) \neq r$,
and to be \emph{unmatched} if $\mu(r) = r$.
For each $r \in P \cup Q$ and two matchings $\mu_1$ and $\mu_2$,
we write $\mu_1 \succ_r \mu_2$ if $\mu_1(r) \succ_r \mu_2(r)$,
where we suppose that $r$ is virtually added to the bottom of the preference $\succ_r$ over $\Gamma_I(r)$
(i.e., each $r$ prefers to be matched to an acceptable partner rather than to be unmatched).
A matching $\mu$ in $I$ is often identified with the corresponding set of acceptable pairs $\{(p, \mu(p)) \mid p \in P,~\mu(p) \neq p\} = \{(\mu(q), q) \mid q \in Q,~\mu(q) \neq q\}$.
Let $\mathcal{M}_I$ denote the set of all matchings in $I$.

A \emph{sequential matching game}\footnote{We also use the term ``sequential matching game'' to refer the set of all games defined in this way.} is defined by a pair of a matching instance $I=(P,Q,E,\succ)$ and an \emph{offering order} $\sigma$ over $E$,
which is a bijection from $\{1, \dots, |E|\}$ to $E$.
We assume that $\sigma$ is consistent with $({\succ_p})_{p\in P}$, i.e., if $\sigma(i),\sigma(j)\in \delta_I(p)$ and $\sigma(i)\succ_p \sigma(j)$ then $i<j$.
We denote the set of all consistent orders by $\Sigma_I$.
We say that $\sigma\in\Sigma_I$ is induced by a \emph{position order} $\pi \colon \{1, \dots, |P|\} \to P$ (which is also a bijection) if
$\pi^{-1}(p)<\pi^{-1}(p')$ implies $\sigma^{-1}(e)<\sigma^{-1}(e')$ for all $p,p'\in P$, $e\in\delta_I(p)$, and $e'\in\delta_I(p')$.
In addition, we say that $\sigma$ (or the game $(I, \sigma)$) is \emph{position-based} if $\sigma$ is induced by some position order.

We introduce two fundamental operations on a sequential matching game.
For a set $X$ and an element $e \in X$,
we simply denote $X \setminus \{e\}$ by $X - e$.

\begin{definition}[Deletion]
Let $e \in E$ be an acceptable pair in a matching instance $I=(P,Q,E,{\succ})$.
The {\em deletion} of $e$ from $I$ is defined as 
$I-e \coloneq (P,Q,E-e, {\succ'})$,
where ${\succ'_r}$ is the preference over $\Gamma_{I-e}(r)$ that is consistent with ${\succ_r}$ for each $r \in P \cup Q$,
i.e., $s \succ'_r t$ if and only if $s \succ_r t$ for every $s, t \in \Gamma_{I-e}(r)$.
For an order $\sigma$ over $E$, 
the {\em deletion} of $e$ from $\sigma$ is an order $\sigma -e$ over $E-e$ that is consistent with $\sigma$,
i.e., $(\sigma-e)^{-1}(e_1) < (\sigma-e)^{-1}(e_2)$ if and only if $\sigma^{-1}(e_1) < \sigma^{-1}(e_2)$ for every $e_1, e_2 \in E - e$.
\end{definition}

\begin{definition}[Contraction]
Let $e = (p, q) \in E$ be an acceptable pair in a matching instance $I=(P,Q,E,{\succ})$.
The {\em contraction} of $I$ by $e$ is defined as 
$I/e\coloneq(P-p,Q-q,E/e,{\succ'})$,
where $E/e \coloneq E \setminus (\delta_I(p) \cup \delta_I(q))$,
and ${\succ'_r}$ is the preference over $\Gamma_{I/e}(r)$ that is consistent with ${\succ_r}$ for each $r \in (P - p) \cup (Q - q)$.
For an order $\sigma$ over $E$, 
the {\em contraction} of $\sigma$ by $e$ is an
order $\sigma/e$ over $E/e$
that is consistent to $\sigma$.
\end{definition}

\begin{observation}\label{obs:commutative}
Deletion and contraction are commutative.
That is, for any two disjoint pairs $e_1, e_2 \in E$,
we have $(I/e_1)-e_2 = (I -e_2)/e_1$ and $(\sigma/e_1)-e_2 = (\sigma -e_2)/e_1$.
\end{observation}

Let $e = (p, q) \in E$.
For a matching $\mu'$ in $I/e$,
we denote by $\mu' + e$ the matching $\mu$ in $I$ such that $\mu(p) = q$, $\mu(q) = p$,
and $\mu(r) = \mu'(r)$ for every $r \in (P - p) \cup (Q - q)$.
Conversely, for a matching $\mu$ in $I$ with $\mu(p) = q$ and $\mu(q) = p$,
we denote by $\mu - e$ the matching $\mu'$ in $I/e$ such that $\mu'(r) = \mu(r)$ for every $r \in (P - p) \cup (Q - q)$.
Note that these can be regarded as ordinary set operations: an element $e$ is added to a set $\mu'$ and removed from a set $\mu$, respectively.

The game $(I,\sigma)$ is recursively defined as follows.
Let $\sigma(1) = e = (p,q)$.
In the first round, $p$ offers to $q$, and $q$ chooses \ACCEPT or \REJECT for the offer.
If $q$ chooses \ACCEPT, then $p$ and $q$ are matched irrevocably, and then $(I/e,\sigma/e)$ is played
(i.e., the outcome is $\mu' + e$ for some matching $\mu'$ in $I/e$).
If $q$ chooses \REJECT, then $p$ and $q$ are unmatched irrevocably, and then $(I-e,\sigma-e)$ is played (i.e., the outcome is some matching in $I -e$).
As we have already seen in Fig.~\ref{fig:tree},
this process can be represented by a rooted tree,
where the root corresponds to the first round of the game $(I, \sigma)$, and
each node (except the leaves) corresponds to a game $(I',\sigma')$ and has two children corresponding to the two subgames $(I'/\sigma'(1), \sigma'/\sigma'(1))$ and $(I'-\sigma'(1), \sigma'-\sigma'(1))$.
We call this tree the \emph{tree representation} of the game $(I, \sigma)$.

We consider a \emph{subgame perfect equilibrium} (SPE) of the sequential matching game.
A strategy profile is a subgame perfect equilibrium if it is a Nash equilibrium of every subgame of the original game.
In every round of the game, the two actions \ACCEPT and \REJECT for an offer $(p,q)$
result in different outcome matchings, say $\mu_1$ and $\mu_2$, respectively.
One of them is preferred to the other by the worker $q$,
because $\mu_1(q) = p\neq \mu_2(q)$. 
Hence, the optimal strategy is uniquely defined in each subgame by backward induction,
and every game admits a unique SPE\@.
The outcome matching according to the SPE of a game $(I, \sigma)$ is called \emph{the SPE matching} of $(I, \sigma)$,
which is denoted by $\spe(I, \sigma)$.
The next property immediately follows from the above definitions.

\begin{proposition}\label{prop:operations_spe}
For $e=(p,q)=\sigma(1)$ in any sequential matching game $(I, \sigma)$, we have the following properties.
\begin{description}
\item[(a)] If $\spe(I-e,\sigma-e)(q)\succ_{q} p$, then 
$\spe(I,\sigma)=\spe(I-e,\sigma-e)$.
\item[(b)] If $\spe(I-e,\sigma-e)(q) \prec_{q} p$, 
then $\spe(I,\sigma)=\spe(I/e,\sigma/e)+(p,q)$.
\end{description}
\end{proposition}

\begin{example}\label{ex:simple}
Let us formulate Example~\ref{ex:intro} as a sequential matching game $(I, \sigma)$ with $I = (P, Q, E, \succ)$.
The firm set is $P=\{p_1,p_2,p_3\}$ and the worker set is $Q=\{q_1,q_2,q_3\}$.
The set of acceptable pairs is
\[E=\{(p_1,q_1),\,(p_2,q_1),\,(p_2,q_2),\,(p_2,q_3),\,(p_3,q_2),\,(p_3,q_3)\},\]
and the preferences $({\succ_r})_{r \in P \cup Q}$ are
\begin{align*}
  &p_1:~ q_1                                 && q_1:~ p_2 \succ_{q_1} p_1\\
  &p_2:~ q_2 \succ_{p_2} q_1 \succ_{p_2} q_3 && q_2:~ p_3 \succ_{q_2} p_2\\
  &p_3:~ q_3 \succ_{p_3} q_2                 && q_3:~ p_2 \succ_{q_3} p_3.
\end{align*}
The offering order $\sigma\in \Sigma_I$ is
\begin{align*}
	\bigl(\sigma(1), \sigma(2), \dots, \sigma(6)\bigr) = \bigl((p_1, q_1),\,(p_2, q_2),\,(p_2, q_1),\,(p_2, q_3),\,(p_3, q_3),\,(p_3, q_2)\bigr).
\end{align*}
We can also see that 
$\sigma$ is induced by a position order $\pi \colon \{1, 2, 3\} \to P$ such that $(\pi(1),\pi(2),\pi(3))=(p_1,p_2,p_3)$.
As seen in Example~\ref{ex:intro} (see the tree representation in Fig.~\ref{fig:tree}),
the SPE matching is 
\begin{align*}
	\spe(I, \sigma) = \{(p_1, q_1),\,(p_2, q_3),\,(p_3, q_2)\}.
\end{align*}

Let us demonstrate that
a different offering order may result in a different SPE matching
even if the matching instance is exactly the same.
\mbox{Consider the above instance $I$, and let $\sigma' \in \Sigma_I$ be} 
\begin{align*}
	\bigl(\sigma'(1), \sigma'(2), \dots, \sigma'(6)\bigr) = \bigl((p_2, q_2),\,(p_2, q_1),\,(p_2, q_3),\,(p_3, q_3),\,(p_3, q_2),\,(p_1, q_1)\bigr).
\end{align*}
Note that this $\sigma'$ is also position-based
(induced by a position order $\pi'$ with $(\pi'(1),\pi'(2),\pi'(3)) = (p_2, p_3, p_1)$).
The tree representation of $(I, \sigma')$ is shown in Fig.~\ref{fig:tree2},
and the SPE matching is
\begin{align*}
	\spe(I, \sigma') = \{(p_1, q_1),\,(p_2, q_2),\,(p_3, q_3)\}.
\end{align*}

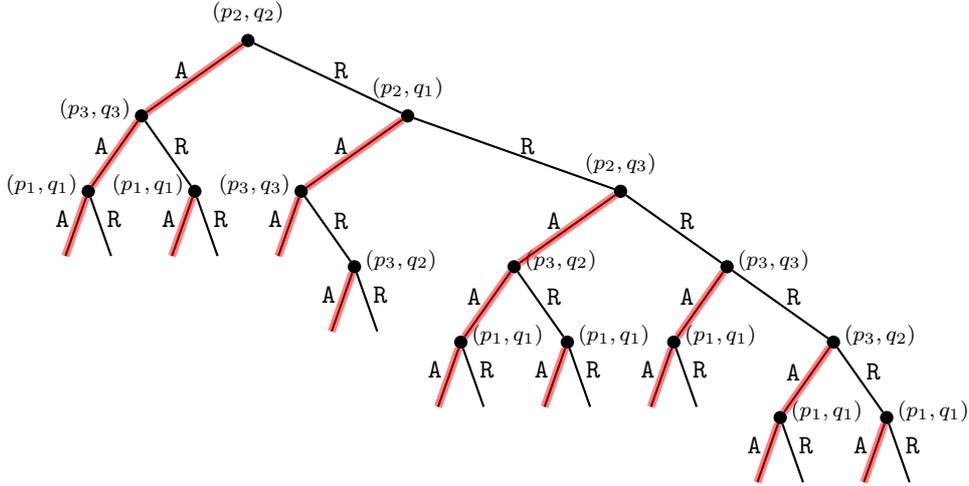
\begin{figure}[h]
\centering
\scalebox{1}{\begin{tikzpicture}[thick,xscale=0.7,yscale=1.0,
    p/.style={circle,fill=black,inner sep=0pt,outer sep=0pt,minimum size=5pt}]
  \node[p,label={[]{\scriptsize$(p_2,q_2)$}}] at (4.5,6)  (1-1) {}; 
  \node[p,label={[]{\scriptsize$(p_2,q_1)$}}] at (7.5,5)  (2-2) {}; 
  \node[p,label={[xshift=0pt,yshift=0pt]{\scriptsize$(p_2,q_3)$}}] at (11.5,4) (3-3) {}; 
  \node[p,label={[left]{\scriptsize$(p_3,q_3)$}}] at (2.5,5) (4-1) {};
  \node[p,label={[left]{\scriptsize$(p_3,q_3)$}}] at (5.5,4) (4-2) {};
  \node[p,label={[right]{\scriptsize$(p_3,q_3)$}}] at (13.5,3) (4-4) {};
    
  \node[p,label={[right]{\scriptsize$(p_3,q_2)$}}] at (6.5,3) (5-4) {};
  \node[p,label={[right]{\scriptsize$(p_3,q_2)$}}] at (9.5,3) (5-5) {};
  \node[p,label={[right]{\scriptsize$(p_3,q_2)$}}] at (15.5,2) (5-7) {};

  \foreach \x/\i/\j in {8.5/6/2,10.5/7/2,12.5/8/2,14.5/9/1,16.5/10/1}
    \node[p,label={[right]{\scriptsize$(p_1,q_1)$}}] at (\x,\j) (6-\i) {};
  \foreach \x/\i/\j in {1.5/1/4,3.5/2/4}
    \node[p,label={[left]{\scriptsize$(p_1,q_1)$}}] at (\x,\j) (6-\i) {};

  \foreach \i/\j in {1/3,2/3,3/3,4/3,5/3,6/2,7/2,8/1,9/1,10/1,11/1,12/1,13/1,14/0,15/0,16/0,17/0}
    \node[] at (\i,\j) (7-\i) {};

  \foreach \u/\v in {1-1/4-1,2-2/4-2,3-3/5-5,4-1/6-1,4-2/7-5,4-4/6-8,5-4/7-6,5-5/6-6,5-7/6-9,6-1/7-1,6-2/7-3,6-6/7-8,6-7/7-10,6-8/7-12,6-9/7-14,6-10/7-16}
    \draw (\u) -- (\v) node [pos=.5,font=\small,xshift=-5pt,yshift=3pt] {\texttt{A}};
  \foreach \u/\v in {1-1/2-2,2-2/3-3,3-3/4-4,4-1/6-2,4-2/5-4,4-4/5-7,5-4/7-7,5-5/6-7,5-7/6-10,6-1/7-2,6-2/7-4,6-6/7-9,6-7/7-11,6-8/7-13,6-9/7-15,6-10/7-17}
    \draw (\u) -- (\v) node [pos=.5,font=\small,xshift=5pt,yshift=3pt] {\texttt{R}};
  \foreach \u/\v in {1-1/4-1,2-2/4-2,3-3/5-5,4-1/6-1,4-2/7-5,4-4/6-8,5-4/7-6,5-5/6-6,5-7/6-9,6-1/7-1,6-2/7-3,6-6/7-8,6-7/7-10,6-8/7-12,6-9/7-14,6-10/7-16}
    \draw[line width=3pt,red,opacity=.5] (\u) -- (\v);

\end{tikzpicture}}
\vspace*{-4mm}
\caption{The tree representation of $(I, \sigma')$. The bold red edges indicate the SPE.}\label{fig:tree2}
\end{figure}
\end{example}

As mentioned in definition, for any sequential matching game, a unique SPE can be computed by the backward induction algorithm.
The computational time is, however, exponential in the input size because the tree representation has exponential size.

In this paper, we consider the problem of deciding whether $\partial_Q(\sigma(1))$ selects \ACCEPT in the first round.
We call this problem \SPEM as defined in Problem~\ref{prob:spem}.
The restriction of \SPEM in which the numbers of acceptable partners of each firm and of each worker are at most $s$ and $t$, respectively, is denoted by $(s,t)$-\SPEM.

\begin{remark}\label{rem:HW2011}
In \cite{HW2011}, similar sequential matching games were considered,
and there are several differences between their models and ours.
The most similar setting was studied in \cite[Section 4]{HW2011},
where the only difference is that,
after all firms have offered (where each firm offers to the most preferred worker among the workers to whom it has never offered),
the workers choose their actions (either to accept exactly one offer and reject the rest or to reject all offers) sequentially in a prespecified order.

For instance, let us consider the matching instance $I$ in Example~\ref{ex:simple}.
In the first stage, each firm $p_i$ offers to $q_i$ $(i=1,2,3)$.
Suppose that the first decision maker is $q_1$.
If $q_1$ rejects $p_1$'s offer, then $q_2$ and $q_3$ should accept the offers they received,
because otherwise they will be unmatched.
In this case, $q_1$ will be unmatched (no other firm will offer to $q_1$),
and hence $q_1$ should conclude to accept $p_1$'s offer.
Then, $q_2$ and $q_3$ can reject the offers from $p_2$ and $p_3$, respectively,
and they subsequently receive offers from $p_3$ and $p_2$, which are, respectively, their preferred firms, in the second stage.
Thus the SPE of this game yields a matching $\{(p_1, q_1),\, (p_2, q_3),\, (p_3, q_2)\}$, which is the same as $\spe(I, \sigma)$ in Example~\ref{ex:simple}.

In \cite[Theorem 2]{HW2011}, it is claimed that such a game enjoys a unique SPE,
and it results in the worker-optimal stable matching (which is formally defined in the next section).
However, this claim is not true, as evidenced by the above example,
in which the worker-optimal stable matching in $I$ (as well as a unique stable matching) is $\{(p_1, q_1),\, (p_2, q_2),\, (p_3, q_3)\}$.
\end{remark}

\subsection{Stability and Deferred Acceptance Algorithm}\label{sec:DA}
In this section, we briefly overview the stability and deferred acceptance algorithm (see \cite{knuth1976,GI1989,RS1991,manlove2013} for more details).
Let $I = (P, Q, E, {\succ})$ be a matching instance.

\begin{definition}[Stability \cite{GS62}]\label{def:stability}
For a matching $\mu \in \mathcal{M}_I$,
an acceptable pair $e = (p, q) \in E$ is called a \emph{blocking pair} (or \emph{blocks} $\mu$)
if $q \succ_p \mu(p)$ and $p \succ_q \mu(q)$
(where recall that each $r \in P \cup Q$ is regarded as the bottom of the preference $\succ_r$).
A matching $\mu \in \mathcal{M}_I$ is said to be \emph{stable} if there exists no blocking pair for $\mu$, and \emph{unstable} otherwise.
\end{definition}

The following statement is a one-to-one setting version of the rural hospital theorem \cite{GS1985, Roth1984, Roth1986},
which is well-known in the many-to-one stable matching setting.

\begin{theorem}[Invariance of unmatched agents \cite{MW70}]\label{thm:lwt}
If there exists a stable matching in $I$ under which $r\in P\cup Q$ is matched (resp., unmatched), then $r$ is matched (resp., unmatched) under every stable matching in $I$.
\end{theorem}

For $\mu, \mu' \in \mathcal{M}_I$,
we write $\mu \succeq_Q \mu'$ if $\mu(q) \succeq_q \mu'(q)$ $(\forall q \in Q)$,
and write $\mu \succ_Q \mu'$ if $\mu \neq \mu'$ in addition.
It is easy to observe that $\succeq_Q$ is a partial order over $\mathcal{M}_I$.
Subject to stability in $I$,
a maximal (minimal) element with respect to this partial order is unique.

\begin{definition}[Optimality \cite{knuth1976} (attributed to J.~H.~Conway)]\label{def:QOPT}
The set of all stable matchings in $I$ forms a distributive lattice
with respect to the partial order $\succeq_Q$.\footnotemark%
  \footnotetext{For every two matchings $\mu,\mu'\in\mathcal{M}_I$, there exist
  $\mu\vee \mu'\coloneqq \min\{\nu\in\mathcal{M}_I\mid \nu\succeq_Q \mu\text{ and }\nu\succeq_Q \mu'\}$ and
  $\mu\wedge \mu'\coloneqq \max\{\nu\in\mathcal{M}_I\mid \nu\preceq_Q \mu\text{ and }\nu\preceq_Q \mu'\}$.
  In addition, for every three matchings $\mu,\mu',\mu'\in\mathcal{M}_I$, we have
  $\mu\vee(\mu'\wedge \mu'')=(\mu\vee\mu')\wedge(\mu\vee\mu'')$ and
  $\mu\wedge(\mu'\vee \mu'')=(\mu\wedge\mu')\vee(\mu\wedge\mu'')$.
}
A unique maximal element in this distributive lattice is
said to be \emph{worker-optimal} or \emph{$Q$-optimal}, and denoted by $\qopt(I)$.
That is, $\qopt(I)(q) \succeq_q \mu(q)$ for any stable matching $\mu$ in $I$ and any worker $q \in Q$.
\end{definition}

The worker-optimal stable matching can be computed by the well-known \emph{deferred acceptance (DA) algorithm} \cite{GS62}, which is described in Algorithm 1.
Through the algorithm, the proposal with respect to each pair $e = (p, q) \in E$ and its rejection
occur at most once in Lines~\ref{line:4} and \ref{line:7}, respectively,
and the number of comparisons with respect to $\succ_p$ in Line~\ref{line:6} is equal to the number of workers rejected by $p$ in Line~\ref{line:7}.
Hence, the algorithm requires $\mathrm{O}(|E|)$ time in total.

\begin{algorithm}[htb]
  \caption{\sf $Q$-oriented DA~\cite{GS62}}\label{alg:QOPT}
  \SetKwInOut{Input}{input}\Input{A matching instance $I=(P,Q,E,\succ)$}
  \SetKwInOut{Output}{output}\Output{The $Q$-optimal stable matching $\qopt(I)$} 
  \vspace{1mm}
  Set $\hat{I} \ot I$ and $\mu(r) \ot r$ for each $r\in P \cup Q$\;
  \While{$\exists q \in Q$ such that $\mu(q) = q$ and $\Gamma_{\hat{I}}(q) \neq \emptyset$\label{line:2}}{
    \For{each $q \in Q$ such that $\mu(q) = q$ and $\Gamma_{\hat{I}}(q) \neq \emptyset$}{
      Set to $\mu(q)$ the most preferred firm in $\Gamma_{\hat{I}}(q)$ ($q$ proposes to $\mu(q) \in P$)\;\label{line:4}
    }
    \For{each $p \in P$ such that $\exists q \in Q$ with $\mu(q) = p$\label{line:5}}{
      Set to $\mu(p)$ the most preferred worker in $\{q\in Q \mid \mu(q) = p\}$ (among the workers who have proposed to $p$)\;\label{line:6}
      $\mu(q) \ot q$ and $\hat{I} \ot \hat{I} -(p, q)$ for each $q \in Q$ such that $\mu(q) = p$ and $\mu(p) \neq q$ ($p$ irrevocably rejects $q$'s proposal)\;\label{line:7}
    }
  }
  \Return $\mu$\;
\end{algorithm}

We here show several properties on the worker-optimal stable matchings,
which will be utilized in Sections~\ref{sec:restrict} and \ref{sec:imp_qopt}.

\begin{property}\label{prop:first_edge}
For $p \in P$ and $q \in Q$, if $q$ is on the top of $p$'s list,
then $\qopt(I)(q)\succeq_{q} p$.
\end{property}
\begin{proof}
Suppose to the contrary that $\qopt(I)(q) \prec_q p$.
We then have $\qopt(I)(p) \prec_p q$ as $\qopt(I)(p) \neq q$,
and hence $(p, q)$ is a blocking pair, contradicting the stability of $\qopt(I)$.
\end{proof}

\begin{property}\label{prop:delete_Q1}
For $e=(p,q) \in E$, if $\qopt(I)(q)=p$, 
then $\qopt(I)(q)\succ_{q} \qopt(I-e)(q)$.
\end{property}
\begin{proof}
Suppose to the contrary that $\qopt(I-e)(q) \succ_q \qopt(I)(q)= p$.
Since $\qopt(I -e)$ is a stable matching in $I -e$,
any $e' \in E - e$ is not a blocking pair.
Then, $\qopt(I-e)$ is also a stable matching in $I$,
because $e$ cannot be a blocking pair (due to $\qopt(I-e)(q) \succ_q p$).
This contradicts the $Q$-optimality of $\qopt(I)$.
\end{proof}

\begin{property}\label{prop:contract_Q}
For $e=(p,q) \in E$, if $\qopt(I)(q)=p$, 
then $\qopt(I/e)+(p,q)\succeq_{Q}\qopt(I)$.
\end{property}
\begin{proof}
Let $\mu \coloneq \qopt(I) - e$.
Since $\qopt(I)$ is stable in $I$, there is no pair $e' = (p', q') \in E/e$
such that $q' \succ_{p'} \mu(p')$ and $p' \succ_{q'} \mu(q')$,
which implies that $\mu$ is stable in $I/e$.
By $Q$-optimality, we have $\qopt(I/e) \succeq_{Q-q} \mu$,
and hence $\qopt(I/e) + (p, q) \succeq_{Q} \mu + (p, q) = \qopt(I)$.
\end{proof}

\begin{property}\label{prop:delete_Q2}
For $e=(p,q) \in E$, if $\qopt(I)(q)\succ_{q} p$, 
then $\qopt(I)=\qopt(I-e)$.
\end{property}
\begin{proof}
Since $\qopt(I)$ is stable in $I -e$,
we have $\qopt(I) \preceq_Q \qopt(I-e)$ by $Q$-optimality.
Then, $p \prec_q \qopt(I)(q) \preceq_q \qopt(I -e)(q)$,
and hence $e = (p, q)$ cannot be a blocking pair for $\qopt(I -e)$.
Thus $\qopt(I -e)$ is also stable in $I$,
and hence $\qopt(I) \succeq_Q \qopt(I-e)$ by $Q$-optimality.
\end{proof}

\subsection{Perplexing Properties of SPE Matchings}\label{subsec:bo}
In this section, we demonstrate how difficult to capture SPE matchings in the sequential matching game is.
Specifically, we show several more properties that may not be achieved with concrete examples.
The contents are not directly related to the main results, and the readers who are convinced enough by Example~\ref{ex:intro} can skip this section.

The first one claims that SPE matchings do not satisfy a ``rural hospital theorem''-like property (cf.~Theorem~\ref{thm:lwt})
in contrast to stable matchings.

\begin{proposition}\label{prop:rural_hospital}
  There exists a matching instance such that 
  the set of firms and workers matched in the SPE matching changes according to the offering order.
\end{proposition}

\begin{example}[Proof of Proposition~\ref{prop:rural_hospital}]\label{ex:rural_hospital}
  Consider a matching instance $I=(P,Q,E,\succ)$ with four firms $P=\{p_1,p_2,p_3,p_4\}$ and four workers $Q=\{q_1,q_2,q_3,q_4\}$.
  The set of acceptable pairs is
  \begin{align*}
    E=\{
    (p_1,q_1),\,(p_1,q_2),\,(p_1,q_4),\,
    (p_2,q_2),\,(p_2,q_4),\,
    (p_3,q_1),\,(p_3,q_3),\,
    (p_4,q_4)
    \}
  \end{align*}
  and the preferences are
  \begin{align*}
    &p_1:~ q_2 \succ_{p_1} q_1 \succ_{p_1} q_4 && q_1:~ p_1 \succ_{q_1} p_3\\
    &p_2:~ q_4 \succ_{p_2} q_2 && q_2:~ p_2 \succ_{q_2} p_1\\
    &p_3:~ q_1 \succ_{p_3} q_3 && q_3:~ p_3\\
    &p_4:~ q_4 && q_4:~ p_1 \succ_{q_4} p_4 \succ_{q_4} p_2.
  \end{align*}  
  Let $\sigma_1$ and $\sigma_2$ be offering orders that are respectively induced by position orders $\pi_1$ and $\pi_2$ such that
  \begin{align*}
    \bigl(\pi_1(1),\pi_1(2),\pi_1(3),\pi_1(4)\bigr)=(p_1,p_2,p_3,p_4) ~~\text{and}~~
    \bigl(\pi_2(1),\pi_2(2),\pi_2(3),\pi_2(4)\bigr)=(p_4,p_3,p_2,p_1).
  \end{align*}
  Then, the SPE matchings are
  \begin{align*}
  \spe(I, \sigma_1) = \{(p_1,q_1),\,(p_2,q_2),\,(p_3,q_3),\,(p_4,q_4)\} ~~\text{and}~~
  \spe(I, \sigma_2) = \{(p_1,q_4),\,(p_2,q_2),\,(p_3,q_1)\}.
  \end{align*}
  Hence, the set of matched firms and workers (and even its size) is changed. 
\end{example}

\begin{example}[Alternative proof of Proposition~\ref{prop:rural_hospital}]\label{ex:rural_hospital2}
  In the previous example, the larger matching is stable, but this is not always true.
  Consider a matching instance $I=(P,Q,E,\succ)$ with four firms $P=\{p_1,p_2,p_3,p_4\}$ and four workers $Q=\{q_1,q_2,q_3,q_4\}$.
  The set of acceptable pairs is
  \begin{align*}
    E=\{
    (p_1,q_1),\,(p_1,q_2),\,(p_1,q_3),\,
    (p_2,q_2),\,(p_2,q_4),\,
    (p_3,q_2),\,(p_3,q_3),\,
    (p_4,q_1)
    \}
  \end{align*}
  and the preferences are
  \begin{align*}
    &p_1:~ q_2 \succ_{p_1} q_1 \succ_{p_1} q_3 && q_1:~ p_1 \succ_{q_1} p_4\\
    &p_2:~ q_2 \succ_{p_2} q_4 && q_2:~ p_3 \succ_{q_2} p_2 \succ_{q_2} p_1\\
    &p_3:~ q_3 \succ_{p_3} q_2 && q_3:~ p_1 \succ_{q_3} p_3\\
    &p_4:~ q_1 && q_4:~ p_2.
  \end{align*}  
  Let $\sigma_1$ and $\sigma_2$ be offering orders that are respectively induced by position orders $\pi_1$ and $\pi_2$ such that
  \begin{align*}
    \bigl(\pi_1(1),\pi_1(2),\pi_1(3),\pi_1(4)\bigr)=(p_1,p_2,p_3,p_4) ~~\text{and}~~
    \bigl(\pi_2(1),\pi_2(2),\pi_2(3),\pi_2(4)\bigr)=(p_2,p_4,p_1,p_3).
  \end{align*}
  Then, the SPE matchings are
  \begin{align*}
  \spe(I, \sigma_1) = \{(p_1,q_1),\,(p_2,q_2),\,(p_3,q_3)\} ~~\text{and}~~
  \spe(I, \sigma_2) = \{(p_1,q_3),\,(p_2,q_4),\,(p_3,q_2),\,(p_4,q_1)\}.
  \end{align*}
  Hence, the set of matched firms and workers (and even its size) is changed, and in this case,
  the smaller matching is stable.
  Moreover, by combining the above instance with that in Example~\ref{ex:rural_hospital}, one can construct a matching instance 
  such that the set of SPE matchings with all offering orders includes (i) a matching smaller than its stable matchings, (ii) a matching larger than them, and (iii) a matching incomparable to them in the sense of the set of matched firms and workers. 
\end{example}

The second one is on stability;
we already observe in Example~\ref{ex:intro} that an SPE matching may not be stable in the standard sense (Definition~\ref{def:stability}),
and this is also true for weaker stabilities, which are defined in Definitions~\ref{def:vNM} and~\ref{def:essential_stability}.

\begin{definition}[von Neumann--Morgenstern (vNM) stability~\cite{Ehlers2007}]\label{def:vNM}
  For matchings $\mu,\mu'\in\mathcal{M}_I$, we say that $\mu$ dominates $\mu'$ if some pair $(p,q)\in \mu$ blocks $\mu'$, 
  i.e., $q\succ_p\mu'(p)$ and $p\succ_q\mu'(q)$.
  A nonempty set of matchings $V\subseteq\mathcal{M}_I$ is called a \emph{vNM-stable set} if it satisfies the following two properties:
  (internal stability) any $\mu \in V$ does not dominate another $\mu' \in V$;
  (external stability) any $\mu'\in\mathcal{M}_I\setminus V$ is dominated by some $\mu\in V$.
  It is known that $V$ is uniquely determined~\cite{Ehlers2007,Wako2010}.
  We say that $\mu\in\mathcal{M}_I$ is a \emph{vNM-stable matching} or satisfies \emph{vNM-stability} if $\mu\in V$.
\end{definition}

\begin{remark}
The vNM-stable set also forms a distributive lattice, which includes all stable matchings.
Hence, the worker-optimal vNM-stable matching is defined,
and all workers weakly prefer it to the worker-optimal stable matching.
Furthermore, it is also worker-side Pareto-efficient (see Definition~\ref{def:Pareto} in Section~\ref{sec:imp}).
\end{remark}

\begin{definition}[Essential stability~\cite{KT2016}]\label{def:essential_stability}
For a matching $\mu \in \mathcal{M}_I$ and a blocking pair $(p_0,q_0)$ for $\mu$,
the \emph{reassignment chain initiated by $q_0$} is the list
$q_0, p_0, q_1, p_1,\dots, q_\ell, p_\ell$ defined as follows:
let $\mu_0 \coloneq \mu$, and for each $k = 1, 2, \dots, \ell$,
\begin{align*}
  \mu_k(q) \coloneq \begin{cases}\mu_{k-1}(q)&(q\ne q_{k-1}, \mu_{k-1}(p_{k-1}))\\p_{k-1}&(q=q_{k-1})\\q&(q=\mu_{k-1}(p_{k-1}) \eqqcolon q_k),\end{cases} 
\end{align*}
and let $p_k$ be $q_k$'s most preferred position subject to the condition that $(p_k, q_k)$ blocks $\mu_k$,
where $\ell \geq 1$ is defined so that $\mu_\ell(p_\ell) = p_\ell$ and $\mu_k(p_k)\ne p_k$ for all $k<\ell$.
We say that a blocking pair $(p_0,q_0)$ is \emph{vacuous for $q_0$} if $\mu_\ell(q_0)\ne p_0$.
A matching $\mu \in \mathcal{M}_I$ is \emph{essentially stable for the workers}\footnote{The essential stability for the firms are defined symmetrically.} if all blocking pairs $(p,q)$ for $\mu$ are vacuous for $q$.
\end{definition}

\begin{proposition}\label{prop:weak_stability}
  There exists an SPE matching that is neither vNM-stable nor essentially stable.
\end{proposition}

\begin{example}[Proof of Proposition~\ref{prop:weak_stability}]\label{ex:weak_stability}
  Consider a matching instance $I=(P,Q,E,\succ)$ with three firms $P=\{p_1,p_2,p_3\}$ and three workers $Q=\{q_1,q_2,q_3\}$.
  The set of acceptable pairs is
  \begin{align*}
    E=\{(p_1,q_1),\,(p_1,q_2),\,(p_2,q_1),\,(p_2,q_2),\,(p_2,q_3),\,(p_3,q_1),\,(p_3,q_3)\}
  \end{align*}
  and the preferences are
  \begin{align*}
    &p_1:~ q_1 \succ_{p_1} q_2                 && q_1:~ p_3 \succ_{q_1} p_1 \succ_{q_1} p_2\\
    &p_2:~ q_1 \succ_{p_2} q_2 \succ_{p_2} q_3 && q_2:~ p_2 \succ_{q_2} p_1\\
    &p_3:~ q_3 \succ_{p_3} q_1                 && q_3:~ p_2 \succ_{q_3} p_3.
  \end{align*}  
  Let $\sigma$ be an offering order that is induced by a position order $\pi$ such that
  \[\bigl(\pi(1),\pi(2),\pi(3)\bigr)=(p_1,p_2,p_3).\]
  Then, $\spe(I, \sigma) = \{(p_1,q_2),\,(p_2,q_3),\,(p_3,q_1)\}$,
  while the vNM-stable set is a singleton consisting of $\{(p_1,q_1),\,(p_2,q_2),\,(p_3,q_3)\}$.
  Hence, the SPE matching does not satisfy vNM-stability.
  Moreover, this SPE matching is not essentially stable because a unique blocking pair $(p_2, q_2)$ is non-vacuous both for the workers and for the firms.
\end{example}

\section{Algorithm for Tractable Cases of SPE Computation}\label{sec:restrict}
As one side of our main dichotomy result Theorem~\ref{thm:dichotomy}, we show the following theorem.
Recall that $(2,\infty)$-\SPEM (resp., $(\infty,2)$-\SPEM) is the restriction of \SPEM in which the length of each firm's (resp., worker's) preference list is at most two.

\begin{theorem}\label{thm:restricted_1}
The problems $(2,\infty)$-\SPEM and $(\infty,2)$-\SPEM are polynomially solvable.
\end{theorem}
Note that, for an instance $(I, \sigma)$, 
we can solve \SPEM efficiently if we can compute $\spe(I,\sigma)$ efficiently,
because whether $\partial_Q(\sigma(1))$ selects \ACCEPT at the first round or not
corresponds to whether $\sigma(1)\in \spe(I,\sigma)$ holds or not.
Therefore, to show Theorem~\ref{thm:restricted_1},
it suffices to prove the following theorem, which is the purpose of this section.

\begin{theorem}\label{thm:restricted_2}
For any $(2,\infty)$-\SPEM or $(\infty,2)$-\SPEM instance  $(I,\sigma)$,
we can compute the SPE matching $\spe(I,\sigma)$ in polynomial time.
\end{theorem}

We show Theorem~\ref{thm:restricted_2} in the following way.
In Section~\ref{sec:SFDA}, we introduce an efficient algorithm,
which we call the ($Q$-oriented) {\em sequentially fixing deferred acceptance algorithm} (SFDA),  
and provide its properties.
Then, Section~\ref{sec:proof_restrict} shows that 
the output of SFDA for $(I,\sigma)$ coincides with $\spe(I,\sigma)$ if $(I,\sigma)$ is a $(2,\infty)$- or $(\infty, 2)$-\SPEM instance.

\subsection{Sequentially Fixing Deferred Acceptance Algorithm}\label{sec:SFDA}
This section introduces the ($Q$-oriented) {\em sequentially fixing deferred acceptance algorithm} (SFDA), which is sketched as follows: 
given an instance $(I,\sigma)$, 
the algorithm repeatedly computes the $Q$-optimal stable matching and 
updates the instance by removing the firm-worker pair that appears first in the offering order among all the matched pairs.
The algorithm stops the repetition when there is no acceptable pair, 
and then outputs 
the matching that consists of all the pairs removed so far. 
We denote the output matching by $\sfda(I,\sigma)$.
The formal description of SFDA is given in Algorithm \ref{alg:SFDA}.
Recall that $\qopt(I)$ in the description 
denotes the $Q$-optimal stable matching
of a matching instance $I$, which is computed by Algorithm~\ref{alg:QOPT}.

\begin{algorithm}[htb]
  \caption{\sf $Q$-oriented SFDA}\label{alg:SFDA}
  \SetKwInOut{Input}{input}\Input{An instance $(I,\sigma)$, where $I=(P,Q,E,\succ)$}
  \SetKwInOut{Output}{output}\Output{A matching $\sfda(I,\sigma)$}
  \vspace{1mm}
  Set $(\hat{I},\hat{\sigma})\ot (I,\sigma)$ and let $\mu(r)\ot r$ for every $r\in P\cup Q$\; 
  \While{there is some acceptable pair in $\hat{I}$}{
    Let $e=(p,q)$ be the pair that minimizes $\sigma^{-1}(e)$ subject to $e\in \qopt(\hat{I})$\; \label{line:fix}
    $(\hat{I},\hat{\sigma})\ot (\hat{I}/e,\hat{\sigma}/e)$\;
    $\mu(p)\ot q$ and $\mu(q)\ot  p$\;
  }
  \Return $\mu$\;
\end{algorithm}

The algorithm SFDA is inspired by the ($Q$-oriented) efficiency adjusted deferred acceptance algorithm (EADA), which is introduced by Kesten~\cite{Kesten2010} and simplified by Bando~\cite{Bando2014} and Tang and Yu~\cite{TY2014}.
While EADA iteratively fixes a last proposer\footnotemark{} under DA, our algorithm SFDA fixes the minimum pair in the predetermined order $\sigma$.
EADA produces the worker-optimal vNM-stable matching, and hence SFDA also produces it if the order $\sigma$ is consistent with the fixing order in EADA\@.
\footnotetext{This is a modified version of the algorithm EADA in~\cite{Bando2014}. In~\cite{Bando2014}, the set of all last proposers is fixed at a time, while, in this version, only one of them is fixed. Fortunately, this modification does not change the outcome.}

For any \SPEM instance $(I,\sigma)$ where $I = (P, Q, E, {\succ})$,
the output matching $\sfda(I,\sigma)$ has the following properties.

\begin{lemma}\label{claim:sfda_vs_qopt}
$\sfda(I,\sigma)\succeq_{Q}\qopt(I)$ for any \SPEM instance $(I, \sigma)$.
\end{lemma}
\begin{proof}
We use induction on subgames. 
Let $e=(p,q)$ be the first pair chosen at Line~\ref{line:fix} in Algorithm~\ref{alg:SFDA}. 
Then, $\sfda(I,\sigma)=\sfda(I/e,\sigma/e)+(p,q)$.
Since $\sfda(I/e,\sigma/e)\succeq_{Q}\qopt(I/e)$ by induction,
then $\sfda(I,\sigma)=\sfda(I/e,\sigma/e)+(p,q)\succeq_{Q}\qopt(I/e)+(p,q)\succeq_{Q}\qopt(I)$,
where the last relation follows from Property~\ref{prop:contract_Q}.
\end{proof}

\begin{lemma}\label{claim:operations_sfda}
For $e=(p,q)=\sigma(1)$ in any \SPEM instance $(I, \sigma)$, we have the following properties.
\begin{description}
\item[(a)] If $\qopt(I)(q)\neq p$, 
then $\sfda(I,\sigma)=\sfda(I-e,\sigma-e)$.
\item[(b)] If $\qopt(I)(q)= p$, 
then $\sfda(I,\sigma)=\sfda(I/e,\sigma/e)+(p,q)$.
\end{description}
\end{lemma}
\begin{proof}
By definition, (b) is clear; since $e=(p,q)=\sigma(1)$, 
the condition $\qopt(I)(q)=p$ means that $e$ is the first pair chosen at Line~\ref{line:fix}, 
and hence $\sfda(I,\sigma)=\sfda(I/e,\sigma/e)+(p,q)$.

We show (a) by induction on subgames. 
Since $e=(p,q)=\sigma(1)$ and $\qopt(I)(q)\neq p$, Property~\ref{prop:first_edge} implies $\qopt(I)(q)\succ_{q} p$.
Then Property~\ref{prop:delete_Q2} implies $\qopt(I)=\qopt(I-e)$.
Let $e^{\ast} \in \qopt(I)$ 
be the first pair chosen at Line~\ref{line:fix}.
As $\qopt(I)=\qopt(I-e)$, SFDA also chooses $e^{\ast}$ when the input is $(I-e, \sigma-e)$.
By definition, $\sfda(I,\sigma)=\sfda(I/e^{\ast},\sigma/e^{\ast})+e^{\ast}$ and
$\sfda(I-e,\sigma-e)=\sfda((I-e)/e^{\ast},(\sigma-e)/e^{\ast})+e^{\ast}
=\sfda((I/e^{\ast})-e,(\sigma/e^{\ast})-e)+e^{\ast}$.
Note that Property~\ref{prop:contract_Q} implies 
$\qopt(I/e^{\ast})(q)\succeq_{q} \qopt(I)(q)\succ_{q} p$,
and hence
$\sfda(I/e^{\ast},\sigma/e^{\ast})=\sfda((I/e^{\ast})-e,(\sigma/e^{\ast})-e)$ 
by induction.
Thus, we obtain $\sfda(I,\sigma)=\sfda(I-e, \sigma-e)$.
\end{proof}

The output of SFDA depends on the order $\sigma$, 
i.e., we may have $\sfda(I,\sigma)\neq \sfda(I,\sigma')$ if $\sigma\neq \sigma'$.
However, the set of firms and workers matched in $\sfda(I,\sigma)$ 
coincides with that of $\qopt(I)$ independently of $\sigma$,
i.e., the following lemma holds (in contrast to Proposition~\ref{prop:rural_hospital} for $\spe(I, \sigma)$).
\begin{lemma}\label{claim:sfda_rural_hospital}
For any $r\in P\cup Q$ in any \SPEM instance $(I, \sigma)$,
we have $\sfda(I,\sigma)(r)= r$ if and only if $\qopt(I)(r)= r$.
\end{lemma}
\begin{proof}
We use induction on subgames. 
Let $e=(p,q)$ be the first pair chosen at Line~\ref{line:fix} in Algorithm~\ref{alg:SFDA}.
Clearly, $p$ and $q$ are matched both in $\sfda(I,\sigma)$ and in $\qopt(I)$.
For each $r\in (P-p)\cup (Q-q)$, we have
$\sfda(I,\sigma)(r)=\sfda(I/e,\sigma/e)(r)$ by definition.
Also, by induction, 
we have $\sfda(I/e,\sigma/e)(r)= r$ if and only if $\qopt(I/e)(r)= r$.
Note that $\qopt(I)-(p,q)$ is a stable matching in $I/e$ 
by the definition of the stability. 
Then Theorem~\ref{thm:lwt} implies that, for each $r\in (P-p)\cup (Q-q)$, 
$\qopt(I/e)(r)= r$ if and only if $\qopt(I)(r)= r$.
Combining these, we obtain that, for every $r\in P\cup Q$,
we have $\sfda(I,\sigma)(r)= r$ if and only if $\qopt(I)(r)= r$.
\end{proof}
Finally, we analyze the time complexity of SFDA\@. 
The algorithm repeats finding the $Q$-optimal stable matching of each updated instance, 
which is done in $\mathrm{O}(|E|)$ time by Algorithm~\ref{alg:QOPT} (see Section~\ref{sec:DA}). 
Because one pair is fixed at each iteration stage, the number of iteration is at most $\min\{|P|, |Q|\}$.
Thus, the algorithm SFDA runs in $\mathrm{O}(|E|\cdot\min\{|P|, |Q|\})$ time.
\begin{lemma}\label{claim:time_comp}
For any \SPEM instance $(I, \sigma)$, the matching $\sfda(I,\sigma)$ can be computed in polynomial time.
\end{lemma}

\subsection{Proof of the Tractability}\label{sec:proof_restrict}
This part is devoted to show the following theorem,
which completes the proof of Theorem~\ref{thm:restricted_2}, and hence of Theorem~\ref{thm:restricted_1}.
\begin{theorem}\label{thm:restricted_3}
$\sfda(I,\sigma)=\spe(I,\sigma)$
for any $(2,\infty)$- or $(\infty,2)$-\SPEM instance $(I,\sigma)$.
\end{theorem}

We prove Theorem~\ref{thm:restricted_3} by induction on subgames
linking Lemma~\ref{claim:operations_sfda} and Proposition~\ref{prop:operations_spe},
the inductive properties of $\sfda(I,\sigma)$ and $\spe(I,\sigma)$, respectively.
For this purpose, we provide the following special properties of 
$(2,\infty)$- and $(\infty,2)$-\SPEM instances. 

\begin{lemma}\label{claim:2inf}
$\sfda(I,\sigma)=\qopt(I)$ for any $(2,\infty)$-\SPEM instance $(I,\sigma)$.
\end{lemma}
\begin{proof}
We use induction on subgames.
Let $e=(p,q)=\sigma(1)$.
There are two cases.

If $\qopt(I)(q)\neq p$, then 
$\sfda(I,\sigma)=\sfda(I-e, \sigma -e)$ by Lemma~\ref{claim:operations_sfda} (a),
and 
$\qopt(I)=\qopt(I-e)$ by Properties~\ref{prop:first_edge} and \ref{prop:delete_Q2}.
Since $\qopt(I-e)=\sfda(I-e, \sigma -e)$ by induction,
we obtain $\qopt(I)=\sfda(I,\sigma)$.

If $\qopt(I)(q)= p$, then 
Lemma~\ref{claim:operations_sfda} (b) implies 
$\sfda(I,\sigma)=\sfda(I/e, \sigma/e)+(p,q)$,
which equals $\qopt(I/e)+(p,q)$ by induction.
Then it suffices to show 
$\qopt(I)=\qopt(I/e)+(p,q)$.
Suppose, to the contrary, this equality fails.
As we have $\qopt(I/e)+(p,q)\succeq_{Q}\qopt(I)$
by Property~\ref{prop:contract_Q}, the $Q$-optimality of $\qopt(I)$  implies that
$\mu\coloneq \qopt(I/e)+(p,q)$ is unstable in $I$, i.e., 
there is $(p',q')\in E$ with $q'\succ_{p'}\mu(p')$ and $p'\succ_{q'}\mu(q')$.
Since $\qopt(I/e)$ is stable in $I/ e$, 
we have $(p',q')\not\in E/e = E \setminus (\delta_I(p) \cup \delta_I(q))$.
Hence, $p'=p$ or $q'=q$ holds. 
If $p'=p$, then $q'\succ_{p}\mu(p)=q=\qopt(I)(p)$ and $p\succ_{q'}\mu(q')\succeq_{q'} \qopt(I)(q')$,
and hence $(p,q')$ blocks $\qopt(I)$, a contradiction.
Therefore, we have $q'=q$ and $p'\neq p$, from which
$q\succ_{p'}\mu(p')$ and $p'\succ_{q}\mu(q)=p=\qopt(I)(q)$ follow. 
Since $\qopt(I)$ is not blocked by $(p',q')~(=(p',q))$,  
the latter condition implies $\qopt(I)(p')\succeq_{p'}q$.
As we have $\qopt(I)(p')\neq \qopt(I)(p)=q$, 
we obtain $\qopt(I)(p')\succ_{p'} q\succ_{p'}\mu(p')$.
Note that $\mu(p')$ is not $p'$ itself because $p'$ is matched in the stable matching $\qopt(I)-(p,q)$ of $I/e$ and hence in $\qopt(I/e)=\mu-(p,q)$ by Theorem~\ref{thm:lwt} (invariance of unmatched agents).
Thus, $\qopt(I)(p'), q, \mu(p')\in Q$ are three different entries of the preference list of $p'$,
which contradicts that the list of $p'$ has at most two entries.
\end{proof}

\begin{lemma}\label{claim:2inf_inf2}
For any $(2,\infty)$- or $(\infty,2)$-\SPEM instance $(I,\sigma)$
and any $e=(p,q)\in E$, the condition $\qopt(I)(q)=p$ implies
$p\succ_{q} \sfda(I-e,\sigma-e)(q)$.
\end{lemma}
\begin{proof}
In the case of $(2,\infty)$-\SPEM, 
the subgame $(I-e,\sigma-e)$ is also a $(2,\infty)$-\SPEM
instance, and hence
$\qopt(I-e)=\sfda(I-e,\sigma-e)$
by Lemma~\ref{claim:2inf}. 
Therefore, Property~\ref{prop:delete_Q1} implies that $p\succ_{q}\qopt(I-e)(q)=\sfda(I-e,\sigma-e)(q)$.

In the case of $(\infty,2)$-\SPEM, 
since $q$'s preference list is of length at most $2$, 
$p$ is on the top or the bottom of the list.
If $p$ is on the top, clearly $p\succ_{q} \sfda(I-e,\sigma-e)(q)$.
If $p$ is on the bottom, by Property~\ref{prop:delete_Q1},
$p=\qopt(I)(q)\succ_{q}\qopt(I-e)(q)=q$. 
By Lemma~\ref{claim:sfda_rural_hospital}, this implies $\sfda(I-e,\sigma-e)(q)=q$.
Thus, we have $p\succ_{q} q=\sfda(I-e,\sigma-e)(q)$.
\end{proof}

Now we are ready to show Theorem~\ref{thm:restricted_3}, 
which states that the output of SFDA coincides with the SPE matching 
for a $(2,\infty)$- or $(\infty,2)$-\SPEM instance.

\begin{proof}[Proof of Theorem~\ref{thm:restricted_3}]
We show $\sfda(I,\sigma)=\spe(I,\sigma)$ by induction on subgames.
Let $e=(p,q)=\sigma(1)$. 
Note that for a $(2,\infty)$-\SPEM (resp., $(\infty,2)$-\SPEM) instance $(I,\sigma)$,
the subgames $(I-e,\sigma-e)$ and $(I/e,\sigma/e)$ are also $(2,\infty)$-\SPEM (resp., $(\infty,2)$-\SPEM) instances.
Therefore, by induction, we have $\sfda(I/e,\sigma/e)=\spe(I/e,\sigma/e)$
and $\sfda(I-e,\sigma-e)=\spe(I-e,\sigma-e)$.
We consider two cases: $\qopt(I)(q)\neq p$ and $\qopt(I)(q)=p$.

When $\qopt(I)(q)\neq p$, Property~\ref{prop:first_edge} and Lemma~\ref{claim:sfda_vs_qopt}
imply $\sfda(I,\sigma)(q)\succeq_{q} \qopt(I)(q)\succ_{q} p$.
Also, Lemma~\ref{claim:operations_sfda} (a)
implies $\sfda(I,\sigma)=\sfda(I-e,\sigma-e)=\spe(I-e,\sigma-e)$.
Therefore, $\spe(I-e,\sigma-e)(q)\succ_{q}p$.
By Proposition~\ref{prop:operations_spe} (a),  $\spe(I,\sigma)=\spe(I-e,\sigma-e)=\sfda(I,\sigma)$.

When $\qopt(I)(q)= p$, Lemma~\ref{claim:operations_sfda} (b)
implies $\sfda(I,\sigma)=\sfda(I/e,\sigma/e)+(p,q)=\spe(I/e,\sigma/e)+(p,q)$.
Since $(I/e,\sigma/e)$ is a $(2,\infty)$- or $(\infty,2)$-\SPEM instance,
by Lemma~\ref{claim:2inf_inf2}, we have $p\succ_{q}\sfda(I-e,\sigma-e)(q)=\spe(I-e,\sigma-e)(q)$.
By Proposition~\ref{prop:operations_spe} (b), then 
$\spe(I,\sigma)=\spe(I/e,\sigma/e)+(p,q)=\sfda(I,\sigma)$.
\end{proof}
In particular, for $(2,\infty)$-\SPEM, combining
Theorem~\ref{thm:restricted_3} and Lemma~\ref{claim:2inf}
gives the following corollary.
\begin{corollary}\label{cor:2inf}
$\spe(I,\sigma)=\qopt(I)$ for  any $(2,\infty)$-\SPEM instance $(I,\sigma)$.
\end{corollary}
Thanks to Theorem~\ref{thm:restricted_3},
for any $(2,\infty)$-\SPEM or $(\infty,2)$-\SPEM instance,
we can compute the SPE matching of a given instance by the algorithm SFDA\@.
(In particular, for $(2,\infty)$-\SPEM, we can compute it by DA by Corollary~\ref{cor:2inf}.) 
Because SFDA runs in polynomial time as shown in Lemma~\ref{claim:time_comp}, 
Theorem~\ref{thm:restricted_3} immediately implies Theorem~\ref{thm:restricted_2}.
Then Theorem~\ref{thm:restricted_1} also follows.

\section{Hardness of General SPE Computation}\label{sec:hardness}
In this section, we show PSPACE-hardness of computing the SPE of a given sequential matching game
even if the length of each firm's and worker's preference list is at most three and the offering order is restricted to position-based.
Formally, we prove the following theorem.

\newcommand{\Q}{\mathscr{Q}}
\begin{theorem}\label{thm:PSPACE-complete}
$(3,3)$-\SPEM is PSPACE-complete even if the problem is restricted to position-based.
\end{theorem}

  First of all, we can solve $\SPEM$ in polynomial space by implementing backward induction as a depth-first search of the game tree.
  Thus, the problem is in class PSPACE\@.

  In what follows, we prove the hardness by giving a reduction from \QSAT, which is a PSPACE-complete problem~\cite{garey1979cai}.
  Let $V=\{v_1,\dots,v_n\}$ be the set of variables and
  $(\Q_1 v_1)(\Q_2 v_2)\cdots\allowbreak(\Q_n v_n)~\varphi(v_1,\dots,v_n)$ be a quantified Boolean formula.
  Here, $\Q_i$ is either $\forall$ or $\exists$ for $i=1,\dots,n$ and $\varphi(v_1,\dots,v_n)$ is a Boolean expression in 3-CNF\@.
  \QSAT asks whether the given quantified Boolean formula is $\TRUE$ or not.
  Let $J^\exists=\{i\mid \Q_i=\exists\}$ and $J^\forall=\{i\mid \Q_i=\forall\}$.
    
  In our reduction, we create a position-based $(3,3)$-\SPEM instance $(I,\sigma)$
  with $I = (P, Q, E, {\prec})$ that consists of
  three phases---assignment, evaluation, and output---as shown in Fig.~\ref{fig:hardness_overview}.

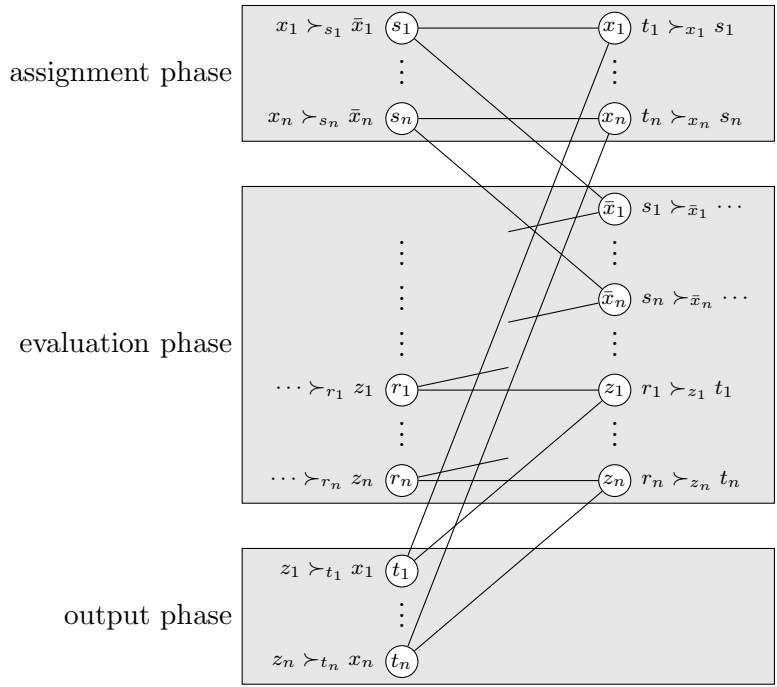
\begin{figure}[htbp]
\centering  
\begin{tikzpicture}[xscale=0.7,yscale=0.6,
    p/.style={circle,inner sep=0pt,draw=black,font=\scriptsize,fill=white,minimum size=12pt}]

  \node[left] at (-2,14) {assignment phase};
  \node[left] at (-2,8) {evaluation phase};
  \node[left] at (-2,2) {output phase};

  \draw[fill=black!10] (-2,12.5) rectangle (8,15.5);
  \draw[fill=black!10] (-2,4.5) rectangle (8,11.5);
  \draw[fill=black!10] (-2,0.5) rectangle (8,3.5);

  \node[p,label=left:{\scriptsize$x_1\succ_{s_1}\bar{x}_1$}] at (1,15)  (s1) {$s_1$};
  \node    at (1,14)     {$\rvdots$};
  \node[p,label=left:{\scriptsize$x_n\succ_{s_n}\bar{x}_n$}] at (1,13)  (sn) {$s_n$};

  \node[p,label=right:{\scriptsize$t_1\succ_{x_1}s_1$}] at (5,15)  (x1) {$x_1$};
  \node    at (5,14)     {$\rvdots$};
  \node[p,label=right:{\scriptsize$t_n\succ_{x_n}s_n$}] at (5,13)  (xn) {$x_n$};

  \node    at (1,10)      {$\rvdots$};
  \node    at (1,9)       {$\rvdots$};
  \node    at (1,8)       {$\rvdots$};
  \node[p,label=left:{\scriptsize$\cdots\succ_{r_1}z_1$}] at (1,7)  (r1) {$r_1$};
  \node    at (1,6)     {$\rvdots$};
  \node[p,label=left:{\scriptsize$\cdots\succ_{r_n}z_n$}] at (1,5)  (rn) {$r_n$};

  \node[p,label=right:{\scriptsize$s_1\succ_{\bar{x}_1}\cdots$}] at (5,11) (bx1) {$\bar{x}_1$};
  \node    at (5,10)      {$\rvdots$};
  \node[p,label=right:{\scriptsize$s_n\succ_{\bar{x}_n}\cdots$}] at (5,9)  (bxn) {$\bar{x}_n$};
  \node    at (5,8)        {$\rvdots$};

  \node[p,label=left:{\scriptsize$z_1\succ_{t_1}x_1$}] at (1,3)  (t1) {$t_1$};
  \node    at (1,2)     {$\rvdots$};
  \node[p,label=left:{\scriptsize$z_n\succ_{t_n}x_n$}] at (1,1)  (tn) {$t_n$};

  \node[p,label=right:{\scriptsize$r_1\succ_{z_1}t_1$}] at (5,7) (z1) {$z_1$};
  \node    at (5,6)    {$\rvdots$};
  \node[p,label=right:{\scriptsize$r_n\succ_{z_n}t_n$}] at (5,5) (zn) {$z_n$};

  \foreach \u/\v in {s1/x1,s1/bx1,sn/xn,sn/bxn,r1/z1,rn/zn,t1/x1,t1/z1,tn/xn,tn/zn} \draw (\u) -- (\v);
  \draw (r1)--(3,7.5);
  \draw (rn)--(3,5.5);
  \draw (bx1)--(3,10.5);
  \draw (bxn)--(3,8.5);

\end{tikzpicture}
\caption{Framework of the reduction}\label{fig:hardness_overview}
\end{figure}

  In the assignment phase, for each Boolean variable $v_i\in V$,
  we create two workers $x_i$ and $\bar{x}_i$, whose actions correspond to an assignment of $\TRUE$ or $\FALSE$ to $v_i$.
  Each $x_i$ has exactly two acceptable firms $s_i$ and $t_i$ with $t_i \succ_{x_i} s_i$.
  Suppose that $x_i$ is the best and the worst choice for $s_i$ and $t_i$, respectively.
  Thus, $x_i$ must be matched at least to $s_i$ in the matching $\spe(I, \sigma)$ (cf.~Lemma~\ref{lem:top_offer}).
  In addition, suppose that the best choice for $\bar{x}_i$ is $s_i$ and the second choice for $s_i$ is $\bar{x}_i$.
  Then, in the matching $\spe(I, \sigma)$, $\bar{x}_i$ matches to $s_i$ if $x_i$ does not match to $s_i$.

  Let $\pi$ be the position order inducing $\sigma$
  such that it starts with $s_1,\dots,s_n$ and ends with $t_1,\dots,t_n$.
  Let $a_i \in \{\TRUE, \FALSE\}$ for each $i = 1, \dots, n$.
  For each $0 \leq k \leq n$,
  let us denote by $(I, \sigma)_{[a_1, \dots, a_k]}$ the subgame of $(I, \sigma)$
  just after $s_1,\dots,s_k$ have finished their offers in which, for each $i = 1, \dots, k$,
  the contraction of $(s_i, x_i)$ has occurred (i.e., $x_i$ has chosen \ACCEPT for $s_i$'s offer) if $a_i = \TRUE$
  and the contraction of $(s_i, \bar{x}_i)$ has occurred (i.e., $x_i$ has chosen \REJECT for $s_i$'s offer and then $\bar{x}_i$ has chosen \ACCEPT) if $a_i = \FALSE$.
  Note that the subgames $(I, \sigma)_{[a_1, \dots, a_n]}$
  correspond to the assignments to the variables $v_1, \dots, v_n$
  in a one-to-one manner.

  In the evaluation phase, for each Boolean variable $v_i\in V$,
  we create a firm $r_i$ and a worker $z_i$ with $(r_i, z_i)\in E$.
  We will design the evaluation phase so that it correctly evaluates
  the value of $\varphi(a_1, \dots, a_n)$ for every assignment $(a_1, \dots, a_n) \in \{\TRUE, \FALSE\}^n$
  by the SPE of the corresponding subgame $(I, \sigma)_{[a_1, \dots, a_n]}$.
  Formally, we prove the following claim.

\begin{claim}\label{claim:hoge}
One can design the evaluation phase in Fig.~\ref{fig:hardness_overview} so that, 
for any assignment $(a_1, \dots, a_n)$
with the corresponding subgame $(I', \sigma') \coloneq (I,\sigma)_{[a_1, \dots, a_n]}$, the following holds: 
\begin{itemize}
\item[(i)]
  for any $i\in J^\exists$, we have $(r_i, z_i) \in \spe(I', \sigma') \iff \varphi(a_1, \dots, a_n) = \TRUE$, and
\item[(ii)]
  for any $i\in J^\forall$, we have $(r_i, z_i) \in \spe(I', \sigma') \iff \varphi(a_1, \dots, a_n) = \FALSE$.
\end{itemize}
\end{claim}

We postpone proving Claim~\ref{claim:hoge},
and suppose that we use the evaluation phase designed so.
The following claim shows whether each worker $x_i$ should choose \ACCEPT for $s_i$'s offer or not.

\begin{claim}\label{claim:fuga}
  For any $1 \leq i \leq n$ and any partial assignment $(a_1, \dots, a_{i-1}) \in \{\TRUE, \FALSE\}^{i-1}$, define
  \[\varphi' \coloneq (\Q_{i+1} v_{i+1})\cdots(\Q_n v_n)~\varphi(a_1,\dots,a_{i-1},\FALSE,v_{i+1},\dots,v_n).\]
  Then, in the corresponding subgame $(I', \sigma') \coloneq (I, \sigma)_{[a_1, \dots, a_{i-1}]}$,
\begin{itemize}
\item[(i)]
  if $i\in J^\exists$, we have $(s_i, \bar{x}_i) \in \spe(I', \sigma') \iff \varphi' = \TRUE$, and
\item[(ii)]
  if $i\in J^\forall$, we have $(s_i, \bar{x}_i) \in \spe(I', \sigma') \iff \varphi' = \FALSE$.
\end{itemize}
\end{claim}

\begin{proof}
  We prove this by backward induction on $i$.
  The base case $i=n$ holds by Claim~\ref{claim:hoge} because $x_n$ gets an offer from $t_n$ $(\succ_{x_n} s_n)$
  if and only if $(\Q_n, \varphi') = (\exists, \TRUE)$ or $(\forall, \FALSE)$.
  For a given $1\le i<n$, suppose that the claim holds for $j\ge i+1$.
  By the induction hypothesis and Claim~\ref{claim:hoge},
  if $x_i$ chooses \REJECT for $s_i$'s offer (note that then $\bar{x}_i$ must choose \ACCEPT for $s_i$'s offer),
  then $x_i$ will be offered from $t_i$ $(\succ_{x_i} s_i)$ if and only if
  $(\Q_i, \varphi') = (\exists, \TRUE)$ or $(\forall, \FALSE)$.
  Thus, the claim holds.
\end{proof}

  Without loss of generality, we may assume that $\Q_1=\exists$ and $v_1$ does not appear in the 3-CNF\@.
  Then, $x_1$ chooses \REJECT for $s_1$'s offer in the SPE if and only if the given \QSAT is a YES-instance.
  Hence, deciding whether $x_1$ chooses \ACCEPT for the first offer $\sigma(1) = (s_1, x_1)$ in the SPE of $(I,\sigma)$ is PSPACE-complete.

It remains to explain the details of the evaluation phase.
Before the proof of Claim~\ref{claim:hoge}, we prepare several gadgets as follows. 
We use gadgets to simulate AND, OR, and NOT-gates (see Figs.~\ref{fig:AND}, \ref{fig:OR}, and \ref{fig:NOT}, respectively).
Also, we use a gadget for branching, i.e., to make a copy of each value (see Fig.~\ref{fig:BRANCHING}).
In our gadgets, an input is represented by a worker; she means \TRUE if and only if she is matched with a firm outside the gadget when the first firm in the gadget moves.
An output is represented by a firm; it means \TRUE if and only if it is not matched to a worker in the gadget when all the firms in the gadget have finished their offers.
For each gadget, we put an index $u$ to distinguish individuals.

\begin{figure}[t]
\begin{minipage}{.5\textwidth}
\centering  
\scalebox{.9}{%
\begin{tikzpicture}[xscale=0.8,yscale=0.8,
    p/.style={circle,inner sep=0pt,draw=black,font=\scriptsize,fill=white,minimum size=13pt}]
  \draw[fill=black!10] (-3,0.5) rectangle (5.5,3.5);
  \node[p,label=left:{\scriptsize$q\succ_{\hat{\alpha}_1^u}\alpha_1^u$}] at (1,3)  (p1) {$\hat{\alpha}_1^u$};
  \node[p,label=left:{\scriptsize$q'\succ_{\hat{\alpha}_2^u}\alpha_2^u$}] at (1,2)  (p2) {$\hat{\alpha}_2^u$};
  \node[p,label=left:{\scriptsize$\alpha_1^u\succ_{\hat{\alpha}_3^u}\alpha_2^u\succ_{\hat{\alpha}_3^u}\alpha_3^u$}] at (1,1)  (p3) {$\hat{\alpha}_3^u$};
  \node[p,label=left:{\scriptsize$\alpha_3^u\succ_{p}\cdots$}] at (1,0)  (pout) {$p$};

  \node[p,label=right:{\scriptsize$\cdots\succ_q \hat{\alpha}_1^u$}] at (3,5)  (qin1) {$q$};
  \node[p,label=right:{\scriptsize$\cdots\succ_{q'} \hat{\alpha}_2^u$}] at (3,4)  (qin2) {$q'$};
  \node[p,label=right:{\scriptsize$\hat{\alpha}_1^u\succ_{\alpha_1^u}\hat{\alpha}_3^u$}] at (3,3)  (q1) {$\alpha_1^u$};
  \node[p,label=right:{\scriptsize$\hat{\alpha}_2^u\succ_{\alpha_2^u}\hat{\alpha}_3^u$}] at (3,2)  (q2) {$\alpha_2^u$};
  \node[p,label=right:{\scriptsize$\hat{\alpha}_3^u\succ_{\alpha_3^u}p$}] at (3,1)  (q3) {$\alpha_3^u$};

  \foreach \u/\v in {qin1/p1,qin2/p2,q1/p1,q1/p3,q2/p2,q2/p3,q3/p3,q3/pout} \draw (\u) -- (\v);
\end{tikzpicture}}
\caption{AND-gate simulator $\AND^u$}\label{fig:AND}
\end{minipage}%
\begin{minipage}{.5\textwidth}
\centering  
\scalebox{.9}{%
\begin{tikzpicture}[xscale=0.8,yscale=0.8,
    p/.style={circle,inner sep=0pt,draw=black,font=\scriptsize,fill=white,minimum size=13pt}]
  \draw[fill=black!10] (-1.6,0.5) rectangle (7,3.5);
  \node[p,label=left:{\scriptsize$q\succ_{\hat{\beta}_1^u}\beta_1^u$}] at (1,3)  (p1) {$\hat{\beta}_1^u$};
  \node[p,label=left:{\scriptsize$q'\succ_{\hat{\beta}_2^u}\beta_1^u$}] at (1,2)  (p2) {$\hat{\beta}_2^u$};
  \node[p,label=left:{\scriptsize$\beta_1^u\succ_{\hat{\beta}_3^u}\beta_2^u$}] at (1,1)  (p3) {$\hat{\beta}_3^u$};
  \node[p,label=left:{\scriptsize$\beta_2^u\succ_{p}\cdots$}] at (1,0)  (pout) {$p$};

  \node[p,label=right:{\scriptsize$\cdots\succ_{q}\hat{\beta}_1^u$}] at (3,5)  (qin1) {$q$};
  \node[p,label=right:{\scriptsize$\cdots\succ_{q'}\hat{\beta}_2^u$}] at (3,4)  (qin2) {$q'$};
  \node[p,label=right:{\scriptsize$\hat{\beta}_1^u\succ_{\beta_1^u}\hat{\beta}_2^u\succ_{\beta_1^u}\hat{\beta}_3^u$}] at (3,3)  (q1) {$\beta_1^u$};
  \node[p,label=right:{\scriptsize$\hat{\beta}_3^u\succ_{\beta_2^u}p$}] at (3,1)  (q2) {$\beta_2^u$};

  \foreach \u/\v in {qin1/p1,qin2/p2,q1/p1,q1/p2,q1/p3,q2/p3,q2/pout} \draw (\u) -- (\v);
\end{tikzpicture}}
\caption{OR-gate simulator $\OR^u$}\label{fig:OR}
\end{minipage}
\end{figure}
\begin{figure}[t]
\begin{minipage}[b]{.5\textwidth}
\centering  
\scalebox{.9}{%
\begin{tikzpicture}[xscale=0.8,yscale=0.8,
    p/.style={circle,inner sep=0pt,draw=black,font=\scriptsize,fill=white,minimum size=13pt}]
  \draw[fill=black!10] (-2.8,0.5) rectangle (6.8,5.5);

  \node[p,label=left:{\scriptsize$q\succ_{\hat{\gamma}_1^u}\gamma_1^u$}] at (1,5)  (p1) {$\hat{\gamma}_1^u$};
  \node[p,label=left:{\scriptsize$\gamma_3^u\succ_{\hat{\gamma}_2^u}\gamma_4^u$}] at (1,4)  (p2) {$\hat{\gamma}_2^u$};
  \node[p,label=left:{\scriptsize$\gamma_2^u$}] at (1,3)  (p3) {$\hat{\gamma}_3^u$};
  \node[p,label=left:{\scriptsize$\gamma_1^u\succ_{\hat{\gamma}_4^u}\gamma_2^u\succ_{\hat{\gamma}_4^u}\gamma_3^u$}] at (1,2)  (p4) {$\hat{\gamma}_4^u$};
  \node[p,label=left:{\scriptsize$\gamma_3^u\succ_{\hat{\gamma}_5^u}\gamma_1^u$}] at (1,1)  (p5) {$\hat{\gamma}_5^u$};
  \node[p,label=left:{\scriptsize$\gamma_4^u\succ_{p}\cdots$}] at (1,0)  (p) {$p$};

  \node[p,label=right:{\scriptsize$\cdots\succ_{q}\hat{\gamma}_1^u$}] at (3,6)  (q) {$q$};
  \node[p,label=right:{\scriptsize$\hat{\gamma}_1^u\succ_{\gamma_1^u}\hat{\gamma}_5^u\succ_{\gamma_1^u}\hat{\gamma}_4^u$}] at (3,4)  (q1) {$\gamma_1^u$};
  \node[p,label=right:{\scriptsize$\hat{\gamma}_4^u\succ_{\gamma_2^u}\hat{\gamma}_3^u$}] at (3,3)  (q2) {$\gamma_2^u$};
  \node[p,label=right:{\scriptsize$\hat{\gamma}_4^u\succ_{\gamma_3^u}\hat{\gamma}_2^u\succ_{\gamma_3^u}\hat{\gamma}_5^u$}] at (3,2)  (q3) {$\gamma_3^u$};
  \node[p,label=right:{\scriptsize$\hat{\gamma}_2^u\succ_{\gamma_4^u}p$}] at (3,1)  (q4) {$\gamma_4^u$};
  \foreach \u/\v in {q/p1,q1/p1,q1/p5,q1/p4,q2/p4,q2/p3,q3/p4,q3/p2,q3/p5,q4/p2,q4/p} \draw (\u) -- (\v);
\end{tikzpicture}}
\caption{NOT-gate simulator $\NOT^u$}\label{fig:NOT}
\end{minipage}%
\begin{minipage}[b]{.5\textwidth}
\centering  
\scalebox{.9}{%
\begin{tikzpicture}[xscale=0.77,yscale=0.8,
    p/.style={circle,inner sep=0pt,draw=black,font=\scriptsize,fill=white,minimum size=13pt}]
  \draw[fill=black!10] (-2.8,1.5) rectangle (6.5,6.5);

  \node[p,label=left:{\scriptsize$\delta_1^u\succ_{p}\cdots$}] at (1,1)  (pout1) {$p$};
  \node[p,label=left:{\scriptsize$\delta_2^u\succ_{p'}\cdots$}] at (1,0)  (pout2) {$p'$};
  \node[p,label=left:{\scriptsize$q\succ_{\hat{\delta}_1^u}\delta_1^u$}] at (1,6)  (p1) {$\hat{\delta}_1^u$};
  \node[p,label=left:{\scriptsize$\delta_5^u\succ_{\hat{\delta}_2^u}\delta_2^u$}] at (1,5)  (p2) {$\hat{\delta}_2^u$};
  \node[p,label=left:{\scriptsize$\delta_3^u$}] at (1,4)  (p3) {$\hat{\delta}_3^u$};
  \node[p,label=left:{\scriptsize$\delta_5^u\succ_{\hat{\delta}_4^u}\delta_3^u\succ_{\hat{\delta}_4^u}\delta_4^u$}] at (1,3)  (p4) {$\hat{\delta}_4^u$};
  \node[p,label=left:{\scriptsize$\delta_4^u\succ_{\hat{\delta}_5^u}\delta_1^u\succ_{\hat{\delta}_5^u}\delta_5^u$}] at (1,2)  (p5) {$\hat{\delta}_5^u$};

  \node[p,label=right:{\scriptsize$\cdots\succ_{q}\hat{\delta}_1^u$}] at (3,7)  (q) {$q$};
  \node[p,label=right:{\scriptsize$\hat{\delta}_1^u\succ_{\delta_1^u}\hat{\delta}_5^u\succ_{\delta_1^u}p$}] at (3,6)  (q1) {$\delta_1^u$};
  \node[p,label=right:{\scriptsize$\hat{\delta}_2^u\succ_{\delta_2^u}p'$}] at (3,5)  (q2) {$\delta_2^u$};
  \node[p,label=right:{\scriptsize$\hat{\delta}_4^u\succ_{\delta_3^u}\hat{\delta}_3^u$}] at (3,4)  (q3) {$\delta_3^u$};
  \node[p,label=right:{\scriptsize$\hat{\delta}_4^u\succ_{\delta_4^u}\hat{\delta}_5^u$}] at (3,3)  (q4) {$\delta_4^u$};
  \node[p,label=right:{\scriptsize$\hat{\delta}_5^u\succ_{\delta_5^u}\hat{\delta}_2^u\succ_{\delta_5^u}\hat{\delta}_4^u$}] at (3,2)  (q5) {$\delta_5^u$};
  \foreach \u/\v in {q/p1,q1/p1,q1/p5,q2/p2,q3/p4,q3/p3,q4/p4,q4/p5,q5/p5,q5/p2,q5/p4,q1/pout1,q2/pout2} \draw (\u) -- (\v);
\end{tikzpicture}}
\caption{Branching simulator $\BRANCHING^u$}\label{fig:BRANCHING}
\end{minipage}
\end{figure}

$\AND^u$ is a piece of a sequential matching game which consists of three firms $P[\AND^u]=\{\hat{\alpha}_1^u,\hat{\alpha}_2^u,\hat{\alpha}_3^u\}$ and three workers $Q[\AND^u]=\{\alpha_1^u,\alpha_2^u,\alpha_3^u\}$ with two input workers $q$ and $q'$ and one output firm $p$.
The set of acceptable pairs related to $\AND^u$ and the preferences of $P[\AND^u]\cup Q[\AND^u]$ are defined as Fig.~\ref{fig:AND}.
Let $\argmax\{\sigma^{-1}(e)\mid e\in \delta_I(q)\}=(\hat{\alpha}_1^u,q)$, $\argmax\{\sigma^{-1}(e)\mid e\in \delta_I(q')\}=(\hat{\alpha}_2^u,q)$, and $\pi^{-1}(\hat{\alpha}_1^u)<\pi^{-1}(\hat{\alpha}_2^u)<\pi^{-1}(\hat{\alpha}_3^u)<\pi^{-1}(p)$. 
Let us consider a subgame where the first firm is $\hat{\alpha}_1^u$.
If both $q$ and $q'$ are already matched, then in the SPE matching for this subgame,
$\hat{\alpha}_i^u$ is matched with $\alpha_i^u$ for each $i = 1, 2, 3$,
and hence $p$ is not matched with $\alpha_3^u$.
On the other hand, if $q$ or $q'$ is not matched yet, then $p$ is matched with $\alpha_3^u$ in the SPE matching for this subgame.
The input-output relation is summarized in Fig.~\ref{fig:AND_IO}.
Thus, we can see $\AND^u$ as an AND-gate simulator.

\begin{figure}[t]
\begin{minipage}{.25\textwidth}
\centering  
\begin{tikzpicture}[xscale=0.75,yscale=0.75,
    p/.style={circle,inner sep=0pt,draw=black,font=\scriptsize,fill=white,minimum size=13pt}]
  \draw[fill=black!10] (0,0.5) rectangle (4,3.5);
  \node[p] at (1,3)  (p1) {$\hat{\alpha}_1^u$};
  \node[p] at (1,2)  (p2) {$\hat{\alpha}_2^u$};
  \node[p] at (1,1)  (p3) {$\hat{\alpha}_3^u$};
  \node[p,label=left:{\tiny\TRUE}] at (1,0)  (pout) {$p$};
  \node at (3,0)  (poutd) {};

  \node[p,label=right:{\tiny\TRUE}] at (3,5)  (qin1) {$q$};
  \node at (1,5.5)  (qin1d) {};
  \node[p,label=right:{\tiny\TRUE}] at (3,4)  (qin2) {$q'$};
  \node at (1,4.5)  (qin2d) {};
  \node[p] at (3,3)  (q1) {$\alpha_1^u$};
  \node[p] at (3,2)  (q2) {$\alpha_2^u$};
  \node[p] at (3,1)  (q3) {$\alpha_3^u$};

  \foreach \u/\v in {qin1/p1,qin2/p2,q1/p1,q1/p3,q2/p2,q2/p3,q3/p3,q3/pout} \draw (\u) -- (\v);
  \foreach \u/\v in {qin1/qin1d,qin2/qin2d,p1/q1,p2/q2,p3/q3} \draw[line width=3pt,red,opacity=.5] (\u) -- (\v);
\end{tikzpicture}
\end{minipage}%
\begin{minipage}{.25\textwidth}
\centering  
\begin{tikzpicture}[xscale=0.75,yscale=0.75,
    p/.style={circle,inner sep=0pt,draw=black,font=\scriptsize,fill=white,minimum size=13pt}]
  \draw[fill=black!10] (0,0.5) rectangle (4,3.5);
  \node[p] at (1,3)  (p1) {$\hat{\alpha}_1^u$};
  \node[p] at (1,2)  (p2) {$\hat{\alpha}_2^u$};
  \node[p] at (1,1)  (p3) {$\hat{\alpha}_3^u$};
  \node[p,label=left:{\tiny\FALSE}] at (1,0)  (pout) {$p$};
  \node at (3,0)  (poutd) {};

  \node[p,label=right:{\tiny\TRUE}] at (3,5)  (qin1) {$q$};
  \node at (1,5.5)  (qin1d) {};
  \node[p,label=right:{\tiny\FALSE}] at (3,4)  (qin2) {$q'$};
  \node at (1,4.5)  (qin2d) {};
  \node[p] at (3,3)  (q1) {$\alpha_1^u$};
  \node[p] at (3,2)  (q2) {$\alpha_2^u$};
  \node[p] at (3,1)  (q3) {$\alpha_3^u$};

  \foreach \u/\v in {qin1/p1,qin2/p2,q1/p1,q1/p3,q2/p2,q2/p3,q3/p3,q3/pout} \draw (\u) -- (\v);
  \foreach \u/\v in {qin1/qin1d,qin2/p2,p1/q1,p3/q2,pout/q3} \draw[line width=3pt,red,opacity=.5] (\u) -- (\v);
\end{tikzpicture}
\end{minipage}%
\begin{minipage}{.25\textwidth}
\centering  
\begin{tikzpicture}[xscale=0.75,yscale=0.75,
    p/.style={circle,inner sep=0pt,draw=black,font=\scriptsize,fill=white,minimum size=13pt}]
  \draw[fill=black!10] (0,0.5) rectangle (4,3.5);
  \node[p] at (1,3)  (p1) {$\hat{\alpha}_1^u$};
  \node[p] at (1,2)  (p2) {$\hat{\alpha}_2^u$};
  \node[p] at (1,1)  (p3) {$\hat{\alpha}_3^u$};
  \node[p,label=left:{\tiny\FALSE}] at (1,0)  (pout) {$p$};
  \node at (3,0)  (poutd) {};

  \node[p,label=right:{\tiny\FALSE}] at (3,5)  (qin1) {$q$};
  \node at (1,5.5)  (qin1d) {};
  \node[p,label=right:{\tiny\TRUE}] at (3,4)  (qin2) {$q'$};
  \node at (1,4.5)  (qin2d) {};
  \node[p] at (3,3)  (q1) {$\alpha_1^u$};
  \node[p] at (3,2)  (q2) {$\alpha_2^u$};
  \node[p] at (3,1)  (q3) {$\alpha_3^u$};

  \foreach \u/\v in {qin1/p1,qin2/p2,q1/p1,q1/p3,q2/p2,q2/p3,q3/p3,q3/pout} \draw (\u) -- (\v);
  \foreach \u/\v in {qin1/p1,qin2/qin2d,p2/q2,p3/q1,pout/q3} \draw[line width=3pt,red,opacity=.5] (\u) -- (\v);
\end{tikzpicture}
\end{minipage}%
\begin{minipage}{.25\textwidth}
\centering  
\begin{tikzpicture}[xscale=0.75,yscale=0.75,
    p/.style={circle,inner sep=0pt,draw=black,font=\scriptsize,fill=white,minimum size=13pt}]
  \draw[fill=black!10] (0,0.5) rectangle (4,3.5);
  \node[p] at (1,3)  (p1) {$\hat{\alpha}_1^u$};
  \node[p] at (1,2)  (p2) {$\hat{\alpha}_2^u$};
  \node[p] at (1,1)  (p3) {$\hat{\alpha}_3^u$};
  \node[p,label=left:{\tiny\FALSE}] at (1,0)  (pout) {$p$};
  \node at (3,0)  (poutd) {};

  \node[p,label=right:{\tiny\FALSE}] at (3,5)  (qin1) {$q$};
  \node at (1,5.5)  (qin1d) {};
  \node[p,label=right:{\tiny\FALSE}] at (3,4)  (qin2) {$q'$};
  \node at (1,4.5)  (qin2d) {};
  \node[p] at (3,3)  (q1) {$\alpha_1^u$};
  \node[p] at (3,2)  (q2) {$\alpha_2^u$};
  \node[p] at (3,1)  (q3) {$\alpha_3^u$};

  \foreach \u/\v in {qin1/p1,qin2/p2,q1/p1,q1/p3,q2/p2,q2/p3,q3/p3,q3/pout} \draw (\u) -- (\v);
  \foreach \u/\v in {p1/qin1,p2/qin2,p3/q1,pout/q3} \draw[line width=3pt,red,opacity=.5] (\u) -- (\v);
\end{tikzpicture}
\end{minipage}%
\caption{Input-output relation of $\AND^u$}\label{fig:AND_IO}
\end{figure}
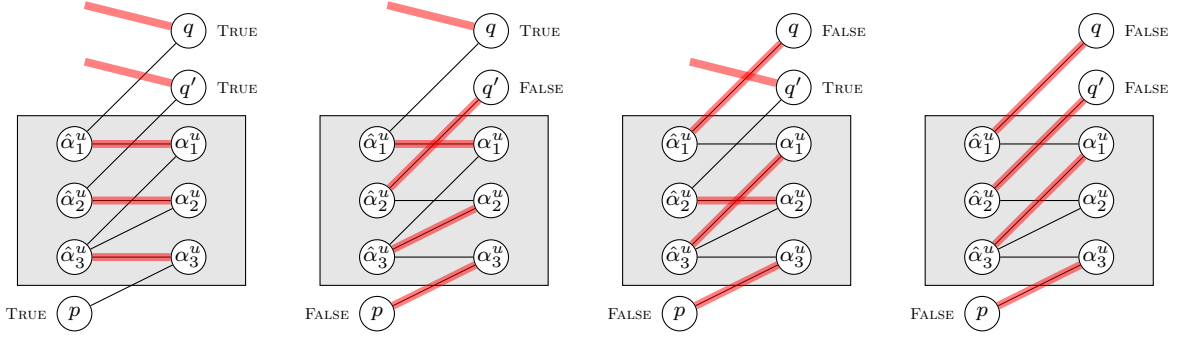

$\OR^u$ is a piece of a sequential matching game which consists of three firms $P[\OR^u]=\{\hat{\beta}_1^u,\hat{\beta}_2^u,\hat{\beta}_3^u\}$ and two workers $Q[\OR^u]=\{\beta_1^u,\beta_2^u\}$ with two input workers $q$ and $q'$ and one output firm $p$.
The set of acceptable pairs related to $\OR^u$ and the preferences of $P[\OR^u]\cup Q[\OR^u]$ are defined as Fig.~\ref{fig:OR}.
Let $\argmax\{\sigma^{-1}(e)\mid e\in \delta_I(q)\}=(\hat{\beta}_1^u,q)$, $\argmax\{\sigma^{-1}(e)\mid e\in \delta_I(q')\}=(\hat{\beta}_2^u,q)$, and
$\pi^{-1}(\hat{\beta}_1^u)<\pi^{-1}(\hat{\beta}_2^u)<\pi^{-1}(\hat{\beta}_3^u)<\pi^{-1}(p)$.
Let us consider a subgame where the first firm is $\hat{\beta}_1^u$.
If $q$ or $q'$ is already matched, then in the SPE matching for this subgame, $\beta_1^u$ is matched with $\hat{\beta}_1^u$ or $\hat{\beta}_2^u$, and $\beta_2^u$ is matched with $\hat{\beta}_3^u$. Hence $p$ is not matched with $\beta_2^u$.
On the other hand, if neither $q$ nor $q'$ is matched yet, then the SPE matching for this subgame contains $\{(\hat{\beta}_3^u,\beta_1^u),(p,\beta_2^u)\}$.
The input-output relation is summarized in Fig.~\ref{fig:OR_IO}.
Thus, we can see $\OR^u$ as an OR-gate simulator.
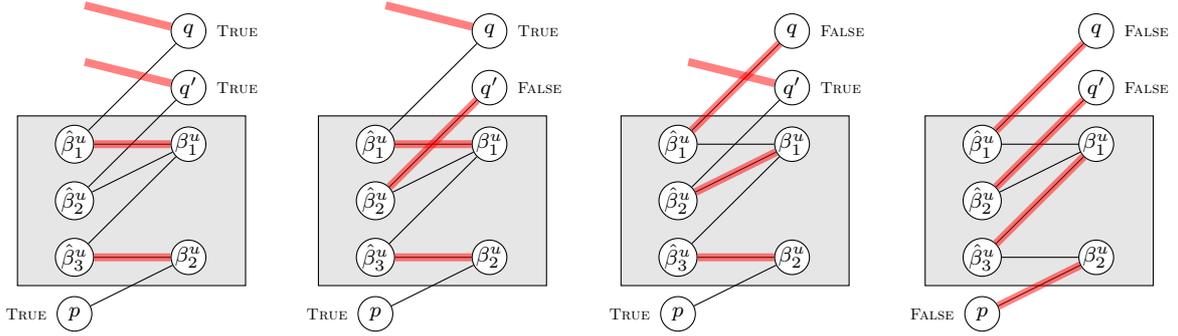
\begin{figure}[t]
\begin{minipage}{.25\textwidth}
\centering  
\begin{tikzpicture}[xscale=0.75,yscale=0.75,
    p/.style={circle,inner sep=0pt,draw=black,font=\scriptsize,fill=white,minimum size=13pt}]
  \draw[fill=black!10] (0,0.5) rectangle (4,3.5);
  \node[p] at (1,3)  (p1) {$\hat{\beta}_1^u$};
  \node[p] at (1,2)  (p2) {$\hat{\beta}_2^u$};
  \node[p] at (1,1)  (p3) {$\hat{\beta}_3^u$};
  \node[p,label=left:{\tiny\TRUE}] at (1,0)  (pout) {$p$};
  \node at (3,0)  (poutd) {};

  \node[p,label=right:{\tiny\TRUE}] at (3,5)  (qin1) {$q$};
  \node at (1,5.5)  (qin1d) {};
  \node[p,label=right:{\tiny\TRUE}] at (3,4)  (qin2) {$q'$};
  \node at (1,4.5)  (qin2d) {};
  \node[p] at (3,3)  (q1) {$\beta_1^u$};
  \node[p] at (3,1)  (q2) {$\beta_2^u$};

  \foreach \u/\v in {qin1/p1,qin2/p2,q1/p1,q1/p2,q1/p3,q2/p3,q2/pout} \draw (\u) -- (\v);
  \foreach \u/\v in {qin1/qin1d,qin2/qin2d,p1/q1,p3/q2} \draw[line width=3pt,red,opacity=.5] (\u) -- (\v);
\end{tikzpicture}
\end{minipage}%
\begin{minipage}{.25\textwidth}
\centering  
\begin{tikzpicture}[xscale=0.75,yscale=0.75,
    p/.style={circle,inner sep=0pt,draw=black,font=\scriptsize,fill=white,minimum size=13pt}]
  \draw[fill=black!10] (0,0.5) rectangle (4,3.5);
  \node[p] at (1,3)  (p1) {$\hat{\beta}_1^u$};
  \node[p] at (1,2)  (p2) {$\hat{\beta}_2^u$};
  \node[p] at (1,1)  (p3) {$\hat{\beta}_3^u$};
  \node[p,label=left:{\tiny\TRUE}] at (1,0)  (pout) {$p$};
  \node at (3,0)  (poutd) {};

  \node[p,label=right:{\tiny\TRUE}] at (3,5)  (qin1) {$q$};
  \node at (1,5.5)  (qin1d) {};
  \node[p,label=right:{\tiny\FALSE}] at (3,4)  (qin2) {$q'$};
  \node at (1,4.5)  (qin2d) {};
  \node[p] at (3,3)  (q1) {$\beta_1^u$};
  \node[p] at (3,1)  (q2) {$\beta_2^u$};

  \foreach \u/\v in {qin1/p1,qin2/p2,q1/p1,q1/p2,q1/p3,q2/p3,q2/pout} \draw (\u) -- (\v);
  \foreach \u/\v in {qin1/qin1d,qin2/p2,p1/q1,p3/q2} \draw[line width=3pt,red,opacity=.5] (\u) -- (\v);
\end{tikzpicture}
\end{minipage}%
\begin{minipage}{.25\textwidth}
\centering  
\begin{tikzpicture}[xscale=0.75,yscale=0.75,
    p/.style={circle,inner sep=0pt,draw=black,font=\scriptsize,fill=white,minimum size=13pt}]
  \draw[fill=black!10] (0,0.5) rectangle (4,3.5);
  \node[p] at (1,3)  (p1) {$\hat{\beta}_1^u$};
  \node[p] at (1,2)  (p2) {$\hat{\beta}_2^u$};
  \node[p] at (1,1)  (p3) {$\hat{\beta}_3^u$};
  \node[p,label=left:{\tiny\TRUE}] at (1,0)  (pout) {$p$};
  \node at (3,0)  (poutd) {};

  \node[p,label=right:{\tiny\FALSE}] at (3,5)  (qin1) {$q$};
  \node at (1,5.5)  (qin1d) {};
  \node[p,label=right:{\tiny\TRUE}] at (3,4)  (qin2) {$q'$};
  \node at (1,4.5)  (qin2d) {};
  \node[p] at (3,3)  (q1) {$\beta_1^u$};
  \node[p] at (3,1)  (q2) {$\beta_2^u$};

  \foreach \u/\v in {qin1/p1,qin2/p2,q1/p1,q1/p2,q1/p3,q2/p3,q2/pout} \draw (\u) -- (\v);
  \foreach \u/\v in {qin1/p1,qin2/qin2d,p2/q1,p3/q2} \draw[line width=3pt,red,opacity=.5] (\u) -- (\v);
\end{tikzpicture}
\end{minipage}%
\begin{minipage}{.25\textwidth}
\centering  
\begin{tikzpicture}[xscale=0.75,yscale=0.75,
    p/.style={circle,inner sep=0pt,draw=black,font=\scriptsize,fill=white,minimum size=13pt}]
  \draw[fill=black!10] (0,0.5) rectangle (4,3.5);
  \node[p] at (1,3)  (p1) {$\hat{\beta}_1^u$};
  \node[p] at (1,2)  (p2) {$\hat{\beta}_2^u$};
  \node[p] at (1,1)  (p3) {$\hat{\beta}_3^u$};
  \node[p,label=left:{\tiny\FALSE}] at (1,0)  (pout) {$p$};
  \node at (3,0)  (poutd) {};

  \node[p,label=right:{\tiny\FALSE}] at (3,5)  (qin1) {$q$};
  \node at (1,5.5)  (qin1d) {};
  \node[p,label=right:{\tiny\FALSE}] at (3,4)  (qin2) {$q'$};
  \node at (1,4.5)  (qin2d) {};
  \node[p] at (3,3)  (q1) {$\beta_1^u$};
  \node[p] at (3,1)  (q2) {$\beta_2^u$};

  \foreach \u/\v in {qin1/p1,qin2/p2,q1/p1,q1/p2,q1/p3,q2/p3,q2/pout} \draw (\u) -- (\v);
  \foreach \u/\v in {qin1/p1,qin2/p2,p3/q1,pout/q2} \draw[line width=3pt,red,opacity=.5] (\u) -- (\v);
\end{tikzpicture}
\end{minipage}%
\caption{Input-output relation of $\OR^u$}\label{fig:OR_IO}
\end{figure}

$\NOT^u$ is a piece of a sequential matching game which consists of five firms $P[\NOT^u]=\{\hat{\gamma}_1^u,\dots,\hat{\gamma}_5^u\}$ and three workers $Q[\NOT^u]=\{\gamma_1^u,\dots,\gamma_4^u\}$ with one input worker $q$ and one output firm $p$.
The set of acceptable pairs related to $\NOT^u$ 
and the preferences of $P[\NOT^u]\cup Q[\NOT^u]$ are defined as Fig.~\ref{fig:NOT}.
Let $\argmax\{\sigma^{-1}(e)\mid e\in \delta_I(q)\}=(\hat{\gamma}_1^u,q)$ and
$\pi^{-1}(\hat{\gamma}_1^u)<\dots<\pi^{-1}(\hat{\gamma}_5^u)<\pi^{-1}(p)$.
Let us consider a subgame $(I',\sigma')$ where the first firm is $\hat{\gamma}_1^u$.
If $q$ is already matched, then $\spe(I',\sigma')\cap E[\NOT^u]=\{(\hat{\gamma}_1^u,\gamma_1^u),\,(\hat{\gamma}_2^u,\gamma_3^u),\,(\hat{\gamma}_4^u,\gamma_2^u),\,(p,\gamma_4^u)\}$.
Otherwise, i.e., $q$ is not matched yet, $\spe(I',\sigma')\cap E[\NOT^u]=\{(\hat{\gamma}_1^u,q),\,(\hat{\gamma}_2^u,\gamma_4^u),\,(\hat{\gamma}_3^u,\gamma_2^u),\,(\hat{\gamma}_4^u,\gamma_3^u),\,(\hat{\gamma}_5^u,\gamma_1^u)\}$.
The input-output relation is summarized in Fig.~\ref{fig:NOT_IO}.
Thus, we can see $\NOT^u$ as a NOT-gate simulator.

$\BRANCHING^u$ is a piece of a sequential matching game which consists of five firms $P[\BRANCHING^u]\allowbreak{}=\{\hat{\delta}_1^u,\dots,\hat{\delta}_5^u\}$ and five workers $Q[\BRANCHING^u]=\{\delta_1^u,\dots,\delta_5^u\}$ with one input worker $q$ and two output firms $p$ and $p'$.
The set of acceptable pairs related to $\BRANCHING^u$ 
and the preferences of $P[\BRANCHING^u]\cup Q[\BRANCHING^u]$ are defined as Fig.~\ref{fig:BRANCHING}.
Let $\argmax\{\sigma^{-1}(e)\mid e\in \delta_I(q)\}=(\hat{\delta}_1^u,q)$ and
$\pi^{-1}(\hat{\delta}_1^u)<\dots<\pi^{-1}(\hat{\delta}_5^u)<\pi^{-1}(p)<\pi^{-1}(p')$.
Let us consider a subgame $(I',\sigma')$ where the first firm is $\hat{\delta}_1^u$.
If $q$ is already matched, then $\spe(I',\sigma')\cap E[\BRANCHING^u]=\{(\hat{\delta}_i^u, \delta_i^u) \mid i = 1, 2, 3, 4, 5 \}$.
Otherwise, i.e., $q$ is not matched yet, $\spe(I',\sigma')\cap E[\BRANCHING^u]=\{(\hat{\delta}_1^u,q),\,(\hat{\delta}_2^u,\delta_5^u),\,(\hat{\delta}_4^u,\delta_3^u),\,(\hat{\delta}_5^u,\delta_4^u),\,(p,\delta_1^u),\,(p',\delta_2^u)\}$.
The input-output relation is summarized in Fig.~\ref{fig:BRANCHING_IO}.
Thus, we can see $\BRANCHING^u$ as a branching simulator.

\begin{figure}
\begin{minipage}[b]{.5\textwidth}
\begin{minipage}{.5\textwidth}
\centering  
\begin{tikzpicture}[xscale=0.75,yscale=0.75,
    p/.style={circle,inner sep=0pt,draw=black,font=\scriptsize,fill=white,minimum size=13pt}]
  \draw[fill=black!10] (0,0.5) rectangle (4,5.5);

  \node[p] at (1,5)  (p1) {$\hat{\gamma}_1^u$};
  \node[p] at (1,4)  (p2) {$\hat{\gamma}_2^u$};
  \node[p] at (1,3)  (p3) {$\hat{\gamma}_3^u$};
  \node[p] at (1,2)  (p4) {$\hat{\gamma}_4^u$};
  \node[p] at (1,1)  (p5) {$\hat{\gamma}_5^u$};
  \node[p,label=left:{\tiny\FALSE}] at (1,0)  (p) {$p$};

  \node[p,label=right:{\tiny\TRUE}] at (3,6)  (q) {$q$};
  \node[] at (1,6.5)  (qd) {};
  \node[p] at (3,4)  (q1) {$\gamma_1^u$};
  \node[p] at (3,3)  (q2) {$\gamma_2^u$};
  \node[p] at (3,2)  (q3) {$\gamma_3^u$};
  \node[p] at (3,1)  (q4) {$\gamma_4^u$};
  \foreach \u/\v in {q/p1,q1/p1,q1/p5,q1/p4,q2/p4,q2/p3,q3/p4,q3/p2,q3/p5,q4/p2,q4/p} \draw (\u) -- (\v);
  \foreach \u/\v in {q/qd,p1/q1,p2/q3,p4/q2,p/q4} \draw[line width=3pt,red,opacity=.5] (\u) -- (\v);
\end{tikzpicture}
\end{minipage}%
\begin{minipage}{.5\textwidth}
\centering  
\begin{tikzpicture}[xscale=0.75,yscale=0.75,
    p/.style={circle,inner sep=0pt,draw=black,font=\scriptsize,fill=white,minimum size=13pt}]
  \draw[fill=black!10] (0,0.5) rectangle (4,5.5);

  \node[p] at (1,5)  (p1) {$\hat{\gamma}_1^u$};
  \node[p] at (1,4)  (p2) {$\hat{\gamma}_2^u$};
  \node[p] at (1,3)  (p3) {$\hat{\gamma}_3^u$};
  \node[p] at (1,2)  (p4) {$\hat{\gamma}_4^u$};
  \node[p] at (1,1)  (p5) {$\hat{\gamma}_5^u$};
  \node[p,label=left:{\tiny\TRUE}] at (1,0)  (p) {$p$};

  \node[p,label=right:{\tiny\FALSE}] at (3,6)  (q) {$q$};
  \node[] at (1,6.5)  (qd) {};
  \node[p] at (3,4)  (q1) {$\gamma_1^u$};
  \node[p] at (3,3)  (q2) {$\gamma_2^u$};
  \node[p] at (3,2)  (q3) {$\gamma_3^u$};
  \node[p] at (3,1)  (q4) {$\gamma_4^u$};
  \foreach \u/\v in {q/p1,q1/p1,q1/p5,q1/p4,q2/p4,q2/p3,q3/p4,q3/p2,q3/p5,q4/p2,q4/p} \draw (\u) -- (\v);
  \foreach \u/\v in {p1/q,p2/q4,p3/q2,p4/q3,p5/q1} \draw[line width=3pt,red,opacity=.5] (\u) -- (\v);
\end{tikzpicture}
\end{minipage}%
\caption{Input-output relation of $\NOT^u$}\label{fig:NOT_IO}
\end{minipage}%
\begin{minipage}[b]{.5\textwidth}
\begin{minipage}{.5\textwidth}
\centering  
\begin{tikzpicture}[xscale=0.75,yscale=0.75,
    p/.style={circle,inner sep=0pt,draw=black,font=\scriptsize,fill=white,minimum size=13pt}]
  \draw[fill=black!10] (0,1.5) rectangle (4,6.5);

  \node[p,label=left:{\tiny\TRUE}] at (1,1)  (pout1) {$p$};
  \node[p,label=left:{\tiny\TRUE}] at (1,0)  (pout2) {$p'$};
  \node[p] at (1,6)  (p1) {$\hat{\delta}_1^u$};
  \node[p] at (1,5)  (p2) {$\hat{\delta}_2^u$};
  \node[p] at (1,4)  (p3) {$\hat{\delta}_3^u$};
  \node[p] at (1,3)  (p4) {$\hat{\delta}_4^u$};
  \node[p] at (1,2)  (p5) {$\hat{\delta}_5^u$};

  \node[p,label=right:{\tiny\TRUE}] at (3,7)  (q) {$q$};
  \node at (1,7.5)  (qd) {};
  \node[p] at (3,6)  (q1) {$\delta_1^u$};
  \node[p] at (3,5)  (q2) {$\delta_2^u$};
  \node[p] at (3,4)  (q3) {$\delta_3^u$};
  \node[p] at (3,3)  (q4) {$\delta_4^u$};
  \node[p] at (3,2)  (q5) {$\delta_5^u$};
  \foreach \u/\v in {q/p1,q1/p1,q1/p5,q2/p2,q3/p4,q3/p3,q4/p4,q4/p5,q5/p5,q5/p2,q5/p4,q1/pout1,q2/pout2} \draw (\u) -- (\v);
  \foreach \u/\v in {qd/q,p1/q1,p2/q2,p3/q3,p4/q4,p5/q5} \draw[line width=3pt,red,opacity=.5] (\u) -- (\v);
\end{tikzpicture}%
\end{minipage}%
\begin{minipage}{.5\textwidth}
\centering  
\begin{tikzpicture}[xscale=0.75,yscale=0.75,
    p/.style={circle,inner sep=0pt,draw=black,font=\scriptsize,fill=white,minimum size=13pt}]
  \draw[fill=black!10] (0,1.5) rectangle (4,6.5);

  \node[p,label=left:{\tiny\FALSE}] at (1,1)  (pout1) {$p$};
  \node[p,label=left:{\tiny\FALSE}] at (1,0)  (pout2) {$p'$};
  \node[p] at (1,6)  (p1) {$\hat{\delta}_1^u$};
  \node[p] at (1,5)  (p2) {$\hat{\delta}_2^u$};
  \node[p] at (1,4)  (p3) {$\hat{\delta}_3^u$};
  \node[p] at (1,3)  (p4) {$\hat{\delta}_4^u$};
  \node[p] at (1,2)  (p5) {$\hat{\delta}_5^u$};

  \node[p,label=right:{\tiny\FALSE}] at (3,7)  (q) {$q$};
  \node at (1,7.5)  (qd) {};
  \node[p] at (3,6)  (q1) {$\delta_1^u$};
  \node[p] at (3,5)  (q2) {$\delta_2^u$};
  \node[p] at (3,4)  (q3) {$\delta_3^u$};
  \node[p] at (3,3)  (q4) {$\delta_4^u$};
  \node[p] at (3,2)  (q5) {$\delta_5^u$};
  \foreach \u/\v in {q/p1,q1/p1,q1/p5,q2/p2,q3/p4,q3/p3,q4/p4,q4/p5,q5/p5,q5/p2,q5/p4,q1/pout1,q2/pout2} \draw (\u) -- (\v);
  \foreach \u/\v in {p1/q,p2/q5,p4/q3,p5/q4,pout1/q1,pout2/q2} \draw[line width=3pt,red,opacity=.5] (\u) -- (\v);
\end{tikzpicture}
\end{minipage}%
\caption{Input-output relation of $\BRANCHING^u$}\label{fig:BRANCHING_IO}
\end{minipage}
\end{figure}

Now, we are ready to design the evaluation phase required in Claim~\ref{claim:hoge}.

\begin{proof}[Proof of Claim~\ref{claim:hoge}]
Let $\mathcal{C}=\{C_1,\dots,C_m\}$ be the set of clauses for the given 3-CNF\@.
For each clause $C_j\in\mathcal{C}$, let $C_j=\ell_{j,1}\vee \ell_{j,2}\vee \ell_{j,3}$ such that $\ell_{j,k}\in \{x_{\lambda(j,k)},\bar{x}_{\lambda(j,k)}\}$ $(k=1,2,3)$.
Also, let $N^+=\{(j,k)\mid \ell_{j,k}=x_{\lambda(j,k)},~k\in\{1,2,3\}\}$ and $N^-=\{(j,k)\mid \ell_{j,k}=\bar{x}_{\lambda(j,k)},~k\in\{1,2,3\}\}$.
Without loss of generality, we may assume that $\lambda(j,1)<\lambda(j,2)<\lambda(j,3)$.
We construct a $(3,3)$-\SPEM instance $(I,\sigma)$ as follows.
In the instance, we use
$\AND^u$ for $u\in U_{\AND}\coloneq\{1,\dots,m\}$,
$\OR^u$ for $u\in U_{\OR}\coloneq\bigcup_{j=1}^m\{(j,1),(j,2)\}$,
$\NOT^u$ for $u\in U_{\NOT}\coloneq N^-\cup J^\exists$, and
$\BRANCHING^u$ for $u\in U_{\BRANCHING}\coloneq \bigcup_{j=1}^{m-1}\{(1,j),\dots,(n,j)\}\cup\{1,\dots,n-1\}$.
The set of firms and workers are respectively defined as
\begin{align*}
  P&\coloneq\{s_1,\dots,s_n,t_1,\dots,t_n,r_1,\dots,r_n,\hat{c}_1,\dots,\hat{c}_m\}\cup\bigcup_{i=1}^n\bigcup_{j=1}^m\{\hat{y}_{i,j}\}
  \cup\bigcup_{j=1}^m \{\hat{\ell}_{j,1},\hat{\ell}_{j,2},\hat{\ell}_{j,3}\}\\
  &\cup\bigcup_{u\in U_{\OR}}P[\OR^u]
  \cup\bigcup_{u\in U_{\AND}}P[\AND^u]
  \cup\bigcup_{u\in U_{\NOT}}P[\NOT^u]
  \cup\bigcup_{u\in U_{\BRANCHING}}P[\BRANCHING^u]
\end{align*}
and
\begin{align*}
  Q&\coloneq\{x_1,\dots,x_n,\bar{x}_1,\dots,\bar{x}_n,z_1,\dots,z_n,c_1,\dots,c_m\}\cup\bigcup_{i=1}^n\bigcup_{j=1}^m\{y_{i,j}\}
  \cup\bigcup_{j=1}^m \{\ell_{j,1},\ell_{j,2},\ell_{j,3}\}\\
  &\cup\bigcup_{u\in U_{\OR}}Q[\OR^u]
  \cup\bigcup_{u\in U_{\AND}}Q[\AND^u]
  \cup\bigcup_{u\in U_{\NOT}}Q[\NOT^u]
  \cup\bigcup_{u\in U_{\BRANCHING}}Q[\BRANCHING^u].
\end{align*}
The input and output of each gadget are given as shown in Table~\ref{table:connection}.
Moreover, the preferences of the other firms and workers are given as follows:
\begin{center}
\begin{minipage}[t]{.1\linewidth}
\end{minipage}%
\begin{minipage}[t]{.45\linewidth}
\begin{description}
\setlength{\parskip}{0cm}
\setlength{\itemsep}{0cm}
\item[$s_i$:] $x_i\succ_{s_i}\bar{x}_i$ $(i=1,\dots,n)$,
\item[$\hat{\ell}_{j,k}$:] $y_{\lambda(j,k),j}\succ_{\hat{\ell}_{j,k}}\ell_{j,k}$ $((j,k)\in N^+)$,
\item[$\hat{\ell}_{j,k}$:] $\gamma_4^{(j,k)}\succ_{\hat{\ell}_{j,k}}\ell_{j,k}$ $((j,k)\in N^-)$,
\item[$\hat{y}_{i,j}$:] $\delta_1^{(i,j)}\succ_{\hat{y}_{i,j}}y_{i,j}$ $(\substack{i=1,\dots,n\\ j=1,\dots,m-1})$,
\item[$\hat{y}_{i,m}$:] $\delta_2^{(i,m-1)}\succ_{\hat{y}_{i,m}}y_{i,m}$ $(i=1,\dots,n)$,
\item[$\hat{c}_j$:] $\beta_2^{(j,2)}\succ_{\hat{c}_j}c_j$ $(j=1,\dots,m)$,
\item[$r_i$:] $\hat{\delta}_1^i\succ_{r_i}z_i$ $(i=1,\dots,n-1)$,
\item[$r_n$:] $\hat{\delta}_2^{n-1}\succ_{r_i}z_n$,
\item[$t_i$:] $z_i\succ_{t_i}x_i$ $(i=1,\dots,n)$,
\end{description}
\end{minipage}%
\begin{minipage}[t]{.45\linewidth}
\begin{description}
\setlength{\parskip}{0cm}
\setlength{\itemsep}{0cm}
\item[$x_i$:] $t_i\succ_{x_i}s_i$ $(i=1,\dots,n)$,
\item[$\bar{x}_i$:] $s_i\succ_{\bar{x}_i}\hat{\delta}_1^{(i,1)}$ $(i=1,\dots,n)$,
\item[$y_{i,j}$:] $\begin{cases}\hat{y}_{i,j}\succ_{y_{i,j}}\hat{\ell}_{j,k}&(\text{if }\ell_{j,k}=x_i),\\\hat{y}_{i,j}\succ_{y_{i,j}}\hat{\gamma}_1^{(j,k)}&(\text{if }\ell_{j,k}=\bar{x}_i),\\\hat{y}_{i,j}&(\text{otherwise}),\end{cases}$
\item[$c_1$:] $\hat{c}_1\succ_{c_1}\hat{\alpha}_1^1$,
\item[$c_j$:] $\hat{c}_j\succ_{c_j}\hat{\alpha}_2^{j-1}$ $(j=2,\dots,m)$,
\item[$z_i$:] $r_i\succ_{z_i}t_i$ $(i=1,\dots,n)$.
\end{description}
\end{minipage}
\end{center}
We define $E$ and $\succ$ according to the above.
Then, it is easy to check that the degree of each firm or worker is at most three.

\begin{table}[t]
  \caption{The input and output of the gadgets. }\label{table:connection}
  \centering
  \resizebox{1.0\textwidth}{!}{
  \begin{tabular}[htbp]{l|ll}\toprule
                                         & Input                     & Output\\\midrule
    $\BRANCHING^{(i,1)}~(i=1,\dots,n)$   & $\bar{x}_i$               & $\hat{y}_{i,1}$, $\hat{\delta}_1^{(i,2)}$\\
    $\BRANCHING^{(i,j)}~(\substack{i=1,\dots,n\\ j=2,\dots,m-2})$   & $\delta_2^{(i,j-1)}$     & $\hat{y}_{i,j}$, $\hat{\delta}_1^{(i,j+1)}$\\
    $\BRANCHING^{(i,m-1)}~(i=1,\dots,n)$ & $\delta_{2}^{(i,m-2)}$    & $\hat{y}_{i,m-1}$, $\hat{y}_{i,m}$\\
    $\NOT^{(j,k)}~((j,k)\in N^-)$        & $y_{\lambda(j,1),j}$         & $\ell_{j,k}$\\
    $\OR^{(j,1)}~(j=1,\dots,m)$          & $\ell_{j,1}$, $\ell_{j,2}$& $\hat{\beta}_1^{(j,2)}$\\
    $\OR^{(j,2)}~(j=1,\dots,m)$          & $\beta_2^{(j,1)}$, $\ell_{j,3}$ & $\hat{c}_j$\\
    $\AND^{1}$                           & $c_1$, $c_2$              & $\hat{\alpha}_1^{2}$\\
    $\AND^{j}~(j=2,\dots,m-2)$           & $\alpha_3^{j-1}$, $c_{j+1}$& $\hat{\alpha}_1^{j+1}$\\
    $\AND^{m-1}$                         & $\alpha_3^{m-2}$, $c_{m}$  & $\hat{\delta}_1^{1}$\\
    $\BRANCHING^{1}$                     & $\alpha_3^{m-1}$           & $\begin{cases}\hat{\gamma}_1^1&(\text{if }1\in A)\\r_1&(\text{if }1\not\in A)\end{cases}$, $\hat{\delta}_1^{2}$\\
    $\BRANCHING^{i}~(i=2,\dots,n-2)$     & $\delta_2^{i-1}$          & $\begin{cases}\hat{\gamma}_1^i&(\text{if }i\in A)\\r_i&(\text{if }i\not\in A)\end{cases}$, $\hat{\delta}_1^{i+1}$\\
    $\BRANCHING^{n-1}$                   & $\delta_2^{n-2}$          & $\begin{cases}\hat{\gamma}_1^{n-1}&(\text{if }n-1\in A)\\r_{n-1}&(\text{if }n-1\not\in A)\end{cases}$, $\begin{cases}\hat{\gamma}_1^n&(\text{if }n\in J^\exists)\\r_m&(\text{if }n\in J^\forall)\end{cases}$\\
    $\NOT^{i}~(i\in J^\exists)$                  & $\begin{cases}\delta_1^i&(\text{if }i<n)\\\delta_2^{n-1}&(\text{if }i=n)\end{cases}$      & $r_i$\\
    \bottomrule
  \end{tabular}}
\end{table}

Let $\pi$ be a position order that satisfies the conditions of ordering in each gadget
and 
\begin{align*}
  &\{s_1\}\to\dots\to\{s_n\}\to \bigcup_{i=1}^n\bigcup_{j=1}^{m-1}P[\BRANCHING^{(i,j)}]\to \bigcup_{i=1}^n\bigcup_{j=1}^m\hat{y}_{i,j}\\
  &\to \bigcup_{(j,k)\in N^-} P[\NOT^{(j,k)}]
  \to\bigcup_{j=1}^m\{\hat{\ell}_{j,1},\hat{\ell}_{j,2},\hat{\ell}_{j,3}\}
  \to \bigcup_{j=1}^m P[\OR^{(j,1)}]\to \bigcup_{j=1}^m P[\OR^{(j,2)}]\to \{\hat{c}_1,\dots,\hat{c}_m\}\\
  &\to \bigcup_{j=1}^m P[\AND^j]\to \bigcup_{j=1}^{m-1}P[\BRANCHING^{j}]
  \to \bigcup_{i\in J^\exists} P[\NOT^{i}]\to \{r_1,\dots,r_n\}\to \{t_1\}\to\dots\to\{t_n\},
\end{align*}
where $X\to Y$ represents that $\pi^{-1}(x)<\pi^{-1}(y)$ for any $x\in X$ and $y\in Y$.
Let $\sigma$ be the offering order over $E$ that is induced by $\pi$.

Then, $y_{i,j}$ is matched with $\hat{y}_{i,j}$ if and only if $\bar{x}_i$ is matched with $s_i$, i.e., $v_i$ is assigned $\FALSE$ (see Fig.~\ref{fig:copy_variables}).
Also, $\ell_i$ is matched with $\hat{\ell}_i$ if and only if $\ell_i$ corresponds to $\FALSE$
and $c_j$ is matched with $\hat{c}_j$ if and only if $C_j$ is unsatisfied (see Fig.~\ref{fig:clause}).
In addition, $\alpha_3^{m-1}$ is matched with $\hat{\delta}_1^1$ if and only if $\varphi$ is satisfied (see Fig.~\ref{fig:formula}).
Therefore, we conclude that the constructed evaluation phase satisfies the desired condition.
\end{proof}

Thus, we have completed the proof of Theorem~\ref{thm:PSPACE-complete}.

\begin{figure}[htbp]
\begin{minipage}[]{.5\textwidth}
\centering  
\scalebox{.9}{%
\begin{tikzpicture}[xscale=0.9,yscale=0.75,
    p/.style={circle,inner sep=0pt,draw=black,font=\scriptsize,fill=white,minimum size=13pt},
    psmall/.style={circle,inner sep=0pt,draw=black,font=\fontsize{4}{0}\selectfont,fill=white,minimum size=10pt}
  ]

  \node[p] at (5,15)  (xi) {$\bar{x}_i$};
  \node[psmall] at (1,14)  (dh11) {$\hat{\delta}^{(i,1)}_1$};
  \node[psmall] at (5,14)  (d11) {$\delta^{(i,1)}_1$};
  \node[psmall] at (5,13)  (d12) {$\delta^{(i,1)}_2$};
  \node[p] at (1,11)  (yh1) {$\hat{y}_{i,1}$};
  \node[p] at (5,11)  (y1) {$y_{i,1}$};

  \node[psmall] at (1,10)  (dh21) {$\hat{\delta}^{(i,2)}_1$};
  \node[psmall] at (5,10)  (d21) {$\delta^{(i,2)}_1$};
  \node[psmall] at (5,9)  (d22) {$\delta^{(i,2)}_2$};
  \node[p] at (1,7)  (yh2) {$\hat{y}_{i,2}$};
  \node[p] at (5,7)  (y2) {$y_{i,2}$};
  \node at (1,6)  (dh31) {};

  \node[psmall] at (1,4)  (dhm1) {\scalebox{0.7}{$\hat{\delta}^{(i,m-1)}_1$}};
  \node[psmall] at (5,4)  (dm1) {\scalebox{0.7}{$\delta^{(i,m-1)}_1$}};
  \node[psmall] at (5,3)  (dm2) {\scalebox{0.7}{$\delta^{(i,m-1)}_2$}};
  \node[p,font=\fontsize{4}{0}\selectfont] at (1,1)  (yhmm) {\scalebox{0.8}{$\hat{y}_{i,m-1}$}};
  \node[p,font=\fontsize{4}{0}\selectfont] at (5,1)  (ymm) {\scalebox{0.8}{$y_{i,m-1}$}};
  \node[p,font=\fontsize{4.5}{0}\selectfont] at (1,0)  (yhm) {$\hat{y}_{i,m}$};
  \node[p,font=\fontsize{4.5}{0}\selectfont] at (5,0)  (ym) {$y_{i,m}$}; 

  \draw[fill=black!10,fill opacity=.6] (0.5,11.5) rectangle (5.5,14.5);
  \draw[fill=black!10,fill opacity=.6] (0.5,7.5) rectangle (5.5,10.5);
  \draw[fill=black!10,fill opacity=.6] (0.5,1.5) rectangle (5.5,4.5);
  \node at (3,13) {$\BRANCHING^{(i,1)}$};
  \node at (3,9) {$\BRANCHING^{(i,2)}$};
  \node at (3,3) {$\BRANCHING^{(i,m-1)}$};

  \foreach \u/\v in {xi/dh11,d11/yh1,d12/dh21,d21/yh2,d22/dh31,dm1/yhmm,dm2/yhm} \draw[gray] (\u) -- (\v);
  \draw[gray] (dhm1)--(5,5);
  \foreach \u/\v in {y1/yh1,y2/yh2,ymm/yhmm,ym/yhm} \draw (\u) -- (\v);
  \draw (xi)--(4,15.5);
  \draw (y1)--(4,10.8);
  \draw (y2)--(4,6.8);
  \draw (ymm)--(4,0.8);
  \draw (ym)--(4,-0.2);
  \node at (3,6) {\rvdots};
\end{tikzpicture}}
\caption{Copy variables}\label{fig:copy_variables}
\vspace{.5cm}
\scalebox{.9}{%
\begin{tikzpicture}[xscale=0.9,yscale=0.75,
    p/.style={circle,inner sep=0pt,draw=black,font=\scriptsize,fill=white,minimum size=13pt},
    psmall/.style={circle,inner sep=0pt,draw=black,font=\fontsize{4}{0}\selectfont,fill=white,minimum size=10pt}
  ]

  \node[p] at (1,9)  (lh1) {$\hat{\ell}_{j,1}$};
  \node[p] at (1,8)  (lh2) {$\hat{\ell}_{j,2}$};
  \node[p] at (1,7)  (lh3) {$\hat{\ell}_{j,3}$};
  \node[p] at (5,9)  (l1) {$\ell_{j,1}$}; 
  \node[p] at (5,8)  (l2) {$\ell_{j,2}$}; 
  \node[p] at (5,7)  (l3) {$\ell_{j,3}$}; 

  \node[psmall] at (1,5.8)  (a1) {$\hat{\beta}_1^{(j,1)}$};
  \node[psmall] at (1,4.8)  (a2) {$\hat{\beta}_2^{(j,1)}$};
  \node[psmall] at (5,4.2)  (a3) {$\beta_2^{(j,1)}$};

  \node[psmall] at (1,2.8)  (a4) {$\hat{\beta}_1^{(j,2)}$};
  \node[psmall] at (1,1.8)  (a5) {$\hat{\beta}_2^{(j,2)}$};
  \node[psmall] at (5,1.2)  (a6) {$\beta_3^{(j,2)}$};

  \node[p] at (1,0)  (c1) {$\hat{c}_j$};
  \node[p] at (5,0)  (c2) {$c_j$};

  \draw[fill=black!10,fill opacity=.6] (0.5,0.7) rectangle (5.5,3.5);
  \draw[fill=black!10,fill opacity=.6] (0.5,3.7) rectangle (5.5,6.3);
  \node at (3,5) {$\OR^{(j,1)}$};
  \node at (3,2) {$\OR^{(j,2)}$};

  \foreach \u/\v in {lh1/l1,lh2/l2,lh3/l3,l1/a1,l2/a2,l3/a5,a3/a4,a6/c1,c1/c2} \draw[gray] (\u) -- (\v);
  \draw[gray] (lh1)--(2,9.5);
  \draw[gray] (lh2)--(2,8.5);
  \draw[gray] (lh3)--(2,7.5);
\end{tikzpicture}}
\caption{Clause $C_j=\ell_{j,1}\vee \ell_{j,2}\vee \ell_{j,3}$}\label{fig:clause}
\end{minipage}%
\begin{minipage}[]{.5\textwidth}
\centering  
\scalebox{.9}{%
\begin{tikzpicture}[xscale=0.9,yscale=0.75,
    p/.style={circle,inner sep=0pt,draw=black,font=\scriptsize,fill=white,minimum size=13pt},
    psmall/.style={circle,inner sep=0pt,draw=black,font=\fontsize{4}{0}\selectfont,fill=white,minimum size=10pt}
  ]
  \node[p] at (5,26)  (c1) {$c_1$};
  \node[p] at (5,25)  (c2) {$c_2$};
  \node[p] at (5,21)  (c3) {$c_3$};
  \node[p] at (5,16)  (cm) {$c_m$};

  \node[p] at (1,8)  (r1) {$r_{1}$};
  \node[p] at (1,1)  (rnn) {$r_{n-1}$};
  \node[p] at (1,0)  (rn) {$r_n$};

  \node[p] at (5,8)  (z1) {$z_{1}$};
  \node[p] at (5,1)  (znn) {$z_{n-1}$};
  \node[p] at (5,0)  (zn) {$z_n$};

  \node[psmall] at (1,24)  (b11) {$\hat{\alpha}_1^{1}$};
  \node[psmall] at (1,23)  (b12) {$\hat{\alpha}_2^{1}$};
  \node[psmall] at (5,22)  (b13) {$\alpha_3^{1}$};

  \node[psmall] at (1,20)  (b21) {$\hat{\alpha}_1^{2}$};
  \node[psmall] at (1,19)  (b22) {$\hat{\alpha}_2^{2}$};
  \node[psmall] at (5,18)  (b23) {$\alpha_3^{2}$};

  \node[psmall] at (1,15)  (bm1) {$\hat{\alpha}_1^{m-1}$};
  \node[psmall] at (1,14)  (bm2) {$\hat{\alpha}_2^{m-1}$};
  \node[psmall] at (5,13)  (bm3) {$\alpha_3^{m-1}$};
  
  \node[psmall] at (1,11)  (d11) {$\hat{\delta}_1^{1}$};
  \node[psmall] at (5,11)  (d12) {$\delta_1^{1}$};
  \node[psmall] at (5,10)  (d13) {$\delta_2^{1}$};

  \node[psmall] at (1,6)  (dn1) {$\hat{\delta}_1^{n-1}$};
  \node[psmall] at (5,6)  (dn2) {$\delta_1^{n-1}$};
  \node[psmall] at (5,5)  (dn3) {$\delta_2^{n-1}$};

  \node[psmall] at (1,2.8)  (gn1) {\scalebox{0.9}{$\hat{\gamma}_1^{n-1}$}};
  \node[psmall] at (5,2)  (gn2) {$\gamma_4^{n-1}$};

  \draw[fill=black!10,fill opacity=.6] (0.5,21.5) rectangle (5.5,24.5);
  \node at (3,23) {$\AND^{1}$};
  \draw[fill=black!10,fill opacity=.6] (0.5,17.5) rectangle (5.5,20.5);  
  \node at (3,19) {$\AND^{2}$};
  \draw[fill=black!10,fill opacity=.6] (0.5,12.5) rectangle (5.5,15.5);
  \node at (3,14) {$\AND^{m-1}$};

  \draw[fill=black!10,fill opacity=.6] (0.5,8.5) rectangle (5.5,11.5);
  \node at (3,10) {$\BRANCHING^{1}$};
  \draw[fill=black!10,fill opacity=.6] (0.5,3.7) rectangle (5.5,6.5);
  \node at (3,5) {$\BRANCHING^{n-1}$};
  \draw[fill=black!10,fill opacity=.6] (0.5,1.5) rectangle (5.5,3.3);
  \node at (3,2.5) {$\NOT^{n-1}$};

  \foreach \u/\v in {c1/b11,c2/b12,b13/b21,c3/b22,cm/bm2,bm3/d11,d12/r1,dn2/gn1,dn3/rn,gn2/rnn} \draw[gray] (\u) -- (\v);
  \foreach \u/\v in {z1/r1,znn/rnn,zn/rn} \draw (\u) -- (\v);
  \draw (c1)--(4,26.3);
  \draw (c2)--(4,25.3);
  \draw (c3)--(4,21.3);
  \draw (cm)--(4,16.3);
  \draw[gray] (b23)--(3.5,17.25);
  \draw[gray] (bm1)--(2.5,15.75);
  \draw[gray] (d13)--(1,7);
  \draw[gray] (dn1)--(5,7);
  \draw (z1)--(4,7.7);
  \draw (znn)--(4,0.7);
  \draw (zn)--(4,-0.3);
  \node at (3,17) {\rvdots};
  \node at (3,7) {\rvdots};
  
\end{tikzpicture}}
\caption{$(\forall x_1)\dots(\exists x_{n-1})(\forall x_n)\varphi(x_1,\dots,x_n)$}\label{fig:formula}
\end{minipage}
\end{figure}

\clearpage

\section{Leading SPE Matchings to Be Stable}\label{sec:imp_qopt}
As observed in Section~\ref{subsec:bo}, the SPE matching in our sequential matching model may violate the stability condition,
which is one of the desirable properties in two-sided matching that represents a kind of fairness for all firms and workers.
On the other hand, as we have seen in Section \ref{sec:model}, the SPE matching may differ according to offering orders.
In this section, we show that the worker-optimal stable matching can be led to as the outcome of an SPE by choosing an appropriate offering order.
Formally, we prove the following theorem, which is a reformulation of Theorem~\ref{thm:ordering}.

\begin{theorem}\label{thm:qopt}
For any matching instance $I$, there exists an offering order $\sigma\in\Sigma_I$
such that $\spe(I,\sigma)=\qopt(I)$.
\end{theorem}

We give a constructive proof for Theorem~\ref{thm:qopt} as follows.
In Section~\ref{sec:construction}, we construct an offering order $\sigma \in \Sigma_I$ for a given matching instance $I$.
In Section~\ref{sec:lemmas}, we show useful lemmas.
Finally, in Section~\ref{sec:correctness}, we prove $\spe(I, \sigma) = \qopt(I)$ for $\sigma$ constructed in Section~\ref{sec:construction}.

\subsection{Constructing an Offering Order}\label{sec:construction}
Fix an arbitrary matching instance $I = (P, Q, E, {\succ})$.
We shall construct an offering order $\sigma \in \Sigma_I$, for which we prove $\spe(I, \sigma) = \qopt(I)$ in Section~\ref{sec:correctness}.
The construction is two-phased and sketched as follows: we first arrange all the offers preferable to the matched pairs in $\qopt(I)$ for the firms,
and subsequently arrange the rest of the offers according to a position order that is defined by the acceptance order in Algorithm~\ref{alg:QOPT} for computing $\qopt(I)$.

Let $F \coloneq \{ (p, q) \in E \mid q \succ_p \qopt(I)(p) \} \subseteq E \setminus \qopt(I)$.
We arrange the offers in $F$ prior to all the matched pairs in $\qopt(I)$.
Formally, suppose that $\sigma^{-1}(e) < \sigma^{-1}(e^*)$ holds for every $e \in F$ and $e^* \in \qopt(I)$,
i.e., we restrict ourselves to offering orders $\sigma \in \Sigma_I$ such that
$F = \{ \sigma(i) \mid i = 1, 2, \ldots, |F| \}$.
We then denote by $(I', \sigma')$ the subgame obtained by deleting all the pairs in $F$.
Note that we see $\qopt(I') = \qopt(I)$ by applying Properties~\ref{prop:first_edge} and \ref{prop:delete_Q2} repeatedly,
and for every $(p, q) \in \qopt(I')$, the worker $q$ is on the top of $p$'s list in $I'$.

In what follows, let us construct a position order $\pi$ over $P$
that induces an offering order $\sigma' \in \Sigma_{I'}$ over $E \setminus F$.
Recall that, for a matching instance $\hat{I}$, the worker-optimal stable matching $\qopt(\hat{I})$ is computed by Algorithm~\ref{alg:QOPT},
and let $\alpha(\hat{I})$ denote the set of firms $p \in P$ that accept $\qopt(\hat{I})(p)$'s proposal
in the earliest iteration step (in the sense of {\bf while} in Line~\ref{line:2}) among the firms matched in $\qopt(\hat{I})$.

Let $k \coloneq |\qopt(I')|$.
We construct a sequence of matching instances $I_i$ $(i = 0, 1, \dots, k)$ as follows: 
\begin{itemize}
\item
  let $I_0 \coloneq I'$, whose worker-optimal stable matching is the target matching here, and
\item
  for $i = 1, 2, \dots, k$,
  take a matched pair $e_i = (p_i, q_i) \in \qopt(I_{i-1})$ with $p_i \in \alpha(I_{i-1})$ and let $I_i \coloneq I_{i-1} / e_i$, which is, intuitively, a subinstance after the pair $e_i \in \qopt(I_{i-1})$ matched first in Algorithm~\ref{alg:QOPT} leaves the market.
\end{itemize}
Since Lemma~\ref{lem:first_match} (shown in Section~\ref{sec:lemmas}) guarantees $\qopt(I_i) = \qopt(I_{i-1}) - e_i$,
we have $\qopt(I') = \{ e_1, e_2, \dots, e_k \}$,
and $I_k$ has no acceptable pair (otherwise, an acceptable pair in $I_k$ blocks $\qopt(I')$ in $I'$, contradicting the stability).

We now define a position order $\pi$ over $P$ that induces $\sigma'$ as follows:
\begin{itemize}
\item
  $\pi(i) \coloneq p_i$ for $i = 1, 2, \dots, k$, and
\item
  for $i = k + 1, \dots, |P|$, 
  take any firm $p \in P \setminus \{ \pi(1), \pi(2), \dots, \pi(i-1) \}$
  and let $\pi(i) \coloneq p$.
\end{itemize}
Note that $\pi(1), \pi(2), \dots, \pi(k)$ are the firms matched in $\qopt(I')$.

Thus, extending the position-based offering order $\sigma' \in \Sigma_{I'}$ induced by $\pi$
by arranging the offers in $F$ arbitrarily at the beginning, we obtain an offering order $\sigma \in \Sigma_I$.

It is worth remarking that the above construction is done in polynomial time.
In the first phase, computing and arranging the offers in $F$ are done in $\mathrm{O}(|E|)$ time by a single execution of Algorithm~\ref{alg:QOPT}
and by picking the top worker of the list of some firm repeatedly.
In the second phase, $k = \mathrm{O}(|V|)$ executions of Algorithm~\ref{alg:QOPT} are sufficient for constructing the sequence $I_i$ $(i = 0, 1, \dots, k)$
and the position order $\pi$.
Thus, the total computational time is bounded by $\mathrm{O}(|V| \cdot |E|)$.

\subsection{Useful Lemmas}\label{sec:lemmas}
We first show two lemmas on the worker-optimal stable matching.

The first one 
adjusts a well-known property, called the blocking lemma~\cite{GS1985}, to our situation,
and the proof is almost the same.
Recall that, for two matchings $\mu, \mu' \in \mathcal{M}_I$ in a matching instance $I = (P, Q, E, {\succ})$,
we write $\mu \succ_Q \mu'$ if $\mu \neq \mu'$ and $\mu(q) \succeq_q \mu'(q)$ for every worker $q \in Q$.

\begin{lemma}[cf.~blocking lemma {\cite{GS1985}}]\label{lem:blocking}
  Let $I = (P, Q, E, {\succ})$ be a matching instance.
  For a matching $\mu' \in \mathcal{M}_I$ with $\mu' \succ_Q \qopt(I)$,
  define $Q' \coloneq \{ q \in Q \mid \mu'(q) \succ_q \qopt(I)(q) \}$.
  Then, $\mu'$ is blocked by some pair $(p, q) \in E$ with $p \in \mu'(Q')$ and $q \in Q \setminus Q'$
  such that no worker $\tilde{q} \in \Gamma_I(p)$ satisfies both $\qopt(I)(p) \succ_p \tilde{q} \succ_p q$ and $p \succ_{\tilde{q}} \qopt(I)(\tilde{q})$.
\end{lemma}

\begin{proof}
Let $P' \coloneq \mu'(Q') = \qopt(I)(Q')$. Take a firm $p \in P'$ that accepts $\qopt(I)(p)$'s proposal at last among $P'$ in Algorithm~\ref{alg:QOPT}
(where some ambiguity remains due to the order of choices of $p \in P$ in Line~\ref{line:5},
but it does not matter in the following argument).
Let $q' \coloneq \mu'(p) \in Q'$.
By definition of $p$, the proposal to $p$ from any $q'' \in Q'$ with $p \succ_{q''} \qopt(I)(q'')$ is rejected in some strictly earlier iteration step (in the sense of {\bf while} in Line~\ref{line:2}) than $\qopt(I)(p)$ proposes to $p$,
because $q''$ proposes to $\qopt(I)(q'') \in Q'$ (and the proposal is accepted) after the rejection.
Hence, there exists a worker $q \in Q \setminus Q'$ with $\qopt(I)(p) \succ_{p} q \succ_{p} q' = \mu'(p) \in Q'$ and $p \succ_q \qopt(I)(q) = \mu'(q)$, which implies that $(p, q) \in E$ blocks $\mu'$.
If we choose $q$ as the top in $p$'s list among such workers, then no worker $\tilde{q} \in \Gamma_I(p)$ satisfies both $\qopt(I)(p) \succ_p \tilde{q} \succ_p q$ and $p \succ_{\tilde{q}} \qopt(I)(\tilde{q})$.
\end{proof}

The second one is directly used for our construction of a position order $\pi$ in Section~\ref{sec:construction}.
Recall that, for a matching instance $I$, we denote by $\alpha(I)$ the set of firms $p$ that accept $\qopt(I)(p)$'s proposal
in the earliest iteration step (in the sense of {\bf while} in Line~\ref{line:2}) among the firms matched in $\qopt(I)$.

\begin{lemma}\label{lem:first_match}
  Let $I = (P, Q, E, {\succ})$ be a matching instance in which $q$ is on the top of $p$'s list for every pair $(p, q) \in \qopt(I)$.
  Then, for any pair $e^* = (p^*, q^*) \in \qopt(I)$ with $p^* \in \alpha(I)$, we have $\qopt(I/e^*) = \qopt(I) - e^*$.
\end{lemma}

\begin{proof}
Let $\mu \coloneq \qopt(I)$ and $\mu' \coloneq \qopt(I/e^*) + e^*$,
and suppose to the contrary that $\mu' \neq \mu$.
Then, $\mu' \succ_Q \mu$ by Property~\ref{prop:contract_Q},
and let $Q' \coloneq \{ q \in Q \mid \mu'(q) \succ_{q} \mu(q) \}$.
Note that $q^* \not\in Q'$ because $\mu(q^*) = \mu'(q^*) = p^*$.
By Lemma~\ref{lem:blocking}, $\mu'$ is blocked by some pair $(p, q) \in E$ with $p \in \mu'(Q')$ and $q \in Q \setminus Q'$
such that no worker $\tilde{q} \in \Gamma_I(p)$ satisfies both $\mu(p) \succ_p \tilde{q} \succ_p q$ and $p \succ_{\tilde{q}} \mu(\tilde{q})$.
Then, $q$'s proposal to $p~(\succ_q \mu'(q) = \mu(q))$ is rejected only when (or after) $p$ accepts $\mu(p)$'s proposal in Algorithm~\ref{alg:QOPT},
and hence $q$ proposes to $\mu(q)$ in a strictly later iteration step than $p$ accepts $\mu(p)$'s proposal.
This implies $q \neq q^*$ because $\mu(q^*) = p^* \in \alpha(I)$.
However, if $q \in Q \setminus (Q' \cup \{q^*\})$, then $(p, q) \in E/e^*$ blocks $\mu' - e^*$ in $I/e^*$, a contradiction.
\end{proof}

We next observe several properties of SPE matchings.
First, we see that a blocking pair for an SPE matching is restricted to a certain type.

\begin{lemma}\label{lem:blocking_pair}
  For a sequential matching game $(I,\sigma)$, let $\mu \coloneq \spe(I, \sigma)$.
  Then, an acceptable pair $(p, q)$ blocks $\mu$ (i.e., $q\succ_p\mu(p)$ and $p\succ_{q}\mu(q)$) only if $\sigma^{-1}((p,q))>\sigma^{-1}((\mu(q),q))$.
\end{lemma}
\begin{proof}
  We prove the proposition by contradiction.
  Suppose that $q\succ_p\mu(p)$, $p\succ_{q}\mu(q)$, and $\sigma^{-1}((p,q))<\sigma^{-1}((\mu(q),q))$.
  Let us consider the subgame $(I',\sigma')$ of $(I, \sigma)$ that is reached in the SPE where $\sigma'(1)=(p,q)$.
  Here, the subgame $(I',\sigma')$ exists because $\sigma^{-1}((p,q))<\sigma^{-1}((\mu(q),q))$ but $p$ and $q$ are not matched by $q\succ_p\mu(p)$.
  The optimal strategy of $q$ at $(I',\sigma')$ is \ACCEPT since $p\succ_{q}\mu(q)$; however, $q$ selects \REJECT in the SPE, which is a contradiction.
\end{proof}

The next one is analogous to Property~\ref{prop:first_edge} on the worker-optimal stable matchings, and
it can be also interpreted as the \emph{first-choice stability}\footnote{A matching $\mu \in \mathcal{M}_I$ is first-choice stable for the firms (resp., the workers) if $\mu$ is not blocked by any pair $(p, q)$ such that $q$ is on the top of $p$'s list (resp., $p$ is on the top of $q$'s list)~\cite{DMS2018}.
An SPE matching may not be first-choice stable for the workers,
e.g., Example~\ref{ex:weak_stability} in Section~\ref{subsec:bo}.} for the firms.

\begin{lemma}\label{lem:top_offer}
  For a sequential matching game $(I,\sigma)$,
  let $e=(p,q)$ be an acceptable pair such that $q$ is on the top of $p$'s list.
  Then, $\spe(I,\sigma)(q) \succeq_q p$.
\end{lemma}
\begin{proof}
  We prove this by induction on $\sigma^{-1}(e)$.
  For the base case $\sigma^{-1}(e)=1$, by Proposition~\ref{prop:operations_spe},
  the worker $q$ selects \REJECT if  $\spe(I-e, \sigma -e) \succ_q p$, and \ACCEPT otherwise.
  Hence, $\spe(I,\sigma)(q) \succeq_q p$.
  
  Suppose that $\sigma^{-1}(e)>1$.
  Let $\sigma(1)=e'=(p',q')$. Note that $p'\ne p$ since $q$ is on the top of $p$'s list.
  By Proposition~\ref{prop:operations_spe} (i), if $\spe(I-e',\sigma-e')(q')\succ_{q'} p'$,
  then we derive from induction hypothesis
  \[\spe(I,\sigma)(q) = \spe(I-e',\sigma-e')(q) \succeq_q p.\]
  Similarly, by Proposition~\ref{prop:operations_spe} (ii), if $\spe(I-e',\sigma-e')(q') \prec_{q'} p'$ and $q' \neq q$,
  then we derive
  \[\spe(I,\sigma)(q) = \spe(I/e',\sigma/e')(q) \succeq_q p.\]
  Otherwise, we have $q' = q$ and $p' \succ_q \spe(I-e',\sigma-e')(q) \succ_q p$ by induction hypothesis,
  and hence $q$ selects \ACCEPT, i.e., $\spe(I, \sigma)(q) = p' \succ_q p$.
  Thus we are done.
\end{proof}

We can similarly show the following proposition, which states that top-top pairs are matched in the SPE matching and can be removed anytime.
This implies that, if a matching instance satisfies the \emph{Eeckhout condition}\footnotemark{},
the SPE matching coincides with a unique stable matching regardless of the offering order.
\footnotetext{
  The Eeckhout condition is a sufficient condition for the existence of a unique stable matching.
  We say that a matching instance $I=(P,Q,E,\succ)$ satisfies the Eeckhout condition if it is possible to rearrange firms and workers so that (i) for any firm $p_k\in P$, $q_i\succ_{p_k}q_j$ for all $j>i$, and (ii) for any worker $q_k\in Q$, $p_i\succ_{q_k}p_j$ for all $j>i$~\cite{Eeckhout2000}.
}

\begin{proposition}\label{prop:top-pair}
  For a sequential matching game $(I,\sigma)$, let $e=(p,q)$ be an acceptable pair such that each of them is the most preferred partner of the other.
  Then, $\spe(I,\sigma)=\spe(I/e,\sigma/e)+(p,q)$.
\end{proposition}
\begin{proof}
  We prove this by induction on $\sigma^{-1}(e)$.
  The base case $\sigma^{-1}(e) = 1$ is trivial.
  Suppose that $\sigma^{-1}(e)>1$,
  and let $\sigma(1)=e'=(p',q')$. Note that $p'\ne p$ since $q$ is most preferred by $p$.
  By Proposition~\ref{prop:operations_spe} (i), if $\spe(I-e',\sigma-e')(q')\succ_{q'} p'$, then we have
  \begin{align*}
  \spe(I,\sigma) &= \spe(I-e',\sigma-e')\\
  &=\spe((I-e')/e,(\sigma-e')/e)+(p,q) & (\text{induction hypothesis})\\
  &=\spe((I / e) -e', (\sigma / e) -e')+(p,q) & (\text{Observation~\ref{obs:commutative}})\\
  &=\spe(I/e,\sigma/e)+(p,q),
  \end{align*}
  where the last equality follows from $\spe((I / e) -e', (\sigma / e) -e')(q') = \spe(I-e',\sigma-e')(q')\succ_{q'} p'$,
  which is implied by the second equality.
  Similarly, by Proposition~\ref{prop:operations_spe} (ii),
  we have $\spe(I,\sigma)=\spe(I/e',\sigma/e')+(p',q')=\spe((I/e')/e, (\sigma/e')/e)+(p',q')+(p,q)=\spe(I/e,\sigma/e)+(p,q)$
  if $\spe(I-e',\sigma-e')(q')\prec_{q'} p'$.
  Hence, the proposition holds.
\end{proof}

\subsection{Correctness}\label{sec:correctness}
We here prove that $\spe(I, \sigma) = \qopt(I)$ for $\sigma \in \Sigma_I$ constructed in Section~\ref{sec:construction}.
Recall that the construction is two-phased:
we first arrange the offers in $F = \{ (p, q) \in E \mid q \succ_p \qopt(I)(p) \} \subseteq E \setminus \qopt(I)$ prior to the others,
and the order $\sigma' \in \Sigma_{I'}$ on the rest offers is induced by a position order $\pi$ over $P$
in which the firms matched in $\qopt(I')$ are arranged prior to the unmatched firms.

The following claim implies that all the offers in $F$ are supposed to be rejected
as long as $\sigma'$ is designed appropriately, where recall that $\qopt(I') = \qopt(I)$ follows from Properties~\ref{prop:first_edge} and \ref{prop:delete_Q2}.

\begin{claim}\label{claim:F}
  If $\spe(I', \sigma') = \qopt(I')$, then $\spe(I, \sigma) = \qopt(I)$.
\end{claim}

\begin{proof}
We show the claim by induction on $|F|$.
When $F = \emptyset$, it trivially holds.
Then, it suffices to show that, for the first offer $\sigma(1) = e = (p, q) \in F$, the worker $q$ selects \REJECT in the SPE when $F \neq \emptyset$.
By induction hypothesis and Property~\ref{prop:first_edge},
we see $\spe(I -e, \sigma -e)(q) = \qopt(I)(q) \succ_q p$,
which concludes that $q$ indeed selects \REJECT in the SPE\@.
\end{proof}

We then prove that $\sigma'$ is indeed designed appropriately,
which concludes $\spe(I, \sigma) = \qopt(I)$ by combining with Claim~\ref{claim:F}.

\begin{claim}\label{claim:spe=qopt}
  $\spe(I', \sigma') = \qopt(I')$.
\end{claim}

\begin{proof}
Recall that we construct a sequence of matching instances $I_i$ $(i = 0, 1, \dots, k)$,
where $k = |\qopt(I')|$, such that
\begin{itemize}
\setlength{\itemsep}{0mm}
\item
  $I_0 = I'$,
\item
  $I_i = I_{i-1} / e_i$ for some $e_i = (p_i, q_i) \in \qopt(I_{i-1})$ with $\pi(i) = p_i \in \alpha(I_{i-1})$ for $i = 1, 2, \dots, k$,
\item
  $\qopt(I_i) = \qopt(I_{i-1}) - e_i$ for $i = 1, 2, \dots, k$ (by Lemma~\ref{lem:first_match}), and
\item
  $I_k$ has no acceptable pair.
\end{itemize}

We show $\spe(I_i, \sigma_i) = \qopt(I_i)$ by induction on $i$ (in descending order),
where $\sigma_i \in \Sigma_{I_i}$ denotes the restriction of $\sigma$ to $I_i$.
The base case $\spe(I_k, \sigma_k) = \emptyset = \qopt(I_k)$ is trivial,
and the case of $i = 0$ is exactly what we want to prove.
As each step of the induction, it suffices to show the next claim,
which says that $q_i$ select \ACCEPT for the first offer $\sigma_{i-1}(1) = e_i = (p_i, q_i)$ in the SPE of $(I_{i-1}, \sigma_{i-1})$.
\end{proof}

\begin{claim}\label{cl:accept}
$p_i = \qopt(I_{i-1})(q_i) \succ_{q_i} \spe(I_{i-1} -e_i, \sigma_{i-1} -e_i)(q_i)$ for every $i = 1, 2, \dots, k$.
\end{claim}

\begin{proof}
Let $\mu' \coloneq \spe(I_{i-1} -e_i, \sigma_{i-1} -e_i)$,
and suppose to the contrary that $\mu'(q_i) \succ_{q_i} p_i = \qopt(I_{i-1})(q_i)$.
For each worker $q_j$ $(j > i)$ (who is matched in $\qopt(I_{i-1}) - e_i$),
Lemma~\ref{lem:top_offer} implies that $\mu'(q_j) \succeq_{q_j} \qopt(I_{i-1})(q_j)$ because $q_j$ is on the top of $p_j$'s list in $I_{i-1} -e_i$.
In addition, for any other worker $q \in Q \setminus \{ q_1, q_2, \dots, q_k \}$,
obviously $\mu'(q) \succeq_q q = \qopt(I_{i-1})(q)$.
Hence, we have $\mu' \succ_{Q_i} \qopt(I_{i-1})$,
where $Q_i \coloneq Q \setminus \{ q_1, q_2, \dots, q_{i-1} \}$.

Let $Q'_i \coloneq \{ q \in Q_i \mid \mu'(q) \succ_{q} \qopt(I_{i-1})(q) \}$. Note that $q_i \in Q'_i$.
By Lemma~\ref{lem:blocking}, $\mu'$ is blocked by some pair $(p, q) \in E$ with $p \in \mu'(Q'_i)$ and $q \in Q_i \setminus Q'_i$
such that no worker $\tilde{q} \in \Gamma_{I_i}(p)$ satisfies both $\qopt(I_{i-1})(p) \succ_p \tilde{q} \succ_p q$ and $p \succ_{\tilde{q}} \qopt(I_{i-1})(\tilde{q})$.
By Lemma~\ref{lem:blocking_pair}, we have $\sigma_{i-1}^{-1}((p, q)) > \sigma_{i-1}^{-1}((\mu'(q), q))$,
which is equivalent to $\pi^{-1}(p) > \pi^{-1}(\mu'(q))$.
This, however, cannot hold, because when applying Algorithm~\ref{alg:QOPT} to any instance $I_j$ $(j \geq i - 1)$ in which both $p$ and $\mu'(q) = \qopt(I_{i-1})(q)$ remain,
$q$ proposes to $\qopt(I_j)(q) = \qopt(I_{i-1})(q)$ in a strictly later iteration step than is rejected by $p~(\succ_q \mu'(q) = \qopt(I_j)(q))$,
which is no earlier than $p$ accepts the proposal from $\qopt(I_j)(p) = \qopt(I_{i-1})(p)$.
\end{proof}

Thus, we have proved Claim~\ref{claim:spe=qopt}, which completes the proof of Theorem~\ref{thm:qopt}.

Note that Claim~\ref{claim:spe=qopt} also leads to the following corollary,
which corresponds to a special case when $F = \emptyset$ holds for a given instance $I$ in Theorem~\ref{thm:qopt}.

\begin{corollary}\label{cor:qopt}
For any matching instance $I$ such that for every pair $(p, q) \in \qopt(I)$ the worker $q$ is on the top of the firm $p$'s list,
there exists a position-based offering order $\sigma\in\Sigma_I$ such that $\spe(I,\sigma)=\qopt(I)$.
\end{corollary}

We conclude this section with the following conjecture, which generalizes this corollary.
The offering order $\sigma$ constructed in Section~\ref{sec:construction} looks somewhat artificial,
and if one can achieve the same property by a position-based order,
it is further natural in the sense of designing markets.

\begin{conjecture}
For any matching instance $I$, there exists a position-based offering order $\sigma\in\Sigma_I$
such that $\spe(I,\sigma)=\qopt(I)$.
\end{conjecture}

\section{Impossibility to Achieve Socially Desirable Matchings by SPEs}\label{sec:imp}
\newcommand{\SPE}{\mathcal{SPE}}
\newcommand{\SM}{\mathcal{SM}}
\newcommand{\PE}{\mathcal{PE}} 
\newcommand{\FCM}{\mathcal{FCM}} 
\newcommand{\popt}{{\rm POPT}}

In this section, we give some impossibility results.
We show that the SPE matchings for a matching instance may not support the firm-optimal stable matching,
any (firm-side and worker-side) Pareto-efficient matching,
and any (firm-side and worker-side) first-choice maximal matching,
which are defined successively in Definitions~\ref{def:POPT}--\ref{def:FCM}.

Let $I = (P, Q, E, {\succ})$ be a matching instance.
Recall that, for matchings $\mu, \mu' \in \mathcal{M}_I$,
we write $\mu \succeq_Q \mu'$ if $\mu(q) \succeq_q \mu'(q)$ $(\forall q \in Q)$,
and write $\mu \succ_Q \mu'$ if $\mu \neq \mu'$ in addition.
Similarly, we write $\mu \succeq_P \mu'$ if $\mu(p) \succeq_p \mu'(p)$ $(\forall p\in P)$,
and write $\mu \succ_P \mu'$ if $\mu \neq \mu'$ in addition.
Also similar to Definition~\ref{def:QOPT}, we can define the firm-optimal stable matching.

\begin{definition}[Optimality (cf.~Definition~\ref{def:QOPT})]\label{def:POPT}
The set of all stable matchings in $I$ forms a distributive lattice with respect to the partial order $\succeq_P$.
A unique maximal element in this distributive lattice is
said to be \emph{firm-optimal} 
and denoted by $\popt(I)$.
That is, $\popt(I)(p) \succeq_p \mu(p)$ for any stable matching $\mu$ in $I$ and any firm $p \in P$.
\end{definition}

\begin{definition}[Pareto-efficiency]\label{def:Pareto}
  A matching $\mu\in\mathcal{M}_I$ is called \emph{firm-side} (resp., \emph{worker-side}) \emph{Pareto-efficient} if
  no matching $\mu' \in \mathcal{M}_I \setminus \{\mu\}$ satisfies $\mu\succ_P\mu'$ (resp., $\mu\succ_Q\mu'$).
  We denote the set of firm-side (resp., worker-side) Pareto-efficient matchings by $\PE^P(I)$ (resp., $\PE^Q(I)$).
\end{definition}

\begin{definition}[First-choice maximality~\cite{DMS2018}]\label{def:FCM}
  For a matching instance $I=(P,Q,E,\succ)$,
  let $E_P\coloneqq\{(p,q)\in E\mid \text{$q$ is on the top of $p$'s list}\}$ and $E_Q\coloneqq\{(p,q)\in E\mid \text{$p$ is on the top of $q$'s list}\}$.
  A matching $\mu \in \mathcal{M}_I$ is said to be \emph{firm-side} (resp., \emph{worker-side}) \emph{first-choice maximal} if
  $|\mu\cap E_P|\ge |\mu'\cap E_P|$ (resp., $|\mu\cap E_Q|\ge |\mu'\cap E_Q|$) for all $\mu'\in\mathcal{M}_I$.
  We denote the set of firm-side (resp., worker-side) first-choice maximal matchings by $\FCM^P(I)$ (resp., $\FCM^Q(I)$).
\end{definition}

For a matching instance $I$,
let $\SPE(I)$ be the set of matchings supported by SPEs of the sequential matching games with $I$, i.e., $\SPE(I)\coloneqq\{\spe(I,\sigma)\mid \sigma\in\Sigma_I\}$.
Note that, as we saw in Section~\ref{sec:imp_qopt},
the worker-optimal stable matching is always leadable as the outcome of an SPE, i.e., $\qopt(I)\in\SPE(I)$ for every $I$.

We are ready to state our impossibility results that there exist matching instances
in which the following properties cannot be achieved by any SPE matching:
the firm-optimality under stability, the (firm-side and worker-side) Pareto-efficiency, and the (firm-side and worker-side) first-choice maximality.
\begin{proposition}\label{prop:spe_notallsm}
  There exists a matching instance $I$ such that $\popt(I)\not\in \SPE(I)$.
\end{proposition}
\begin{proof}
  See Example~\ref{ex:spe_notallsm}.
\end{proof}

\begin{proposition}\label{prop:nonPareto}
  There exist matching instance $I$ and $I'$ such that $\SPE(I)\cap \PE^P(I)=\emptyset$ and $\SPE(I')\cap \PE^Q(I')=\emptyset$.
\end{proposition}
\begin{proof}
See Example~\ref{ex:spe_notallsm} for the firm-side and Example~\ref{ex:nonPareto} for the worker-side.
\end{proof}

\begin{proposition}\label{prop:nonFCM}
  There exist matching instances $I$ and $I'$ such that $\SPE(I)\cap \FCM^P(I)=\emptyset$ and $\SPE(I')\cap \FCM^Q(I')=\emptyset$.
\end{proposition}
\begin{proof}
See Example~\ref{ex:spe_notallsm} for the firm-side and Example~\ref{ex:nonPareto} for the worker-side.
\end{proof}

\begin{example}[Proof of Proposition~\ref{prop:spe_notallsm} and the firm-side parts of Propositions~\ref{prop:nonPareto} and~\ref{prop:nonFCM}]\label{ex:spe_notallsm}
Let us consider a matching instance $I=(P,Q,E,\succ)$ with two firms $P=\{p_1,p_2\}$ and two workers $Q=\{q_1,q_2\}$
where $E=P\times Q$ and the preferences are
\begin{align*}
  &p_1:~ q_2 \succ_{p_1} q_1 && q_1:~ p_1 \succ_{q_1} p_2\\
  &p_2:~ q_1 \succ_{p_2} q_2 && q_2:~ p_2 \succ_{q_2} p_1.
\end{align*}
In this instance, there are two stable matchings $\qopt(I)=\{(p_1,q_1),\,(p_2,q_2)\}$ and $\popt(I)=\{(p_1,q_2),\,(p_2,q_1)\}$. 

By symmetry, without loss of generality, we may assume that the first offer is $(p_1,q_2)$.
Then, there are three possible offering orders $\sigma_1$, $\sigma_2$, and $\sigma_3$:
\begin{align*}
\bigl(\sigma_1(1),\,\sigma_1(2),\,\sigma_1(3),\,\sigma_1(4)\bigr)&=\bigl((p_1,q_2),\,(p_1,q_1),\,(p_2,q_1),\,(p_2,q_2)\bigr),\\
\bigl(\sigma_2(1),\,\sigma_2(2),\,\sigma_2(3),\,\sigma_2(4)\bigr)&=\bigl((p_1,q_2),\,(p_2,q_1),\,(p_1,q_1),\,(p_2,q_2)\bigr),\\
\bigl(\sigma_3(1),\,\sigma_3(2),\,\sigma_3(3),\,\sigma_3(4)\bigr)&=\bigl((p_1,q_2),\,(p_2,q_1),\,(p_2,q_2),\,(p_1,q_1)\bigr).
\end{align*}
It is easy to compute that $\spe(I,\sigma_i)=\qopt(I)$ for $i=1,2,3$ as shown in Fig.~\ref{fig:spe_notallsm}.
Hence, $\SPE(I)=\{\qopt(I)\}$ and $\popt(I)\not\in\SPE(I)$.
In addition, $\SPE(I)\cap\PE^P(I)=\SPE(I)\cap\FCM^P(I)=\emptyset$ because $\PE^P(I)=\FCM^P(I)=\{\popt(I)\}$.
\end{example}

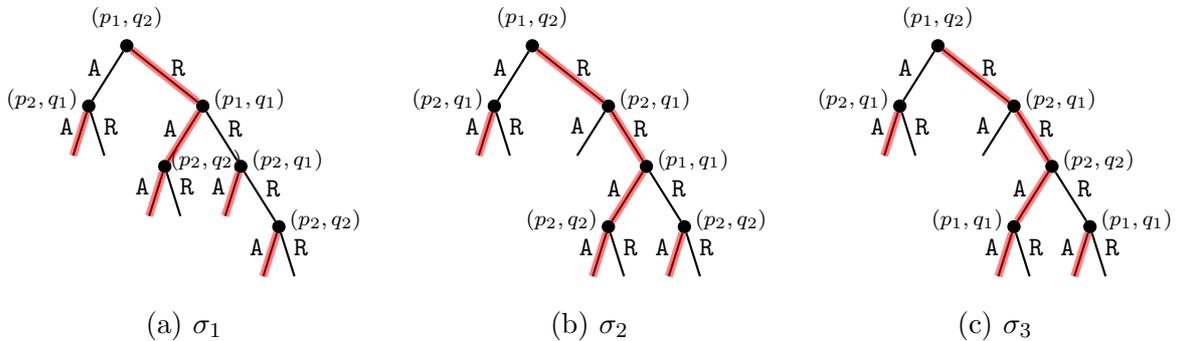
\begin{figure}[htbp]
\begin{minipage}{.33\textwidth}
\centering
\begin{tikzpicture}[thick,xscale=0.25,yscale=0.8,
    p/.style={circle,fill=black,inner sep=0pt,outer sep=0pt,minimum size=5pt}]
  \node[p,label={[]{\scriptsize$(p_1,q_2)$}}] at (3,4)  (1-1) {}; 
  \node[p,label={[left]{\scriptsize$(p_2,q_1)$}}] at (1,3)  (2-1) {}; 
  \node[p,label={[right]{\scriptsize$(p_1,q_1)$}}] at (7,3)  (2-2) {}; 
  \node[] at (0,2)  (3-1) {}; 
  \node[] at (2,2)  (3-2) {}; 
  \node[p,label={[right,xshift=-2pt]{\scriptsize$(p_2,q_2)$}}] at (5,2)  (3-3) {}; 
  \node[p,label={[right]{\scriptsize$(p_2,q_1)$}}] at (9,2)  (3-4) {}; 
  \node[] at (4,1)  (4-1) {}; 
  \node[] at (6,1)  (4-2) {}; 
  \node[] at (8,1)  (4-3) {}; 
  \node[p,label={[right]{\scriptsize$(p_2,q_2)$}}] at (11,1)  (4-4) {}; 
  \node[] at (10,0)  (5-1) {}; 
  \node[] at (12,0)  (5-2) {}; 

  \foreach \u/\v in {1-1/2-1,2-1/3-1,2-2/3-3,3-3/4-1,3-4/4-3,4-4/5-1}
    \draw (\u) -- (\v) node [pos=.5,font=\small,xshift=-5pt,yshift=3pt] {\texttt{A}};
  \foreach \u/\v in {1-1/2-2,2-1/3-2,2-2/3-4,3-3/4-2,3-4/4-4,4-4/5-2}
    \draw (\u) -- (\v) node [pos=.5,font=\small,xshift=5pt,yshift=3pt] {\texttt{R}};
  \foreach \u/\v in {1-1/2-2,2-1/3-1,2-2/3-3,3-3/4-1,3-4/4-3,4-4/5-1}
    \draw[line width=3pt,red,opacity=.5] (\u) -- (\v);  
\end{tikzpicture}
(a) $\sigma_1$
\end{minipage}%
\begin{minipage}{.34\textwidth}
\centering
\begin{tikzpicture}[thick,xscale=0.25,yscale=0.8,
    p/.style={circle,fill=black,inner sep=0pt,outer sep=0pt,minimum size=5pt}]
  \node[p,label={[]{\scriptsize$(p_1,q_2)$}}] at (3,4)  (1-1) {}; 
  \node[p,label={[left]{\scriptsize$(p_2,q_1)$}}] at (1,3)  (2-1) {}; 
  \node[p,label={[right]{\scriptsize$(p_2,q_1)$}}] at (7,3)  (2-2) {}; 
  \node[] at (0,2)  (3-1) {}; 
  \node[] at (2,2)  (3-2) {}; 
  \node[] at (5,2)  (3-3) {}; 
  \node[p,label={[right]{\scriptsize$(p_1,q_1)$}}] at (9,2)  (3-4) {}; 
  \node[p,label={[left]{\scriptsize$(p_2,q_2)$}}] at (7,1)  (4-1) {}; 
  \node[p,label={[right]{\scriptsize$(p_2,q_2)$}}] at (11,1)  (4-2) {}; 
  \node[] at (6,0)  (5-1) {}; 
  \node[] at (8,0)  (5-2) {}; 
  \node[] at (10,0)  (5-3) {}; 
  \node[] at (12,0)  (5-4) {}; 

  \foreach \u/\v in {1-1/2-1,2-1/3-1,2-2/3-3,3-4/4-1,4-1/5-1,4-2/5-3}
    \draw (\u) -- (\v) node [pos=.5,font=\small,xshift=-5pt,yshift=3pt] {\texttt{A}};
  \foreach \u/\v in {1-1/2-2,2-1/3-2,2-2/3-4,3-4/4-2,4-1/5-2,4-2/5-4}
    \draw (\u) -- (\v) node [pos=.5,font=\small,xshift=5pt,yshift=3pt] {\texttt{R}};
  \foreach \u/\v in {1-1/2-2,2-1/3-1,2-2/3-4,3-4/4-1,4-1/5-1,4-2/5-3}
    \draw[line width=3pt,red,opacity=.5] (\u) -- (\v);  
\end{tikzpicture}
(b) $\sigma_2$
\end{minipage}%
\begin{minipage}{.33\textwidth}
\centering
\begin{tikzpicture}[thick,xscale=0.25,yscale=0.8,
    p/.style={circle,fill=black,inner sep=0pt,outer sep=0pt,minimum size=5pt}]
  \node[p,label={[]{\scriptsize$(p_1,q_2)$}}] at (3,4)  (1-1) {}; 
  \node[p,label={[left]{\scriptsize$(p_2,q_1)$}}] at (1,3)  (2-1) {}; 
  \node[p,label={[right]{\scriptsize$(p_2,q_1)$}}] at (7,3)  (2-2) {}; 
  \node[] at (0,2)  (3-1) {}; 
  \node[] at (2,2)  (3-2) {}; 
  \node[] at (5,2)  (3-3) {}; 
  \node[p,label={[right]{\scriptsize$(p_2,q_2)$}}] at (9,2)  (3-4) {}; 
  \node[p,label={[left]{\scriptsize$(p_1,q_1)$}}] at (7,1)  (4-1) {}; 
  \node[p,label={[right]{\scriptsize$(p_1,q_1)$}}] at (11,1)  (4-2) {}; 
  \node[] at (6,0)  (5-1) {}; 
  \node[] at (8,0)  (5-2) {}; 
  \node[] at (10,0)  (5-3) {}; 
  \node[] at (12,0)  (5-4) {}; 

  \foreach \u/\v in {1-1/2-1,2-1/3-1,2-2/3-3,3-4/4-1,4-1/5-1,4-2/5-3}
    \draw (\u) -- (\v) node [pos=.5,font=\small,xshift=-5pt,yshift=3pt] {\texttt{A}};
  \foreach \u/\v in {1-1/2-2,2-1/3-2,2-2/3-4,3-4/4-2,4-1/5-2,4-2/5-4}
    \draw (\u) -- (\v) node [pos=.5,font=\small,xshift=5pt,yshift=3pt] {\texttt{R}};
  \foreach \u/\v in {1-1/2-2,2-1/3-1,2-2/3-4,3-4/4-1,4-1/5-1,4-2/5-3}
    \draw[line width=3pt,red,opacity=.5] (\u) -- (\v);  
\end{tikzpicture}
(c) $\sigma_3$
\end{minipage}
\caption{The tree representation of the games in Example~\ref{ex:spe_notallsm} with offering orders $\sigma_1$, $\sigma_2$, and $\sigma_3$. The bold red edges indicate the SPEs.}\label{fig:spe_notallsm}
\end{figure}

\begin{example}[Proof of the worker-side parts of Propositions~\ref{prop:nonPareto} and~\ref{prop:nonFCM}]\label{ex:nonPareto}
Let us consider a matching instance $I=(P,Q,E,\succ)$ with two firms $P=\{p_1,p_2\}$ and three workers $Q=\{q_1,q_2,q_3\}$
where $E=P\times Q$ and the preferences are
\begin{align*}
  &p_1:~ q_1 \succ_{p_1} q_3 \succ_{p_1} q_2 && q_1:~ p_2 \succ_{q_1} p_1\\
  &p_2:~ q_2 \succ_{p_2} q_3 \succ_{p_2} q_1 && q_2:~ p_1 \succ_{q_2} p_2\\
  &                                          && q_3:~ p_1 \succ_{q_3} p_2.
\end{align*}
We claim that $\SPE(I)=\bigl\{\{(p_1,q_1),\,(p_2,q_2)\}\bigr\} ~(=\{\qopt(I)\})$ by contradiction.
Suppose that $\spe(I,\sigma^*)\ne\{(p_1,q_1),\,(p_2,q_2)\}$ for some $\sigma^*\in\Sigma_I$.
Then, at least one of the offers $(p_1,q_1)$ and $(p_2,q_2)$ is rejected on the path according to the SPE of the game $(I,\sigma^*)$.
If $q_1$ chooses \REJECT for the offer $(p_1,q_1)$, then she must match with a firm better than $p_1$ (i.e., $p_2$) in the SPE matching.
This implies that $q_2$ also chooses \REJECT for the offer $(p_2, q_2)$ on the SPE path, and hence $q_2$ must match with $p_1$ in the SPE matching.
However, this is impossible because $q_3$ would choose \ACCEPT for the offer $(p_1,q_3)$.
By similar arguments, we can obtain a contradiction for the case when $q_2$ chooses \REJECT for the offer $(p_2,q_2)$.
Hence, $\SPE(I)=\bigl\{\{(p_1,q_1),\,(p_2,q_2)\}\bigr\}$.

Here, we have
\begin{align*}
\PE^Q(I)&=\bigl\{\{(p_1,q_2),\,(p_2,q_1)\},\,\{(p_1,q_3),\,(p_2,q_1)\},\,\{(p_1,q_2),\,(p_2,q_3)\},\,\{(p_1,q_3),\,(p_2,q_2)\}\bigr\},\\
\FCM^Q(I)&=\bigl\{\{(p_1,q_2),\,(p_2,q_1)\},\,\{(p_1,q_3),\,(p_2,q_1)\}\bigr\}.
\end{align*}
Thus, we obtain $\SPE(I)\cap\PE^Q(I)=\SPE(I)\cap\FCM^Q(I)=\emptyset$.
\end{example}

It should be noted that, by combining Examples~\ref{ex:spe_notallsm} and~\ref{ex:nonPareto},
one can construct a matching instance $I$ such that $\SPE(I)\cap(\{\popt(I)\}\cup \PE^P(I)\cup \PE^Q(I)\cup \FCM^P(I)\cup \FCM^Q(I))=\emptyset$.

\subsection*{Acknowledgments}
The authors thank anonymous reviewers of EC'18 for their valuable comments.
This work was supported 
by JSPS KAKENHI Grant Numbers JP16H06931, JP16K16005, and JP18K18004,
by JST ACT-I Grant Number JPMJPR17U7,
and by JST CREST Grant Number JPMJCR14D2.

\bibliographystyle{plain}
\bibliography{sequential_matching}

\begin{thebibliography}{10}

\bibitem{AS2003}
A.~Abdulkadiro$\breve{\rm g}$lu and T.~S{\"o}nmez.
\newblock School choice: A mechanism design approach.
\newblock {\em American Economic Review}, 93(3):729--747, 2003.

\bibitem{APR2009}
A.~Abdulkadiro{\u{g}}lu, P.~A. Pathak, and A.~E. Roth.
\newblock Strategy-proofness versus efficiency in matching with indifferences:
  Redesigning the {NYC} high school match.
\newblock {\em American Economic Review}, 99(5):1954--78, 2009.

\bibitem{APR1998}
J.~Alcalde, D.~J. P{\'e}rez-Castrillo, and A.~Romero-Medina.
\newblock Hiring procedures to implement stable allocations.
\newblock {\em J.\ Econom.\ Theory}, 82(2):469--480, 1998.

\bibitem{AR2000}
J.~Alcalde and A.~Romero-Medina.
\newblock Simple mechanisms to implement the core of college admissions
  problems.
\newblock {\em Games Econom.\ Behav.}, 31(2):294--302, 2000.

\bibitem{AR2005}
J.~Alcalde and A.~Romero-Medina.
\newblock Sequential decisions in the college admissions problem.
\newblock {\em Econom.\ Lett.}, 86(2):153--158, 2005.

\bibitem{AHK2016}
G.~Avni, T.~A. Henzinger, and O.~Kupferman.
\newblock Dynamic resource allocation games.
\newblock In {\em International Symposium on Algorithmic Game Theory}, pages
  153--166, 2016.

\bibitem{Bando2014}
K.~Bando.
\newblock On the existence of a strictly strong {N}ash equilibrium under the
  student-optimal deferred acceptance algorithm.
\newblock {\em Games Econom.\ Behav.}, 87:269--287, 2014.

\bibitem{DMS2018}
U.~Dur, T.~Mennle, and S.~Seuken.
\newblock First choice maximizing school choice mechanisms.
\newblock In {\em Proceedings of the 19th ACM Conference on Economics and
  Computation}, pages 251--268, 2018.

\bibitem{Eeckhout2000}
J.~Eeckhout.
\newblock On the uniqueness of stable marriage matchings.
\newblock {\em Econom.\ Lett.}, 69(1):1--8, 2000.

\bibitem{Ehlers2007}
L.~Ehlers.
\newblock Von {N}eumann--{M}orgenstern stable sets in matching problems.
\newblock {\em J.\ Econom.\ Theory}, 134(1):537--547, 2007.

\bibitem{GS62}
D.~Gale and L.~S. Shapley.
\newblock College admissions and the stability of marriage.
\newblock {\em Amer.\ Math.\ Monthly}, 69(1):9--15, 1962.

\bibitem{GS1985}
D.~Gale and M.~Sotomayor.
\newblock Some remarks on the stable matching problem.
\newblock {\em Discrete Appl.\ Math.}, 11(3):223--232, 1985.

\bibitem{garey1979cai}
M.~R. Garey and D.~S. Johnson.
\newblock {\em Computers and Intractability: A Guide to the Theory of
  {NP}-Completeness}.
\newblock Freeman New York, New York, 1979.

\bibitem{GI1989}
D.~Gusfield and R.~W. Irving.
\newblock {\em The Stable Marriage Problem: Structure and Algorithms}.
\newblock MIT Press, Boston, 1989.

\bibitem{HW2011}
G.~Haeringer and M.~Wooders.
\newblock Decentralized job matching.
\newblock {\em Internat.\ J.\ of Game Theory}, 40(1):1--28, 2011.

\bibitem{KB2019}
Y.~Kawase and K.~Bando.
\newblock Subgame perfect equilibria under the deferred acceptance algorithm.
\newblock {\em SSRN:3235068}, 2019.

\bibitem{KYY2018}
Y.~Kawase, Y.~Yamaguchi, and Y.~Yokoi.
\newblock Computing a subgame perfect equilibrium of a sequential matching
  game.
\newblock In {\em Proceedings of the 19th ACM Conference on Economics and
  Computation}, pages 131--148, 2018.

\bibitem{Kesten2010}
O.~Kesten.
\newblock School choice with consent.
\newblock {\em The Quarterly Journal of Economics}, 125(3):1297--1348, 2010.

\bibitem{KT2016}
A.~Kloosterman and P.~Troyan.
\newblock Efficient and essentially stable assignments.
\newblock Mimeo, 2016.

\bibitem{knuth1976}
D.~E. Knuth.
\newblock {\em Marriage Stables}.
\newblock Montr\'{e}al: Les Presses de l'Universit\'{e} de Montr\'{e}al, 1976.

\bibitem{LST2012}
R.~P. Leme, V.~Syrgkanis, and {\'E}.~Tardos.
\newblock The curse of simultaneity.
\newblock In {\em Proceedings of the 3rd Innovations in Theoretical Computer
  Science Conference}, pages 60--67, 2012.

\bibitem{manlove2013}
D.~F. Manlove.
\newblock {\em Algorithmics of Matching under Preferences}.
\newblock World Scientific, 2013.

\bibitem{MW70}
D.~G. McVitie and L.~B. Wilson.
\newblock Stable marriage assignment for unequal sets.
\newblock {\em BIT}, 10(3):295--309, 1970.

\bibitem{MW71}
D.~G. McVitie and L.~B. Wilson.
\newblock The stable marriage problem.
\newblock {\em Communications of the ACM}, 14(7):486--490, 1971.

\bibitem{NRTV07}
N.~Nisan, T.~Roughgarden, {\'E}.~Tardos, and V.~V. Vazirani.
\newblock {\em Algorithmic Game Theory}.
\newblock Cambridge University Press, Cambridge, 2007.

\bibitem{Pais2008}
J.~Pais.
\newblock Incentives in decentralized random matching markets.
\newblock {\em Games Econom.\ Behav.}, 64(2):632--649, 2008.

\bibitem{Roth1984}
A.~E. Roth.
\newblock The evolution of the labor market for medical interns and residents:
  a case study in game theory.
\newblock {\em Journal of Political Economy}, 92(6):991--1016, 1984.

\bibitem{Roth1986}
A.~E. Roth.
\newblock On the allocation of residents to rural hospitals: a general property
  of two-sided matching markets.
\newblock {\em Econometrica}, 54(2):425--427, 1986.

\bibitem{RS1991}
A.~E. Roth and M.~Sotomayor.
\newblock {\em Two-sided Matching: A Study in Game-Theoretic Modeling and
  Analysis}.
\newblock Cambridge University Press, Cambridge, 1991.

\bibitem{RV1991}
A.~E. Roth and J.~H. Vande~Vate.
\newblock Incentives in two-sided matching with random stable mechanisms.
\newblock {\em Econom.\ Theory}, 1(1):31--44, 1991.

\bibitem{SW2008}
S.~Suh and Q.~Wen.
\newblock Subgame perfect implementation of stable matchings in marriage
  problems.
\newblock {\em Soc.\ Choice Welf.}, 31(1):163--174, 2008.

\bibitem{TY2014}
Q.~Tang and J.~Yu.
\newblock A new perspective on {K}esten's school choice with consent idea.
\newblock {\em J.\ Econom.\ Theory}, 154:543--561, 2014.

\bibitem{Wako2010}
J.~Wako.
\newblock A polynomial-time algorithm to find von {N}eumann-{M}orgenstern
  stable matchings in marriage games.
\newblock {\em Algorithmica}, 58(1):188--220, 2010.

\end{thebibliography}

\end{document}